\documentclass[11pt]{amsart}
\usepackage{amssymb}

\usepackage[a4paper,width=185mm,top=20mm,bottom=20mm,bindingoffset=6mm]{geometry}
\usepackage{etex}
\usepackage[all]{xy} 
\usepackage{etoolbox} 
\usepackage{lmodern}
\usepackage{array}
\usepackage{xstring}
\usepackage{longtable}
\usepackage[listings]{tcolorbox} 
\usepackage{natbib}
\usepackage{hyperref}
\tcbuselibrary{breakable}
\tcbuselibrary{skins}
\usepackage{amsmath}
\usetikzlibrary{arrows.meta, positioning, shapes.geometric}
\usepackage{fancyhdr}
\usepackage{amsmath}
\usepackage{amssymb}
\usepackage{amsfonts}
\usepackage{pifont}
\usepackage{natbib}
\usepackage{epsfig}
\usepackage{latexsym}
\usepackage{amssymb}
\usepackage{mathtools}
\usepackage{color}
\usepackage{amscd}
\usepackage{multirow}
\usepackage{graphicx}
\usepackage{lscape}
 \usepackage{tabularx}
\usepackage{float}
\usepackage{xcolor}
\usepackage{bm}
\usepackage{tikz-cd}
 \makeatletter
    \let\stdchapter\section
    \renewcommand*\section{%
    \@ifstar{\starchapter}{\@dblarg\nostarchapter}}
    \newcommand*\starchapter[1]{%
        \stdchapter*{#1}
        \thispagestyle{fancy}
        \markboth{\MakeUppercase{#1}}{}
    }
    \def\nostarchapter[#1]#2{%
        \stdchapter[{#1}]{#2}
        \thispagestyle{fancy}
    }
\makeatother

 \makeatletter
\@addtoreset{equation}{section} 
\makeatother

\newtcolorbox{boxA}{
    fontupper = \bf,
    boxrule = 1.5pt,
    colframe = black 
}

\newtheorem{theorem}{Theorem}[section]
\newtheorem*{theorem*}{Main Theorem}
\newtheorem{lemma}[theorem]{Lemma}
\newtheorem{proposition}[theorem]{Proposition}
\theoremstyle{definition}
\newtheorem{definition}[theorem]{Definition}
\theoremstyle{plain}
\newtheorem{corollary}[theorem]{Corollary}
\newtheorem{remark}[theorem]{Remark}
\newtheorem{example}[theorem]{Example}
\theoremstyle{conclusion}

\usepackage{amsmath}
\usepackage{amssymb}
\usepackage{amsfonts}
\usepackage{esvect}
\usepackage{physics}
\usepackage{mathrsfs}
\usepackage{authblk}
\usepackage{geometry}
\usepackage{abstract}
\usepackage{xcolor}

\begin{document}

\thispagestyle{empty}
\begin{center}
	\Large{{\bf  Poisson Centralizers and Polynomial Superintegrability for Magnetic Geodesic Flows on Reductive Homogeneous Spaces }}
\end{center}
 
\vskip 0.3cm
\begin{center}
	\textsc{Kai Jiang$^{1\star}$,  Guorui Ma$^{2a,2b,\sharp}$,    Ian Marquette$^{3,\bullet}$,   Junze Zhang$^{4,\dagger}$  and Yao-Zhong Zhang$^{4,\ddagger}$}
\end{center}
\vskip 0.2cm
\begin{center}
	$^1$   Paris Curie Engineer School, Beijing University of Chemical Technology, Beijing, China
\end{center}
\begin{center}
	$2^a$ Shanghai Institute for Mathematics and Interdisciplinary Sciences (SIMIS), Shanghai, 200433, China  \\
    $2^b$  Research institute of Intelligent Complex Systems, Fudan University, Shanghai, 200433, China 
\end{center}
\begin{center}
	$^3$ Department of Mathematical and Physical Sciences, La Trobe University, Bendigo, VIC 3552, Australia
\end{center}

\begin{center}
	$^4$ School of Mathematics and Physics, The University of Queensland, Brisbane, QLD 4072, Australia
\end{center}

\begin{center}
	\footnotesize{ $\star$\textsf{kai.jiang.math@gmail.com} \hskip 0.25cm
    $\bullet$\textsf{i.marquette@latrobe.edu.au} \hskip 0.25cm
    $\sharp$\textsf{mgr18@tsinghua.org.cn} \hskip 0.25cm
 $^\dagger$\textsf{junze.zhang@uq.net.au} \hskip 0.25cm
 $^\ddagger$\textsf{yzz@maths.uq.edu.au}}
\end{center}
\vskip  1cm

\begin{abstract}
\noindent  We apply the Poisson projection chain to formulate superintegrable 
magnetic geodesic flows on an adjoint orbit $M=G/A$ with $G$ a compact semisimple Lie group and $A$ a closed subgroup of $G$. In the twisted cotangent bundle $(T^*M,\omega_\varepsilon)$, with $\omega_\varepsilon=\omega_{\mathrm{can}}+\varepsilon\,\pi^*\omega_{\mathrm{KKS}}$ being the canonical plus Kirillov-Kostant-Souriau (KKS) forms, we build two canonical and commuting families of polynomial first integrals: one pulled back from the Lie algebra $\mathfrak{g}$ of $G$ via the magnetic moment map $P$, and one pulled back from a $\mathrm{Ad}(A)$-invariant affine slice of $\mathfrak{m} \cong T_{eA}M$, where $eA$ is the identity of $G/A$. Their common image generates a reduced Poisson algebra obtained from a fiber tensor product, and the natural multiplication map into a Poisson subalgebra of polynomial functions $\mathcal{O}(T^*M) \subset C^\infty(T^*M)$ is Poisson and injective. The center of this fiber tensor product is contained in the Poisson center of the symmetric algebra of $\mathfrak{g}$. 
 In a dense regular locus, 
the resulting projection chain realises a superintegrable system.  
As examples, two $\mathrm{SU}(3)$ cases are studied (regular torus and irregular $\mathrm{S}(\mathrm{U}(2)\times \mathrm{U}(1))$ quotients), which illustrate the construction and produce explicit action-angle coordinates. 
\end{abstract}
\vskip 0.35cm
\hrule

\tableofcontents

\section{Introduction}
\label{sec:introduction} 

Integrable and superintegrable systems play an important role and have broad applications in mathematical physics and symplectic geometry \cite{MR3119484,MR3493688,MR187763,nehorovsev1968action,MR3942135,MR4071113,MR4282982,MR4644061}. A classical Hamiltonian system with $n$ degrees of freedom is said to be completely integrable if it admits $n$ independent first integrals (a.k.a integrals of motion, conserved quantities) in involution, and \textit{superintegrable} if it possesses more than $n$ integrals of motion (typically up to $2n-1$ integrals in a $2n$-dimensional phase space). In particular, a maximally superintegrable system has $2n-1$ independent integrals, which implies that all trajectories are closed  \cite{MR3493688,MR3942135,MR2023556}. 
Based on the extension of the integrals, superintegrability involves a range of other cases (based on the integrable case), from minimally superintegrable ($n+1$) to quasi maximally superintegrable ($2n-2$). These cases appear in the context of the study of superintegrable systems, usually using explicit differential operator realizations in classical and quantum mechanics. The modern study of integrable Hamiltonian systems was initiated by Liouville and is characterised by the Liouville-Arnold theorem.  The notion of noncommutative integrability, as formalized by Nekhoroshev \cite{nehorovsev1968action}, extends the Liouville-Arnold theorem to the noncommutative case and provides geometric descriptions of how superintegrabilities are realized in symplectic manifolds.  A completely integrable system consists of $n$ functionally independent conserved quantities on a $2n$-dimensional symplectic manifold and yields action-angle coordinates \cite{nehorovsev1968action,MR997295}. Moreover, in the 1970s, Mishchenko and Fomenko introduced the argument shift method in the study of Euler equations on Lie algebras, constructing Poisson-commutative algebras of integrals from invariant polynomials \cite{Mishchenkofomenko1978}. Their pioneering work established the complete integrability (in a generalized, noncommutative sense) of the geodesic flow in compact symmetric spaces by exhibiting sufficiently many conserved quantities. Independently, Thimm developed a complementary method based on chains of subalgebras and proved that invariant geodesic flows in certain homogeneous spaces are completely integrable \cite{Thimm1981}. 

Recently, additional developments have advanced in the construction of superintegrable systems using both algebraic and geometric approaches. For instance, the Gelfand-Cetlin system, introduced by Guillemin and Sternberg \cite{GuilleminSternberg1983}, provides a concrete family of conserved quantities on coadjoint orbits (e.g., flag manifolds) of unitary groups, thereby demonstrating integrability in those cases. These works illustrate how algebraic and geometric techniques (such as exploiting Lie algebra invariants or moment maps along a chain of subgroups) can lead to obtaining explicit first integrals for geodesic flows. 
These first integrals are usually polynomials in terms of the basis of Lie algebras and need not all Poisson-commute, provided that they generate a commutative subalgebra of sufficiently large dimension. This broader integrability context has been employed in different frameworks. For example, Efimov \cite{MR2141306} established the noncommutative integrability of certain magnetic geodesic flows on homogeneous symplectic manifolds, and Reshetikhin \cite{MR3493688} used the superintegrable Poisson chain $M \to \mathcal{P} \to \mathcal{B}$ (See Definition \ref{def:superge}) to describe the superintegrability of a symplectic manifold $M$. More recently, novel superintegrable systems have been developed on moduli spaces and in other Lie-theoretic contexts \cite{MR4299124}, demonstrating a growing interest in integrable structures associated with symmetric and homogeneous spaces.

In this paper, we apply Poisson projection chains to formulate superintegrabilities on the cotangent bundle of compact homogeneous manifolds with a twisted symplectic form. We focus on reductive homogeneous spaces $M= G/A$, where $G$ is a compact, simply connected, and semisimple Lie group, and $A$ is a closed and connected subgroup. We work on twisted cotangent bundles where the symplectic form is the sum of the canonical form on $T^*M$ and a magnetic term coming from the Kirillov-Kostant-Souriau (KKS) form. The noncommutative integrability of these magnetic geodesic flows was established by Efimov \cite{MR2141306} and later extended by Bolsinov-Jovanovi{\'c} \cite{MR2410783}. Their approach is primarily dynamical: one starts from two  collections of first integrals coming from moment map type of constructions and proves integrability by analyzing Hamiltonian vector fields and their commutation properties.

In this paper, we treat the same collections of polynomial integrals as algebraic objects. In particular, \begin{itemize}
    \item this perspective allows us to formulate a natural equivalence relation between (a priori different) Poisson projection chains arising from the different physical models, such as the spin Calogero-Moser model in \cite{MR4285685}. We will exploit this equivalence and its consequences in forthcoming work in \cite{Jiang2026reduced,Jiang2026reducedi};
    \item compared with existing constructions that establish (super)integrability by direct dynamical arguments and vector field computations, our approach provides a new functorial algebraic framework: we package the integrals into Poisson algebras and recover the associated Poisson projection chains canonically by passing to spectra;
    \item our main new ingredient is the identification of a canonical central Poisson subalgebra $R_0$ that captures exactly the overlap of the two families of integrals. Working relative to this common base lets us describe the rank stratification and prove superintegrability in a uniform, purely algebraic way.
\end{itemize}
These innovations are achieved by explicitly constructing the integrals and their intersection. We organize the two blocks of polynomial integrals into Poisson algebras $\mathfrak{F}_1^{\mathrm{poly}}$ and $\mathfrak{F}_2^{\mathrm{poly}}$, and we record their overlap as an explicit central Poisson subalgebra $R_0:=\mathfrak{F}_1^{\mathrm{poly}}\cap \mathfrak{F}_2^{\mathrm{poly}}$. The Poisson algebra generated by all polynomial integrals is then canonically identified with the coordinate ring of a (reduced) fiber product over this common base,
\begin{align*}
 \mathcal{A} \cong \mathfrak{F}_1^{\mathrm{poly}} \otimes_{R_0} \mathfrak{F}_2^{\mathrm{poly}}.
\end{align*}
This identification makes the passage from integrals to geometry functorial: taking spectra produces a canonical Poisson projection chain through $\mathrm{Spec}\,\mathcal{A}$, and it provides a natural language for comparing different Poisson projection chains through their associated algebras of polynomial integrals. 
Consequently, the rank decomposition and the associated superintegrability on the cotangent bundle of homogeneous spaces are encoded by morphisms of affine Poisson varieties, rather than obtained by case by case dynamical vector field computations. This algebraic packaging is the main novelty of the present paper and is what makes the construction uniform.

The structure of the paper is as follows: we begin in Section \ref{sec:prel} by reviewing the necessary background on superintegrable systems in symplectic manifolds and setting up the notation. Key definitions and classical results (e.g., Superintegrable Poisson chains, geometric invariant theory quotient in the affine case) are summarized to provide a foundation for our work.  In Section \ref{sec:superconstruction}, we present a uniform Poisson projection chain construction of superintegrable magnetic geodesic flows on compact homogeneous symplectic manifolds $M=G/A$ with $A = \exp (\ker \mathrm{ad}(W))$ for a fixed $W \in \mathfrak{g}$. Working on the twisted cotangent bundle $(T^*M,\omega_\varepsilon)$ with  $\omega_\varepsilon=\omega_{\mathrm{can}}+\varepsilon\,\pi^*\omega_{\mathrm{KKS}}$, we produce two canonical polynomial families of the first integrals: \begin{align*}
    \mathfrak{F}_1^{\mathrm{poly}}:=P^*S(\mathfrak{g}),\qquad \mathfrak{F}_2^{\mathrm{poly}}:=\pi_{\mathfrak{m}}^*S(\mathfrak{m}-\varepsilon W)^A,
\end{align*} where $P(g,X)=\mathrm{Ad}(g)(X-\varepsilon W)$ is the magnetic moment map, $\widetilde{\pi}_{\mathfrak{m}}:G\times\mathfrak{m}\to\mathfrak{m}-\varepsilon W$, $\widetilde{\pi}_{\mathfrak{m}}(g,X)=X-\varepsilon W$, is only the representative-level affine fiber map, and $\pi_{\mathfrak{m}}^*$ denotes the descended invariant pullback on $S(\mathfrak{m}-\varepsilon W)^A$. Here $(g,X)$ is an arbitrary point in $T^*M$ under the left-trivialisation $T^*(G/A) \cong (\mathfrak{m} \times G)/A$. The mixed brackets vanish, $P^*$ is Poisson and $\pi_{\mathfrak{m}}^*$ is anti-Poisson, so all bracket computations reduce to Lie-Poisson brackets on $\mathfrak{g}$ and on the $A$-invariant affine slice. The Poisson commutative base algebra \begin{align*}
    R_0:= \mathfrak{F}_1^{\mathrm{poly}}\cap \mathfrak{F}_2^{\mathrm{poly}} = P^* \big(\mathrm{Im}(\mathrm{Res}_W)\big)
\end{align*}  controls the intersection of two algebras, where $\mathrm{Res}_W: S(\mathfrak{g})^G \rightarrow S(\mathfrak{m})^A$ is a Poisson algebra homomorphism. We then show the fiber tensor product $\mathfrak{F}_1^{\mathrm{poly}}\otimes_{R_0}\mathfrak{F}_2^{\mathrm{poly}} \cong \textbf{Alg} \left\langle F_1^{\mathrm{poly}} \cup F_2^\mathrm{poly} \right\rangle = \mathcal{A} $, which is a reduced Poisson algebra, and obtain a natural Poisson projection chain in the following theorem: \begin{theorem*}
(Theorem \ref{thm:superin}) 
   Let $G$ be a compact, simply-connected, semisimple Lie group, and let $A= \exp(\ker\mathrm{ad}(W)) $ $\supseteq T$. Suppose that $(T^*M,\omega_\varepsilon)$ is the twisted symplectic manifold. On the Zariski open dense regular locus $U \subset T^*M$, we then have the following arguments:

(i) If $W$ is irregular and $T \subsetneq A$. Then the system of the Poisson chain \begin{align}
    T^*M\xrightarrow{\pi_1}\mathrm{Ad}(G)(\mathfrak{m} - \varepsilon W)\times_{R_0}(\mathfrak{m} - \varepsilon W)//A \xrightarrow{\pi_2}\mathrm{Spec} R_0  \label{eq:irr0}
\end{align} is superintegrable.

(ii) If $W$ is regular, then $R_0  = P^*\left(S(\mathfrak{g})^G\right)$ and \begin{align}
    T^*M\xrightarrow{\pi_1}  \mathfrak{g}  \times_{\mathfrak{g}//G} (\mathfrak{m} - \varepsilon W)//T \xrightarrow{\pi_2}\mathfrak{g}//G  \label{eq:re1}
\end{align} is superintegrable.
\end{theorem*} 
Then in Section \ref{sec:examples}, we focus on concrete examples: we consider the reductive homogeneous spaces $ \mathrm{SU}(3)/T $ and $\mathrm{SU}(3)/\mathrm{S}(\mathrm{U} (2) \times \mathrm{U}(1))$ and construct an explicit superintegrable system on these spaces. In particular, we compute a complete family of independent integrals in involution for the magnetic geodesic flow on $\mathrm{SU}(3)/T $, illustrating our algebraic construction in a low-dimensional but nontrivial setting. We also examine the role of regular vs. irregular elements in the (super)integrability of these $\mathrm{SU}(3)$ examples. Finally, Section \ref{sec:conclusion} provides some concluding remarks and suggests directions for future research.


\section{Preliminary}
\label{sec:prel}
In Section \ref{sec:prel}, we recall the definition of a superintegrable Poisson chain and fix the notation and basic properties that will be used throughout the paper. We also review the geometric invariant theory (GIT) quotient in the affine setting, emphasizing the functorial viewpoint needed to step from Poisson algebras of integrals to morphisms of affine Poisson varieties.

We begin with some basic terminology from symplectic geometry.  For more detailed information, we refer the reader to \cite{MR2906391}. Let $(M,\omega)$ be a $2n$-dimensional smooth manifold equipped with a symplectic form $\omega$. For any $f$ in the algebra $C^\infty(M)$ of smooth functions on $(M,\omega)$, there exists a unique Hamiltonian vector field $X_f$ defined by  \begin{align*}
 \iota_{X_f} \omega := \omega(X_f,-) = d f .
\end{align*}

Throughout this paper, the Hamiltonian vector fields are assumed to be complete.

For any $f,g \in C^\infty(M)$,  the natural Poisson bracket $\{\cdot,\cdot\}_\omega : C^\infty(M) \times C^\infty(M) \rightarrow C^\infty (M)$ on $M$ is defined via their associated Hamiltonian vector fields by 
\begin{align*}
\{f,g\}_\omega = \omega(X_f,X_g).
\end{align*}   

We call a smooth function $g$ a \textit{first integral} (a.k.a, integral of motion, conserved quantity, constant of motion) of a given smooth function $f$ if the two functions Poisson commute, i.e., $\{f,g\}_\omega=0$.

From Hamiltonian mechanics, first integrals of $f$ are invariant along the integral curves of $X_f$, i.e., $X_f(g)= \{f,g\} =0$ for any first integral $g$. On the other hand, the first integrals provide the symmetries of the Hamiltonian flow (geodesic flow) generated by $f$.


Let $\mathfrak{g}$ be a finite-dimensional Lie algebra over $\mathbb{R}$ with a non-trivial commutator $[\cdot,\cdot]$, and let $\mathfrak{g}^*$ be its dual, which admits the classical Lie-Poisson bracket: for any $f,g \in C^\infty(\mathfrak{g}^*)$ and $\xi \in \mathfrak{g}^*$, 
\begin{align}
    \{f,g\} (\xi)  = \langle \xi, [d_\xi f,d_\xi g]  \rangle ,\label{eq:poiss}
\end{align} 
where $d_\xi f, d_\xi g \in  T_\xi^*\mathfrak{g}^*\cong \left(\mathfrak{g}^*\right)^* \cong \mathfrak{g}$, and $\langle \cdot,\cdot \rangle : \mathfrak{g}^*\times \mathfrak{g} \rightarrow \mathbb{R}$ is the dual pair between $\mathfrak{g}^*$ and $\mathfrak{g}$. We refer the reader to  \cite[Chapter 7]{MR2906391} for details.




Suppose that

(i) the Poisson algebra $C^\infty(M)$  has a Poisson subalgebra $\mathcal{A} \subset C^\infty(M)$ of rank (transcendence degree) $2n-k$ for some integer $k$ with $1\leqslant k\leqslant n$;

(ii) the Poisson center $\mathcal{Z}(\mathcal{A})$ of $ \mathcal{A}$ is of rank $k$.

Smooth functions in $\mathcal{Z}(\mathcal{A})$ are called Hamiltonian, which are smooth functions Poisson-commuting with first integrals. 



\subsection{Superintegrable Poisson chains}

We now provide a geometric definition of a Hamiltonian dynamical system as being superintegrable. For more details, the reader can refer to \cite{MR3493688}. 



\begin{definition} 
\label{def:superge} 
Let $(M,\omega)$ be a $2n$-dimensional connected symplectic manifold. A \textit{superintegrable system} on $(M,\omega)$ is a sequence consisting of a triple $(M,\mathcal{P}_{2n-k},\mathcal{B}_k)$ of Poisson manifolds with two Poisson maps $\pi_1$ and $\pi_2$ \begin{align}
M \xrightarrow{\;\pi_1\;} \mathcal{P}_{2n-k} \xrightarrow{\;\pi_2\;} \mathcal{B}_k \label{eq:sup}
\end{align} such that the following holds:

(i) The Poisson maps  $\pi_i$ are submersions with connected fibers of constant (generic) dimensions $k$ and $2n-2k$, respectively, $1\leq k\leq n$ 


(ii) $\dim \mathcal{P}_{2n-k} +\dim \mathcal{B}_k=2n$. 
\end{definition}
 \begin{remark}
 \label{rkm:ijc}
        In addition to the structures introduced in Definition \ref{def:superge}, we define
\begin{align*}
J &:= \pi_1^* C^\infty(\mathcal{P}_{2n-k}) \ \subset\ C^\infty(M), \qquad 
I := (\pi_2 \circ \pi_1)^* C^\infty(\mathcal{B}_k) \ \subset\ C^\infty(M).
\end{align*}
That is, $J$ is a Poisson algebra of rank $2n -k$ and consists of all smooth functions on $M$ obtained as pullbacks of smooth functions on $\mathcal{P}_{2n-k}$ along $\pi_1$, while $I$ is a Poisson algebra of rank $k$ formed by pullbacks of smooth functions on $\mathcal{B}_k$ via the composition $\pi_2 \circ \pi_1$.
We define the Poisson centralizers of these integrals in $C^\infty(M)$ by
\begin{align*}
C_I(M) := \bigl\{f\in C^\infty(M) : \{f,\phi\}=0\ \text{ for all }\phi\in I\bigr\},\\
C_J(M) := \bigl\{f\in C^\infty(M):\{f,\psi\}=0\ \text{ for all }\psi\in J\bigr\}.
\end{align*} Since $\pi_1$ is Poisson and every component of $\pi_2$ is a Casimir on $\mathcal{P}_{2n-k}$, we have the inclusions as follows: \begin{align*}
J \subset C_I(M) \quad\text{and}\quad I \subset C_J(M).
\end{align*} Indeed, for $h\in C^\infty(\mathcal{P}_{2n-k})$ and $g\in C^\infty(\mathcal{B}_k)$,\begin{align*}
\bigl\{\pi_1^*h,(\pi_2 \circ \pi_1)^*g\bigr\}_M=\pi_1^*\bigl\{h,\pi_2^*g\bigr\}_{\mathcal{P}_{2n-k}}=0.
\end{align*}
Consequently, $I$ and $J$ form the two mutually Poisson-commuting blocks of a superintegrable system in the standard sense: the integrals in $J$ (called integrals of motion) specify the first projection $\pi_1$, while the integral $I$ (called Hamiltonian) further refines the symplectic foliation of $\mathcal{P}_{2n-k}$, decomposing it into the connected fibers associated with the second projection $\pi_2$.
 \end{remark}


\begin{remark}
\label{rem:fibers}
We now look at the regular fibers of $\pi_1$ and $\pi_2$.

(i) First of all, we describe the generic\footnote{Here and throughout,  \textit{generic} means for $c$ in the regular value set of $\pi_1$ (often an open dense subset), such that the corresponding fiber has the (non-singular) structure.} fibers in $M$.  Let $c \in \mathcal{P}_{2n-k}$ be a regular value of $\pi_1$, and choose a local coordinate $f= (f_1,\ldots,f_{2n-k})$ on $\mathcal{P}_{2n-k}$ such that $J = (J_1,\ldots,J_{2n-k}) : M \rightarrow \mathcal{P}_{2n-k} \rightarrow \mathbb{R}^{2n-k}$ with $J_j = f_j \circ \pi_1$.  Define the fiber by \begin{align}
    \mathcal{L}_{(c_1,\ldots,c_{2n-k})}:=\pi_1^{-1} (c) = \{m \in M : \text{ } J_j(m) = f_j(c) , \text{ } 1 \leq j \leq 2n-k\}, \label{eq:fiberpi1}
\end{align}  which is a $k$-dimensional manifold with $d J_1 \wedge \cdots \wedge d J_{2n-k} \neq 0$.  Let $b = (b_1,\ldots,b_k) = \pi_2^{-1}(c)$, and choose local coordinates $(g_1,\ldots,g_k)$ on $\mathcal{B}_k$ and set $I_j : = g_j \circ \pi_2 \circ \pi_1 \in C_J(M)$. Then the Hamiltonian vector fields $X_{I_1},\ldots,X_{I_k}$ are tangent to $ \mathcal{L}_{(c_1,\ldots,c_{2n-k})}$ and span $T \mathcal{L}_{(c_1,\ldots,c_{2n-k})}$ such that $ \mathcal{L}_{(c_1,\ldots,c_{2n-k})}$ is isotropic. From \cite{MR0365629}, on each connected component of a generic fiber, these flows give a canonical affine structure, then $ \mathcal{L}_{(c_1,\ldots,c_{2n-k})} \cong \mathbb{R}^{k-l} \times (S^1)^l$ with $ 0 \leq l \leq k$ is a diffeomorphism.   

(ii) Now, we focus on the generic fibers in $\mathcal{P}_{2n-k}$.  The Poisson manifold $\mathcal{P}_{2n-k}$ has corank $k$, and the map $\pi_2: \mathcal{P}_{2n-k} \rightarrow \mathcal{B}_k$ is Poisson.  Similarly to the above analysis, for each $H_i := g_i \circ \pi_2 \in C^\infty(\mathcal{P}_{2n-k}) $ and any regular $b \in \mathcal{B}_k$, we have \begin{align}
        \pi_2^{-1}(b) :=\left\{ m_0 \in \mathcal{P}_{2n-k}: H_i (m_0) = g_i(b), \text{ } 1 \leq i \leq k\right\}. \label{eq:fiberpi2}
\end{align} These are level surfaces determined by the center functions in $C_I(M)$. Let $\mathcal{L}_{b,i} \subset \pi_2^{-1}(b)$ be the connected component. The Poisson structure $\pi_i$ on $\mathcal{P}_{2n-k}$ is of rank $k$ in each $\mathcal{L}_{b,i},$ that makes a symplectic manifold $\mathcal{L}_{b,i}$ of same dimension $ \dim \mathcal{P}_{2n-k} - \dim \mathcal{B}_k = 2n - 2k$. Hence, the fibers of $\pi_2$ form a foliation of $\mathcal{P}_{2n-k}$ such that they partition the regular locus $\mathcal{P}_{2n-k}$ into disjoint submanifolds, each of which being a level surface of the Hamiltonians. That is, \begin{align}
    \left(\mathcal{P}_{2n-k}\right)_{\mathrm{reg}} := \pi_2^{-1}\left((\mathcal{B}_k)_{\mathrm{reg}}\right) = \bigsqcup_{b \in \mathcal{B}_k} \pi_2^{-1}(b) =\bigsqcup_{b \in \mathcal{B}_k} \bigsqcup_i \mathfrak{L}_{b,i}.
\end{align} 
\end{remark}


\subsection{Geometric invariant quotient of the polynomial algebra \texorpdfstring{$S(\mathfrak{g})^A$}{S(g)\^{}A}}
\label{subsec:geomi}
Let $S(\mathfrak{g})=\mathrm{Sym}(\mathfrak{g})$ denote the symmetric algebra of $\mathfrak{g}$, which we identify with the algebra of polynomial functions on $\mathfrak{g}^*$ (and hence on $\mathfrak{g}$ after choosing an $\mathrm{Ad}(G)$-invariant bilinear form $B$ to identify $\mathfrak{g}^*\cong\mathfrak{g}$). Throughout the paper, unless explicitly stated otherwise, all group actions are \,\emph{left} actions. The only right $A$-action we use is the quotient action defining the cotangent bundle of the homogeneous space under the following trivialisation \begin{align*}
    (G\times \mathfrak{m})/A,
\end{align*}
whereas the actions on $\mathfrak{g}$ and on $S(\mathfrak{g})$ are the (left) adjoint action $\mathrm{Ad}$ and its induced pullback action on polynomial functions. See, Remark \ref{rmk:rightAaction} below.

Accordingly, when we write \textit{$A$-invariant}, we always mean $\mathrm{Ad}(A)$-invariant, and when we write \textit{$G$-invariant} (e.g., in $S(\mathfrak{g})^G$), we mean $\mathrm{Ad}(G)$-invariant. With these conventions, for any Lie subgroup $A\leq G$, we define \begin{align}
  S(\mathfrak{g})^A := \bigl\{f \in S(\mathfrak{g}) : f(\mathrm{Ad}(a)X)=f(X),\ \forall a\in A,\ \forall X\in \mathfrak{g}\bigr\}. \label{eq:sga}
\end{align} This is a Poisson subalgebra: indeed, $\{f,g\}$ is $A$-invariant whenever $f$ and $g$ are $A$-invariant.  It is well-known that determining a generating set for $S(\mathfrak{g})^{\mathfrak{a}}$ is, in fact, equivalent to performing computations within the Poisson algebra $S(\mathfrak{g})^A$.   Recall that for a finitely-generated integral domain $D$ over the base field $\mathbb{R}$, the rank of $D$ is the transcendence degree defined by \begin{align}
    \mathrm{rank}  D: = \mathrm{trdeg}  D. \label{eq:rank}
\end{align} Throughout this paper, the rank of a finitely-generated algebra is defined in \eqref{eq:rank}.
The main terminologies we applied in this subsection can be found in \cite[Chapter 6]{Dolgachev03} or \cite{Brion2010Actions}.  We define the  \textit{geometric invariant theory (GIT) quotient}  of $\mathfrak{g}^*$ by the coadjoint action of $A$ by \begin{align}
 \mathfrak{g}^*// A := \mathrm{Spec}\bigl(S(\mathfrak{g})^A\bigr). \label{eq:gitquotient}
\end{align}
If $A$ is (real) reductive, then $S(\mathfrak{g})^A$ is finitely generated and \eqref{eq:gitquotient} is the usual \textit{affine GIT quotient}. In particular, \begin{align}
    \dim \mathfrak{g}^*//A = \dim G - d_A   = \mathrm{trdeg}  \, S(\mathfrak{g})^A \text{ with }d_A = {\mathrm{max}}_{ X \in \mathfrak{g}} \left\{\dim (\mathrm{Ad}(A) \cdot X)\right\}. \label{eq:dimgitquo}
\end{align}   Fixed a $\mathrm{Ad}(G)$-invariant inner product $B$, which gives an identification between $\mathfrak{g}$ and $\mathfrak{g}^*$.  In what follows, we work on the affine variety $\mathfrak{g}//A$.  By definition \eqref{eq:gitquotient}, the inclusion $\mathsf{i}_A: S(\mathfrak{g})^A \hookrightarrow S(\mathfrak{g})$ induces a unique morphism of the affine variety \begin{align}
    \chi_A: \mathfrak{g} = \mathrm{Spec} \, S(\mathfrak{g})  \longrightarrow \mathfrak{g}//A = \mathrm{Spec} \, S(\mathfrak{g})^A, \qquad X \longmapsto (f \mapsto f(X)) \label{eq:canonicalpro}
\end{align} such that its pullback $  \mathsf{i}_A^*$ is equal to $\chi_A$ for any $f \in S(\mathfrak{g})^A$.  It is clear that the quotient mapping in \eqref{eq:canonicalpro} is well-defined, as for any two points $X,Y \in \mathfrak{g}$, we have $\chi_A(X) = \chi_A(Y)$ if and only if $f(X) = f(Y)$. In particular, $\chi_A$ is a Poisson morphism. 

An element in $\mathfrak{g}//A$ can be written as \begin{align}
    [X]_A : = \chi_A^{-1}\left(\chi_A(X)\right) =   \left\{Y \in \mathfrak{g}: f(Y) = f(X), \text{ } f \in S(\mathfrak{g})^A\right\}. \label{eq:pointsinquo}
\end{align}   More precisely, if $S(\mathfrak{g})^A$ is generated by homogeneous polynomials $f_1,\dots,f_\zeta$, let $\mathbb{A}^\zeta:=\mathbb{A}^\zeta_{\mathbb{R}}\cong \mathbb{R}^\zeta$ denote affine $\zeta$-space with the coordinate ring $\mathbb{R}[x_1,\ldots,x_\zeta]$. The generators define the polynomial map
\begin{align*}
F : \mathfrak{g} \longrightarrow \mathbb{A}^\zeta, \qquad X \longmapsto (f_1(X),\ldots,f_\zeta(X)).
\end{align*}
Equivalently, there is a surjective $\mathbb{R}$-algebra homomorphism
\begin{align*}
\mathrm{ev}: \mathbb{R}[x_1,\ldots,x_\zeta] \twoheadrightarrow S(\mathfrak{g})^A,\qquad x_i\longmapsto f_i,
\end{align*}
with kernel $I:=\ker(\mathrm{ev})$. Then $\mathfrak{g}//A\cong V(I)\subset \mathbb{A}^\zeta$ is an affine subvariety, and $\chi_A$ factors as $\mathfrak{g}\xrightarrow{\,F\,}\mathbb{A}^\zeta\twoheadrightarrow V(I)$. Here
\begin{align*}
    V(I) := \left\{u \in \mathbb{A}^\zeta : f(u) = 0 \text{ for all } f \in I\right\}.
\end{align*}
In particular, $\mathfrak{g}//A \cong \mathbb{A}^\zeta$ holds if and only if $I=(0)$, equivalently $S(\mathfrak{g})^A \cong \mathbb{R}[f_1,\ldots,f_\zeta]$. 
In particular, for all $f \in S(\mathfrak{g})^A$, \begin{align*}
    \chi_A(X)=\chi_A(Y) \;\Longleftrightarrow\; \big(f_1(X),\dots,f_r(X)\big)=\big(f_1(Y),\dots,f_r(Y)\big) \;\Longleftrightarrow\;f(X)=  f(Y) 
\end{align*} as $f$ is a polynomial in $f_i$.   For a semisimple or compact group $A$ acting on affine varieties, each fiber of $\chi_A$ contains a unique closed $A$-orbit, and $\chi_A$ maps closed orbits bijectively onto $\mathfrak{g}//A$. If $A$ is compact here, all $A$-orbits in $\mathfrak{g}$ are closed, so $\chi_A(X)=\chi_A(Y)$ if and only if $A \cdot X=A \cdot Y$. 

\begin{definition} 
\label{def:regularstr} \cite[Chapter 1, \S 5]{HartshorneAG}
Let $A$ be a semisimple group acting linearly on $\mathfrak{g}$.  We define the regular locus of the quotient morphism by \begin{align*}
    (\mathfrak{g}//A)_{\mathrm{reg}} := \{ y\in \mathfrak{g}//A : \ \chi_A \text{ is smooth over } y  \}.
\end{align*} Equivalently, for any $y\in (\mathfrak{g}//A)_{\mathrm{reg}}$,  there exists an open neighborhood $y\in U\subset \mathfrak{g}//A$ such that $\chi_A^{-1}(U)\to U$ is a smooth morphism.
\end{definition}

\begin{proposition} 
\label{prop:smoothregularg} \cite[Proposition 1.26]{Brion2010Actions}
Let $\mathfrak{g}//A$, and $(\mathfrak{g}//A)_{\mathrm{reg}}$ be the same as defined above. Then $\mathfrak{g}//A$ carries a unique Poisson structure for which $\chi_A$ is a Poisson morphism. Moreover, $(\mathfrak{g}//A)_{\mathrm{reg}}$ is a Zariski open dense subset of $\mathfrak{g}//A$, and the restricted map \begin{align*}
    \chi_A:\ \chi_A^{-1}\big((\mathfrak{g}//A)_{\mathrm{reg}}\big)   \longrightarrow (\mathfrak{g}//A)_{\mathrm{reg}}
\end{align*}  is a smooth Poisson morphism. In particular, the induced map is a Poisson submersion of smooth manifolds.
\end{proposition}

\begin{proof}
As $\mathfrak{g}$ is a Poisson manifold, generic smoothness on $\mathfrak{g}//A$ implies that there exists a Zariski open dense subset $U\subset \mathfrak{g}//A$ such that $\chi_A^{-1}(U)\to U$ is smooth. By definition, $U\subset (\mathfrak{g}//A)_{\mathrm{reg}}$, hence $(\mathfrak{g}//A)_{\mathrm{reg}}$ is non-empty and open dense. Restricting $\chi_A$ to $\chi_A^{-1}\big((\mathfrak{g}//A)_{\mathrm{reg}}\big)$ yields a smooth Poisson morphism. 
\end{proof}

    

\begin{remark}
The GIT quotient construction introduced above does not rely on Lie algebra structure on $\mathfrak{g}$. In general, suppose that $V$ is a finite-dimensional  vector space over $\mathbb{R}$ with a linear action of an algebraic group $A$, then $\mathbb{R}[V]=\mathrm{Sym}(V^*)$ and one defines the affine quotient $V//A:=\mathrm{Spec}(\mathbb{R}[V]^A)$. The quotient morphism $\chi_A:V\to V//A$ is the unique map induced by the inclusion $\mathbb{R}[V]^A\hookrightarrow \mathbb{R}[V]$. Consequently, for any non-zero $v,w \in V$, $\chi_A(v)=\chi_A(w)$ if and only if $f(v)=f(w)$ for all $f\in\mathbb{R}[V]^A$, and the fiber $[v]_A:=\chi_A^{-1}(\chi_A(v))$ is the corresponding invariant equivalence class.
\end{remark}

\section{Superintegrable system on \texorpdfstring{$T^*(G/A)$}{T*(G/A)}}
\label{sec:superconstruction}
Our preliminary work in Section \ref{sec:superconstruction} is to develop a formal framework for building superintegrable systems on a cotangent bundle of an adjoint orbit, as described in \cite{MR2141306} within the context of Mischenko-Fomenko non-commutative integrability. In the rest of this paper, we assume that $G$ is a real, compact, semisimple, simply connected Lie group with Lie algebra $\mathfrak{g}$ endowed with the canonical Lie-Poisson bracket. For any $f,g\in S(\mathfrak{g})$, \eqref{eq:poiss} becomes \begin{align*}
  \{f,g\}(X) = B( X,[\nabla f(X),\nabla g(X)]),\quad X\in\mathfrak{g},  
\end{align*}  where $\nabla f(X)\in \mathfrak{g}$ denotes the gradient of $f$ in $X$, similar to $\nabla g(X)$. Let $B$ be the (negative definite) Killing form on $\mathfrak{g}$ such that $\mathfrak{g}\cong\mathfrak{g}^*$. Fix $W \in\mathfrak{g}$ and write \begin{align*}
\mathfrak{a}:=\mathfrak{z}_{\mathfrak{g}}(W) =\left\{X\in\mathfrak{g}:[W,X]=0\right\}.
\end{align*} Let $A:=Z_G(W)$ be the centralizer of $W$ in $G$. Then $A$ is a closed subgroup of $G$ with Lie algebra $\mathfrak{a}$. In this way, we can decompose the Lie algebra as $\mathfrak{g} = \mathfrak{a} \oplus \mathfrak{m}$ reductively, where $\mathfrak{a} = \mathrm{Lie}(A)$ and $\mathfrak{m} = \mathrm{Im}(\mathrm{ad} W)$ is the $B$-orthogonal complement of $\mathfrak{a}$ in $\mathfrak{g}$. That is, the reductive commutator relations are $[\mathfrak{m},\mathfrak{m}]\subset\mathfrak{a}$ and $[\mathfrak{m},\mathfrak{a}]\subset\mathfrak{m}$. In particular, under this construction, we have $\mathfrak{a}$ such that $A = \exp \left(\ker \mathrm{ad} W\right)$. Define the homogeneous space by \begin{align}
M:=G/A. \label{eq:homoge}
\end{align}  Using left translations, we can identify $T_{[e]}M\cong \mathfrak{m}$. Moreover, through $B$, we have the identification of $\mathfrak{m}^*\cong \mathfrak{m}$. Note that the homogeneous space defined in \eqref{eq:homoge} can be realised as the adjoint orbit $\mathcal{O}_W: = \mathrm{Ad}(G)\cdot W$. 
In other words, $G/A \cong \mathrm{Ad}(G) \cdot W$. In the rest of the sections, we always realise the homogeneous space $M$ as an adjoint orbit. 

Let $T\subset G$ be a maximal torus with Lie algebra $\mathfrak{t}$. 
Let $\Phi\subset\mathfrak{t}^*$ be the root system of $(\mathfrak{g}^{\mathbb{C}},\mathfrak{t}^{\mathbb{C}})$, and fix a positive system $\Phi^+\subset\Phi$. For $\alpha\in\Phi$, let $\mathfrak{g}_{\alpha}\subset\mathfrak{g}^{\mathbb{C}}$ be the corresponding root space. Define the vanishing root set of $W$ by \begin{align}
\Phi_W:=\{\alpha\in\Phi: \alpha(W)=0\},\quad \Phi_W^+:=\Phi_W\cap\Phi^+.
\end{align}

\begin{definition} 
\label{def:regularirr}
The element $W$ is called \textit{regular} if $\Phi_W=\emptyset$ and \textit{irregular} if $\Phi_W\neq\emptyset$. 
\end{definition}

\begin{proposition} 
\label{prop:characterizationirrarr}
Let $M=\mathrm{Ad}(G)\cdot W$ be the adjoint orbit. Under the notation above, we have the following characterization of the regular and irregular points: 

(a) If $W  $ is regular, then $\dim\mathfrak{a}=\mathrm{rank}\, \mathfrak{g}$. In particular,   $\dim M=\dim G-\mathrm{rank}\, \mathfrak{g}$.  

(b) If $W$ is irregular,  then $\dim\mathfrak{a}>\mathrm{rank}\, \mathfrak{g}$ such that the subalgebra $\mathfrak{a}$ has the following decomposition: \begin{align*}
\mathfrak{a} = \mathfrak{t} \oplus  \bigoplus_{\alpha\in\Phi_W^+}\mathfrak{g}_{\alpha}^{\mathbb{R}},
\end{align*} where $ \mathfrak{g}_{\alpha}^{\mathbb{R}}:=\big(\mathfrak{g}_{\alpha}\oplus\mathfrak{g}_{-\alpha}\big)\cap\mathfrak{g} $. In particular, $\dim M<\dim G-\mathrm{rank}\, \mathfrak{g}$.  
\end{proposition}
\begin{remark} 
\label{rmk:consequences}
In the regular case, $A=T$ and $M=G/T$ is the full flag manifold of $G$. In the irregular case, $A\supsetneq T$ and $M=G/A$ is a compact, smooth homogeneous space of smaller dimension: a partial flag manifold. Hence, \begin{align*}
\dim M=\dim G-\dim A =\dim G-\big(\mathrm{rank}\, \mathfrak{g}+2\,\left|\Phi_W^+\right|\big).
\end{align*}
So $\dim M<\dim(G/T)=\dim G -\mathrm{rank}\, \mathfrak{g}$ if $\Phi_W\neq\emptyset$. Many constructions used later, such as moment maps, invariant polynomials on affine slices $\mathfrak{m}-\varepsilon W$, and Poisson quotients, work uniformly in both cases. 
\end{remark}

Let $ T^*M$ be the cotangent bundle of $M$. Assume that it is equipped with a twisted symplectic form \begin{align}
  \omega_\varepsilon = \omega_{\mathrm{can}} +  \varepsilon\pi^*\omega_{\mathrm{KKS}}: T^*(T^*M) \times T^*(T^*M) \rightarrow \mathbb{R},  
\end{align} where $\pi: T^*M \rightarrow M$ is the canonical projection and $\varepsilon \in \mathbb{R}/\{0\}$. Note that for a sufficiently small value of $\varepsilon$, the compatibility of the Poisson bracket always holds. Hence, there always exists a twisted symplectic structure if $\varepsilon$ is sufficiently small. Hence, from this point onward, $(T^*(G/A),\omega_\varepsilon)$ is always a twisted symplectic manifold.

Section \ref{sec:superconstruction} is split into two subsections. In Subsection \ref{subsec:confirst}, we study the magnetic Hamiltonian flow on the cotangent bundle $T^*M$ and provide explicit formulas for the first integrals and their corresponding polynomial Poisson algebras on $T^*M$. In Subsection \ref{subsec:Poisuper}, we then demonstrate how these geometric structures yield alternative perspectives on the construction of magnetic geodesic flows from \cite{MR2141306} and generate superintegrable systems on $T^*(G/A)$ in the sense of Definition \eqref{def:superge}.

\subsection{Construction of the first integrals on \texorpdfstring{$T^*(G/A)$}{T*(G/A)}}
\label{subsec:confirst}
Now, we shall focus on the construction of superintegrable systems in phase space $T^*M$. The concept of superintegrability in Hamiltonian mechanics has significant implications for mathematical physics and symplectic geometry. In \cite{MR2141306}, the author studied magnetic geodesic flows in homogeneous spaces, establishing algebraic integrability conditions using invariant polynomial functions. In Subsection \ref{subsec:confirst}, we first review some results, such as Poisson structures and Hamiltonian vector fields in \cite{MR2141306}, and then provide a description of the corresponding polynomial Poisson algebras generated by the first integrals.  
Let $G,\mathfrak{g}$ and $M$ be the same as defined above.  In a left trivialization, we have \begin{align}
   T^*M \,\cong\, \left(G\times \mathfrak{m}\right)/A,\qquad \mathsf{p} \longmapsto [g,X]_A,  \label{eq:lefttran} 
\end{align} Here $[g,X]_A$ denotes the $A$-equivalence class of $(g,X)\in G\times \mathfrak{m}$ with respect to the (right) $A$-action \begin{align}
(g,X)\cdot a := \bigl(ga,\,\mathrm{Ad}(a^{-1})X\bigr),\qquad a\in A, \label{eq:righta}
\end{align} such that $(g,X)\sim (g',X')$ if and only if there exists $a\in A$ with $g'=ga$ and $X'=\mathrm{Ad}(a^{-1})X$.

\begin{remark}
\label{rmk:rightAaction}
The formula \eqref{eq:righta} is the standard right $A$-action used to form the associated bundle $G\times_A\mathfrak{m}$ over $G/A$ (where $A$ acts on the fiber $\mathfrak{m}$ via the adjoint representation restricting on $\mathfrak{m}$). We may adopt the convention $(g,X) \cdot a= (ga,\mathrm{Ad}(a)X)$ by composing with inversion in $A$, but then all descended maps and invariance statements must be adjusted accordingly. In this manuscript, whenever we speak of a \emph{right} group action, it is this right $A$-action on $G\times\mathfrak{m}$ (and its induced right action on related affine slices) that is intended.
\end{remark}

Moreover, writing $[g]:=gA\in M$, the element $X\in\mathfrak{m}$ represents the corresponding covector at $[g]$ under left trivialization: using the left translation $L_g:G/A \rightarrow G/A$ and the fixed $\mathrm{Ad}(G)$-invariant bilinear form $B$ to identify $\mathfrak{m}\cong (\mathfrak{g}/\mathfrak{a})^*$, the covector $\mathsf{p}\in T^*_{[g]}M$ is characterized by
\begin{align*}
\mathsf{p}\bigl((L_g)_*\bar{v}\bigr) = B(X,v),\qquad v\in\mathfrak{m}, \quad (\bar{v}\in \mathfrak{g}/\mathfrak{a}),
\end{align*} where $L_g$ is the left-translator.
For brevity, we will often write $[g,X]$ instead of $[g,X]_A$ when the group $A$ is understood. Hence, the left-trivialization in \eqref{eq:lefttran} identifies the points of $T^*M$ as $[g,X]$ with $g \in G$ and $X \in \mathfrak{m}$. Let us now look at the tangent space of $T^*M$.  
Recall that from \cite{MR2141306}, under the explicit local coordinate $[g,X]$, a generic tangent vector at $(g,X)$ may be written as\begin{align}
\nonumber (\mathfrak{g}/\mathfrak{a})\oplus \mathfrak{m} & \cong   T_{[g,X]}(T^*M),\\
   \ (\bar{v},\delta X)&\longmapsto  g_*\big\vert_X \left(v,\,-\frac12[v,X]+\delta X\right)  , \label{eq:coordinate01}
\end{align} where $\bar{v} \in \mathfrak{g}/\mathfrak{a} \cong T_{eA}(G/A)$ is the base direction (with representative $v\in\mathfrak{m}$), and $\delta X\in T_X\mathfrak{m}\cong \mathfrak{m}$ is the vertical direction. Here $g_*\vert_X : \mathfrak{m} \oplus \mathfrak{m} \xrightarrow{\ \cong \ } T_{(g,X)}T^*M$ denotes the left-trivialization isomorphism. We define $\delta X$ in the following way: given any smooth curve $t\mapsto [g(t),X(t)]\in (G\times \mathfrak{m})/A$ with $[g(0),X(0)]=[g,X]$, choose a representative with $g(0)=g$ and $X(0)=X$. In other words, the pair $(\bar{v},\delta X)$ records the derivatives at $t=0$ of the base curve $g(t)A\in G/A$ and of the tangent flow $X(t) \in \mathfrak{m}$. Then the induced tangent vector at $t=0$ has a vertical component \begin{align*}
\delta X:=\left.\frac{d}{dt}\right\vert_{t=0} X(t)\in \mathfrak{m}.
\end{align*}
For more details on the derivation of the coordinate description \eqref{eq:coordinate01}, we refer the reader to \cite{Thimm1981,MR2141306}. Moreover, with the Killing form $B$, we can identify $(\mathfrak{g}/\mathfrak{a})^* \subset \mathfrak{g}$ and $\mathfrak{m}^*\cong \mathfrak{m}$. Hence, \begin{align*}
T^*_{[g,X]}(T^*M)\cong (\mathfrak{g}/\mathfrak{a})^*  \oplus \mathfrak{m},\qquad
\langle (\eta,\mu),(\bar{v},\delta X)\rangle = B(\eta,v) + B(\mu,\delta X).
\end{align*} Here $\langle \cdot,\cdot\rangle : T_{[g,X]}^*(T^*M) \times T_{[g,X]}(T^*M) \rightarrow \mathbb{R}$ is a canonical dual pair. Thus, the first component $\eta$ pairs with the base direction $\bar{v}$, and the second component $\mu$ pairs with the vertical direction $\delta X$.
 Note that, due to the identification $\mathfrak{g}/\mathfrak{a} \cong \mathfrak{m}$, throughout this paper we will drop the bar on $\bar{v}$ and the identity $\bar{v}$ with $v \in \mathfrak{m}$, such that $\mathfrak{g}/\mathfrak{a} \oplus \mathfrak{m} $ is identified with $ \mathfrak{m} \oplus \mathfrak{m}$. Next, we define the twisted Poisson structures on $T^*(T^*M)$.

\begin{definition}
\label{def:setup}
    Fix $W \in \mathfrak{g}$ with $\mathrm{Ad}(A)W=W$, and a real parameter $\varepsilon\in\mathbb{R}/\{0\}$. Let $\omega_{\mathrm{can}}$ be the canonical symplectic form on $T^*M$, and let $\omega_{\mathrm{KKS}}$ be the Kirillov-Kostant-Souriau (KKS) two-form on $M$ at $[e] = eA$ by \begin{align*}
  \omega_{\mathrm{KKS}}([e])\big([X],[Y]\big) = B\big(W,[X,Y]\big),\qquad X,Y\in\mathfrak{m}   
\end{align*} and extend $\omega_{\mathrm{KKS}}$ to all of $M$ by left $G$-invariance. Define $\omega_\varepsilon := \omega_{\mathrm{can}}+\varepsilon\,\pi^*\omega_{\mathrm{KKS}},$ where $\pi:T^*M\to M$ is the canonical projection. In a local section $\sigma: U \subset \mathfrak{m} \rightarrow G$, one finds \begin{align}
    \omega_{\varepsilon} =  dB(X,\sigma^{-1}d\sigma) +  \varepsilon B(W,[\sigma^{-1}d\sigma,\sigma^{-1}d\sigma]), \label{eq:symplec}
\end{align} which is an expression of the local coordinates. We now consider the induced Poisson bracket from $\omega_\varepsilon$ in \eqref{eq:symplec}. In local coordinates $(x_i,p_i)$ on $T^*M$ 
adapted to a coordinate chart on $M$, we can write  $\omega_{\mathrm{KKS}}=\frac{1}{2}\sum_{i,j}F_{ij}(x)\,dx_i\wedge dx^j$. Then the associated twisted Poisson bracket in $T^*M$ is \begin{align}
    \{f,g\}_\varepsilon =\underbrace{\sum_{i=1}^n\big(\partial_{x_i}f\,\partial_{p_i}g-\partial_{x_i}g\,\partial_{p_i}f\big)}_{\text{$ = \omega_{\mathrm{can}}^{-1}$}}+\varepsilon\underbrace{\sum_{1 \leq i,j \leq n}F_{ij}(x)\,\partial_{p_i}f\,\partial_{p_j}g}_{\text{$ = \omega_{\mathrm{KKS}}^{-1}$}}, \label{eq:mp1}
\end{align} where $F_{ij}(x) = -F_{ji}(x)$. 
\end{definition}

We now compute Hamiltonian vector fields in the coordinates given in Definition \ref{def:setup}. By the definition of Hamiltonian vector fields, a non-zero vector field $X_\mathcal{H} \in \mathfrak{X} (T^*M)$ satisfies $\iota_{X_\mathcal{H}} \omega_\varepsilon = d\mathcal{H}$. Let $\mathcal{H}(\textbf{x},\textbf{p}) = x_k$ be a coordinate function. We find $X_{x_k} = \partial_{p_k}$,  where $\partial_{p_k} = \dfrac{\partial}{\partial p_k}$. Indeed, the wedge product formula gives $\iota_X(\alpha \wedge \beta) = \iota_X\alpha \wedge \beta + (-1)^{\deg \alpha} \alpha \wedge \iota_X \beta$ for any $X \in \mathfrak{X}(T^*M)$ and $\alpha,\beta \in \Omega^1(T^*M)$. Then \begin{align*}
    \iota_{\partial_{p_k}}(dx_j \wedge dp_j) = \left(\iota_{\partial_{p_k}} dx_j \right) \wedge d p_j - d x_j \wedge  \left(\iota_{\partial_{p_k}}dp_j\right) = - \delta_{kj} dx_j, 
\end{align*} where $\delta_{kj}$ is the Kronecker delta. Using \eqref{eq:mp1}, we obtain $\iota_{\partial_{p_k}} \omega_{\varepsilon} = dx_k$. Similarly, for $\mathcal{H} (\textbf{x},\textbf{p})= p_k$, solving $\iota_{X_{p_k}} \omega_\varepsilon = d p_k$ gives $X_{p_k} = - \partial_{x_k} +  \varepsilon \sum_{j = 1}^n F_{jk}(x) \partial_{p_j}$. Alternatively, it is an easy but computational exercise to show \cite{MR2141306} $$\{x_i,x_j\}_\varepsilon=0,\qquad\{x_i,p_j\}_\varepsilon=\delta^i_j,\qquad \{p_i,p_j\}_\varepsilon=\varepsilon F_{ij}(x).$$ This can be proved by a direct calculation. Given a Hamiltonian $\mathcal{H}(\textbf{x},\textbf{p}) \in C^\infty(T^*M)$, the \textit{magnetic Hamiltonian equations} are \begin{align}
    \dfrac{d x_i}{dt} = \dfrac{\partial \mathcal{H}}{\partial p_i}, \qquad \dfrac{d p_i}{dt} = - \dfrac{\partial \mathcal{H}}{\partial x_i} +  \varepsilon F_{ij}(x) \dfrac{\partial \mathcal{H}}{\partial p_j}. \label{eq:Hamiltoneq}
\end{align}

Let \begin{align}
    \mathcal{H}(g,X)=\frac{1}{2}\,B(X,X)    . \label{eq:hamiltonianquad}
\end{align}    In the rest of this work, we will write $[g,X] := (g,X)\in (G\times\mathfrak{m})/A$. Note that $\mathcal{H}$ is well-defined on $T^*M$ as, for any $a \in A$, $$ \mathcal{H}(ga, \mathrm{Ad}(a^{-1})(X)) = \frac{1}{2} B(\mathrm{Ad}(a^{-1})X,\mathrm{Ad}(a^{-1})X ) = \frac{1}{2} B(X,X) = \mathcal{H}(g,X).$$ Hence, $\mathcal{H}$ is both left $G$-invariant and $\mathrm{Ad}(a)$-invariant. Let us remark that in \cite{MR2141306}, the author implements the magnetic term into the Hamiltonian $\mathcal{H}(g,X)$ and defines a twisted Hamiltonian $\mathcal{H}_\varepsilon (g,X) = \frac{1}{2} B(X - \varepsilon W, X - \varepsilon W)$. Since $A = \exp(\ker\mathrm{ad}(W))$, we have $W\in \ker\mathrm{ad}(W)=\mathfrak{a}$, and by our choice of reductive decomposition $\mathfrak{g}=\mathfrak{a}\oplus \mathfrak{m}$ with $\mathfrak{m}=\mathfrak{a}^\perp$ (with respect to $B$), it follows that $B(X,W) = 0$ for all $X\in\mathfrak{m}$. Expanding the square, therefore, shows that \begin{align*}
    \mathcal{H}_\varepsilon (g,X) = \frac{1}{2} B(X,X) + \frac{\varepsilon^2}{2} B(W,W) = \mathcal{H}(g,X) + c.
\end{align*} Here, $c =  \frac{\varepsilon^2}{2} B(W,W)$ is a constant. Hence, $\mathcal{H}_\varepsilon$ differs from $\mathcal{H}$ by the constant $c$, and they have identical dynamics and, therefore, generate the same Hamiltonian vector fields and first integrals.

In what follows we use the left trivialisation $T^*M\cong (G\times\mathfrak{m})/A$ and write $[g,X]$ for the $A$-equivalence class of $(g,X)\in G\times\mathfrak{m}$. We now introduce a (twisted) left moment map and an affine shift on the $\mathfrak{m}$-factor. 

\begin{definition}
\label{def:maps}
\textup{(i)} Define the smooth map $P:T^*M \rightarrow \mathfrak{g}^*$ by
\begin{align*}
P([g,X]):= P(g,X) =\mathrm{Ad}(g)\bigl(X-\varepsilon W\bigr).
\end{align*}
This is well-defined since $A = \exp(\ker\mathrm{ad}(W))$ implies $\mathrm{Ad}(a)W = W$ for all $a \in A$.

\textup{(ii)} Define the affine translation $\tau_{-\varepsilon W}:\mathfrak{m} \to \mathfrak{m} - \varepsilon W$ by $\tau_{-\varepsilon W}(X)=X-\varepsilon W$, and define
\begin{align*}
\widetilde{\pi}_{\mathfrak{m}}:G\times\mathfrak{m}\to\mathfrak{m}-\varepsilon W,\qquad \widetilde{\pi}_{\mathfrak{m}}(g,X):=X-\varepsilon W.
\end{align*}

\end{definition}
\begin{remark}
(i) Fix an $\mathrm{Ad}(G)$-invariant nondegenerate symmetric bilinear form $B$ on $\mathfrak{g}$, and use it to identify $\mathfrak{g}\cong \mathfrak{g}^*$ via the musical isomorphisms  $\flat  : \mathfrak{g}\to \mathfrak{g}^*$,  $x^\flat := B(x,\cdot)$, $\sharp  : \mathfrak{g}^*\to \mathfrak{g}$,  $\sharp = \flat^{-1}$.  Let \begin{align*}
P  : T^*M \to \mathfrak{g}^*, \qquad  \begin{matrix}
    J  : = P^\sharp = B^{-1}\circ P : T^*M \to \mathfrak{g}, \\
J([g,X])  : = \mathrm{Ad}(g)\bigl(X-\varepsilon W\bigr).
\end{matrix} 
\end{align*} Hence, $P$ can be viewed equivalently as a $\mathfrak{g}$ valued moment map.
For $Y,X_i\in\mathfrak{g}$, define the component functions \begin{align}
P_Y  := \langle P, Y\rangle = B\bigl(J, Y\bigr)  \text{ and }  P_i  := P_{X_i} = \langle P, X_i\rangle = B\bigl(J, X_i\bigr)  . \label{eq:momcoord}
\end{align} In what follows, we adopt the moment map convention
\begin{align}
\iota_{\widehat{u}}\,\omega_{\varepsilon} &= dP_u,  \qquad  P_u := \langle P, u\rangle = B\bigl(J, u\bigr), \qquad  u\in\mathfrak{g}, \label{eq:momcoven}
\end{align} where $\widehat{u}$ is the fundamental vector field generated by $u$ and $\omega_{\varepsilon}$ is the magnetic symplectic form.

(ii) By the $G$-invariant metric induced by $B$, we identify $TM\cong  T^*M$. Under this identification, the twisted symplectic form $ \omega_\varepsilon=\omega_{\mathrm{can}}+\varepsilon\,\pi^*\omega_{\mathrm{KKS}}$ on $T^*M$ corresponds to the metric twist on $TM$ obtained by adding $\varepsilon \pi^*\omega_{\mathrm{KKS}}$ to the canonical symplectic form given by the metric. In other words,  we can equivalently work on $TM$ with a metric $B$ and a magnetic term.  From Section \ref{sec:superconstruction} onward, we continue to work on $T^*M$, not $TM$, and consistently use the convention $P([g,X]) = \mathrm{Ad}(g)(X- \varepsilon W)$ as defined in Definition \ref{def:maps}.

(iii) 
Let $q:G\times\mathfrak{m}\to G\times_A\mathfrak{m}\cong T^*M$ be the quotient map. The map $\widetilde{\pi}_{\mathfrak{m}}$ is a representatitive level map and, in general, is not $A$-invariant. Hence, it does \emph{not} descend to a well-defined map $\pi_\mathfrak{m}: T^*M\to\mathfrak{m}-\varepsilon W$. Indeed, such map is well-defined only if $\mathrm{Ad}(A)$ acts trivially on $\mathfrak{m}$. We will, by abuse of notation, write \begin{align*}
\pi_{\mathfrak{m}}([g,X]) := X-\varepsilon W.
\end{align*}
For $\theta\in S(\mathfrak{m}-\varepsilon W)^A$, the function $\theta\circ\widetilde{\pi}_{\mathfrak{m}}$ is $A$-invariant and hence descends to a well-defined function on $T^*M$. We define  \begin{align*}
\pi_{\mathfrak{m}}^*\theta\in C^{\infty}(T^*M),\qquad (\pi_{\mathfrak{m}}^*\theta)([g,X]) := \theta\bigl(X-\varepsilon W\bigr),
\end{align*} to be the unique descended function characterised by $q^*(\pi_\mathfrak{m}^*\theta) = \theta \circ \widetilde{\pi}_\mathfrak{m}$, which will be shown well-definedness in Lemma \ref{lem:property} \textup{(iii)}. In all the later computation below, the expressions such as $\theta \circ \pi_\mathfrak{m},  \theta(\pi_{\mathfrak{m}}(g,X))$, or $\pi_{\mathfrak{m}}(t)=X(t)-\varepsilon W$ are shorthand for this descended invariant pullback after pulling back by $q$.
\end{remark}

Up to this point, let us remark that in this work, we identify $\mathfrak{g}^* \cong \mathfrak{g}$ and $T^*M \cong TM$ via $B$. Under these identifications, the map $P$ (in Definition \ref{def:maps}) serves as the standard moment map for the left $G$-action, satisfying $\iota_{\widehat{u}} \omega_\varepsilon = d P_u$, and there is no differences between the dual pair $\langle \cdot,\cdot \rangle$ and the Killing form $B(\cdot,\cdot)$. In the rest of this work, to maintain formality, we will still use the dual pair to define the functional induced by the twisted moment map $P$.

We now record the equivariance properties of the moment map $P$ and of the descended $A$-invariant functions (such as $\pi_{\mathfrak{m}}^*\theta$) introduced above. We will invoke these properties repeatedly in later sections. We begin with the following lemma.

 \begin{lemma}\label{lem:descent}
Let $A$ act smoothly on a manifold $N$. Assume that the orbit space $N/A$ carries a smooth manifold structure such that the quotient map $q:N\to N/A$ is smooth (e.g. the action is free and proper). If $F\in C^{\infty}(N)$ is $A$-invariant, i.e. $F(a\cdot x)=F(x)$ for all $a\in A$ and $x\in N$, then there exists a unique $\overline{F}\in C^{\infty}(N/A)$ such that $F=\overline{F}\circ q$.
\end{lemma}
\begin{proof}
Define $\overline{F}([x]):=F(x)$, where $[x]$ denotes the $A$-orbit of $x\in N$. This is well-defined: if $[x]=[y]$, then $y=a\cdot x$ for some $a\in A$, hence $F(y)=F(a\cdot x)=F(x)$ by $A$-invariance. Uniqueness follows from the surjectivity of $q$.

Finally, since $q$ is smooth and surjective, the identity $F=\overline{F}\circ q$ implies that $\overline{F}$ is smooth on $N/A$.
\end{proof}

\begin{lemma}
\label{lem:property}
    Let $P$ and $\widetilde\pi_{\mathfrak{m}}$ be the same as defined above, and let $\xi = X  - \varepsilon W \in\mathfrak{g}$ \footnote{Throughout this paper, we fix the notation $\xi = X - \varepsilon W$ to be a regular element. Moreover, we further assume that $X \in \mathfrak{m}_{\mathrm{reg}}$ and $W$ could be either regular or irregular.} be a regular element for some $X \in \mathfrak{m}$ and a fixed $W \in \mathfrak{g}$. We then have the following arguments: 
    
\textup{(i)} $P$ is left $G$-equivariant. That is, $P(h\cdot[g,X]) = \mathrm{Ad}(h)P([g,X])$.
    
\textup{(ii)} The affine translation $\tau_{-\varepsilon W}:\mathfrak{m}\to\mathfrak{m}-\varepsilon W$ is $\mathrm{Ad}(A)$-equivariant and $S(\mathfrak{m}-\varepsilon W)^A\cong S(\mathfrak{m})^A$ via the translation $\tau_{-\varepsilon W}: \mathfrak{m} \rightarrow \mathfrak{m} -\varepsilon W$. 

\textup{(iii)} Let $A$ act on $G\times\mathfrak{m}$ on the right by $(g,X) \cdot a=(ga,\mathrm{Ad}(a^{-1})X)$, and let $A$ act on the affine space $\mathfrak{m}-\varepsilon W$ by $\mathrm{Ad}(a^{-1})$.
Define the map $\widetilde{\pi}_{\mathfrak{m}}:G\times\mathfrak{m}\to\mathfrak{m}-\varepsilon W$ by $\widetilde{\pi}_{\mathfrak{m}}(g,X)=X-\varepsilon W$. Then \begin{align*}
    \widetilde{\pi}_{\mathfrak{m}}\big((g,X)\cdot a\big)=\mathrm{Ad}(a^{-1})\,\widetilde{\pi}_{\mathfrak{m}}(g,X) .
\end{align*} Consequently, for any $A$-invariant polynomial $\theta\in S(\mathfrak{m}-\varepsilon W)^A$, the polynomial $\theta\circ\widetilde{\pi}_{\mathfrak{m}}$ is $A$-invariant on $G\times\mathfrak{m}$ and descends to a well-defined smooth function $\pi_\mathfrak{m}^*\theta\in C^\infty(T^*M)$ given by \begin{align*}
   (\pi_\mathfrak{m}^*\theta)([g,X])=\theta\bigl(X-\varepsilon W\bigr). 
\end{align*} Moreover, $\pi_\mathfrak{m}^*\theta$ is left $G$-invariant.
\end{lemma}

\begin{proof}
    We first show part (i). In the left trivialisation $T^*M \cong (G\times\mathfrak{m})/A$, for any $h \in G$, the left $G$-action is $h\cdot[g,X]=[hg,X]$. By Definition \ref{def:maps}, $ P([g,X])=\mathrm{Ad}(g)(X-\varepsilon W)$. Therefore, a direct computation shows that \begin{align*}
P\bigl(h\cdot[g,X]\bigr) = P([hg,X]) = \mathrm{Ad}(hg)(X-\varepsilon W) = \mathrm{Ad}(h)\bigl(\mathrm{Ad}(g)(X-\varepsilon W)\bigr) = \mathrm{Ad}(h)\,P([g,X]).
\end{align*} Hence, $P$ is $G$-equivariant.

Now consider the second part. Define the translation $\tau_{-\varepsilon W}:\mathfrak{m}\to\mathfrak{m}-\varepsilon W$ by $\tau_{-\varepsilon W}(X)=X-\varepsilon W$. Since by construction, $W \in \mathfrak{g}$ is $\mathrm{Ad}(A)$-stable, we compute for every $a\in A$ and $X\in\mathfrak{m}$, \begin{align*}
\tau_{-\varepsilon W}\bigl(\mathrm{Ad}(a)X\bigr) = \mathrm{Ad}(a)X-\varepsilon W = \mathrm{Ad}(a)\bigl(X-\varepsilon W\bigr) = \mathrm{Ad}(a)\,\tau_{-\varepsilon W}(X).
\end{align*} Thus, $\tau_{-\varepsilon W}$ is $A$-equivariant. Moreover, note that the inverse of $\tau_{-\varepsilon W}$ is simply given by  $\tau_{+\varepsilon W}:\mathfrak{m}-\varepsilon W\to\mathfrak{m}$. It induces an isomorphism of invariant polynomial rings via pullback: \begin{align*}
\tau_{+\varepsilon W}^*: S(\mathfrak{m})^A \xrightarrow{ \cong } S(\mathfrak{m}-\varepsilon W)^A, \qquad \tau_{-\varepsilon W}^*: S(\mathfrak{m}-\varepsilon W)^A \xrightarrow{\cong } S(\mathfrak{m})^A. 
\end{align*} Thus, $S(\mathfrak{m}-\varepsilon W)^A\cong S(\mathfrak{m})^A$. Indeed, $\tau_{+\varepsilon W}^*$ and $\tau_{-\varepsilon W}^*$ are inverse $A$-equivariant ring isomorphisms.

Finally, we show part (iii). For any $a\in A$, by the definition of $\widetilde{\pi}_{\mathfrak{m}}$ and $\mathrm{Ad}(A)W=W$, we have  \begin{align*}
\widetilde\pi_\mathfrak{m}\big((g,X)\cdot a\big)=\widetilde\pi_\mathfrak{m}\big(ga,\mathrm{Ad}(a^{-1})X\big) =\mathrm{Ad}(a^{-1})X-\varepsilon W =\mathrm{Ad}(a^{-1})(X-\varepsilon W).
\end{align*}  Thus, $\widetilde\pi_\mathfrak{m}$ is right $A$-equivariant. If $\theta \in S(\mathfrak{m}-\varepsilon W)^A$, then $\theta \circ \widetilde{\pi}_\mathfrak{m}$ is $A$-invariant and thus descends to the quotient $(G \times \mathfrak{m})/A \cong T^*M$, which is exactly $\pi_\mathfrak{m}^*\theta$. The left $G$-invariance of $\pi_\mathfrak{m}^*\theta$ is immediate from the trivialization.
\end{proof}

\begin{remark}
\label{re:property}
  The composition $\theta \circ \widetilde{\pi}_\mathfrak{m}$ is $A$-invariant and descends to a function $\pi_\mathfrak{m}^*\theta([g,X])$ in $T^*M$. Thus, $\pi_\mathfrak{m}^*\theta$, rather than $\pi_\mathfrak{m}$ itself, is the globally defined object used below. 

\end{remark}

 Let us now recall some basic notation and terminology discussed in \cite{MR2141306}, which help us process the explicit coordinate computation.  
 With the standard $G$-invariant metric on $M$ induced by $B$ and the magnetic term $\varepsilon\,\omega_{\mathrm{KKS}}(g,X)(gv_1,gv_2) $ $=\varepsilon B(W,[v_1,v_2])$, under the local coordinates on $T_{(g,X)}^*T^*M$ defined in \eqref{eq:coordinate01}, the twisted symplectic form at $(g,X)$ is \begin{align}
 \nonumber
\omega_{\varepsilon}\left(g_*\vert_X\left(v_1,-\frac{1}{2}[v_1,X]+w_1\right),\ g_*\vert_X\left(v_2,-\frac{1}{2}[v_2,X]+w_2\right)\right)  \\
= B(w_1,v_2)-B(w_2,v_1)+\varepsilon\,B(W,[v_1,v_2]). \label{eq:magform}
\end{align} Note that the magnetic symplectic form has the $G$-invariant expression in \eqref{eq:magform}. Let $\xi = X - \varepsilon W$. The differential of the moment map $P(g,X)=\mathrm{Ad}(g)\xi$ at $(g,X)$, applied to a tangent vector $(v,\delta X)$, produces the following: \begin{align}
  d P_{(g,X)}(v,\delta X)=  \left. \dfrac{d}{dt} \right\vert_{t= 0} P(g \exp(tv),X +  t \delta X) = \mathrm{Ad}(g)([v,\xi] +  \delta X) \label{eq:diffmoment}
\end{align}   since $P((g \exp(tv),X))=\mathrm{Ad}(g)\bigl(\exp(t( \mathrm{ad}v)\xi)\bigr)$ (with $X$ fixed) and $P([g,X+t\,\delta X])=\mathrm{Ad}(g)(\xi+t\,\delta X)$ (with $g$ fixed). Note that $d P \in T\mathfrak{g}^* \cong \mathfrak{g}^* \cong \mathfrak{g}$ as $v \in \mathfrak{g}$ such that $[v,\xi] \in \mathfrak{g}$ and $\delta X \in \mathfrak{m}$, hence $[v,\xi] +  \delta X \in \mathfrak{g}$. On the other hand, for any $\theta \in S(\mathfrak{m} - \varepsilon W)^A$, consider the tangent $g_*\vert_X\left(v, w - \frac{1}{2}[v,X]\right)$ used in \eqref{eq:magform}. Similar to what we did in \eqref{eq:diffmoment}, simple computation shows that, on the representative space $G\times\mathfrak{m}$, the differential at $(g,X)$ is \begin{align*}
    d (\widetilde{\pi}_\mathfrak{m})_{(g,X)} \left(g_*\vert_X\left(v, w - \frac{1}{2}[v,X]\right)\right) =   w - \frac{1}{2}[v,X] \in T_\xi(\mathfrak{m} - \varepsilon W) \cong \mathfrak{m}.
\end{align*}   Since $\theta$ is $A$-invariant and the $B$-gradient $\nabla \theta(\xi) \in \mathfrak{m}$ is characterized by  $d \theta(\xi) \left( \left(v, w - \frac{1}{2}[v,X]\right)\right) = B(\nabla \theta(\xi),w - \frac{1}{2}[v,X] )$, we then have \begin{align}
    d (\theta \circ \widetilde{\pi}_\mathfrak{m})_{(g,X)}\left(g_*\vert_X \left(v, w - \frac{1}{2}[v,X]\right)\right)  = B\left(\nabla \theta(\xi),w - \frac{1}{2}[v,X] \right). \label{eq:diffproj}
\end{align} In the following lemma, we will provide the Hamiltonian vector fields of the moment map $P$ and, therefore, present the twisted Poisson bracket relations of the linear first integrals.

 \begin{lemma}
 \label{lem:F1polypbrack} \cite[Proposition 3]{MR2141306}
Let $G $ and $A\subset G$  be the same as defined above. Let $P:T^*M\longrightarrow\mathfrak{g}^*$ be defined by $P(g,X)=\mathrm{Ad}(g)\bigl(X-\varepsilon W\bigr)$ as the magnetic moment map. For $\eta\in\mathfrak{g}$, define \begin{align}
P_\eta(g,X):=\langle P(g,X),\eta \rangle   . \label{eq:comoment}
\end{align}   Then \begin{equation}\label{eq:f1bracket}
\{P_\eta,P_{\eta'}\}_\varepsilon = P_{[\eta,{\eta'}]}\quad \text{ with }\eta,\eta'\in\mathfrak{g} .
\end{equation}  
\end{lemma}

\begin{proof}
Write $\xi:=X-\varepsilon W$ and left-trivialize $T^*M\cong  G\times_A\mathfrak{m}$.  From \eqref{eq:comoment}, we have\begin{align*}
    P_\eta(g,X)=\langle P(g,X),\eta\rangle=B(P^\sharp(g,X),\eta),\quad \text{ with }
P^\sharp(g,X):=\mathrm{Ad}(g)\,\xi .
\end{align*}   Using \eqref{eq:diffmoment}, differentiating $P_\eta(g,X)$ in the direction $(v,\delta X)$ gives: \begin{align}
dP_\eta \left(g_*\vert_X\left(v,w-\frac{1}{2}[v,X] \right)\right) = & \, B\left(\nabla P_\eta(g,X),\mathrm{Ad}(g)\left(\left[v,\xi + w - \frac{1}{2}[v,X] \right]\right)\right) \label{eq:dPxi}\\
\nonumber
= & \, B\left(w, \mathrm{Ad}(g^{-1})\eta)\right) - B\left(v,[\mathrm{Ad}(g^{-1})\eta,\zeta]\right).
\end{align} By definition, the Hamiltonian vector field $X_{P_\eta}$ is characterized by $\iota_{X_{P_\eta}}\omega_\varepsilon=dP_\eta$.  The explicit expression of  $X_{P_\eta}$ at $(g,X)$ is of the form $g_*\vert_X\bigl(v_\eta,-\frac{1}{2}[v_\eta,X]+w_\eta\bigr)$.  Using \eqref{eq:magform}, without loss of generality, we find that for any $(0,w)\in\mathfrak{m}\times\mathfrak{m}$, $\omega_\varepsilon(X_{P_\eta},g_*\vert_X(0,w)) = dP_\eta (g_*\vert_X(0,w))$. Together with \eqref{eq:dPxi}, we deduce \begin{align*}
 -B(w,v_\eta)=B\left(w, \mathrm{Ad}(g^{-1})\eta)\right)\Longrightarrow v_\eta= -\bigl(\mathrm{Ad}(g^{-1})\eta\bigr)_{\mathfrak{m}},
\end{align*} where $(\cdot)_\mathfrak{m}: \mathfrak{g} \to \mathfrak{m}$ is the $B$-orthogonal projection onto $\mathfrak{m}$. Moreover, the coordinate form of $\omega_\varepsilon$ in \eqref{eq:dPxi} implies that \begin{align}
-B(w_\eta,v)+\varepsilon\,B \big(W,[v_\eta,v]\big)=-B \big(v,[\mathrm{Ad}(g^{-1})\eta,\xi]\big) . \label{eq:balance}
\end{align}  $\mathrm{Ad}$-invariance of $B$ gives $B(W,[v_\eta,v])=-B([W,v_\eta],v)$. Then the identity \eqref{eq:balance} becomes \begin{align*} 
B \big(w_\eta+[\varepsilon W,v_\eta]-[\mathrm{Ad}(g^{-1})\eta,\xi],\,v\big)=0\quad \text{ for all } v\in\mathfrak{m},
\end{align*}which yields, after substituting $v_\eta=-(\mathrm{Ad}(g^{-1})\eta)_{\mathfrak{m}}$ and $\xi=X-\varepsilon W$, \begin{align*}
w_\eta= & \, [\mathrm{Ad}(g^{-1})\eta,\xi]_{\mathfrak{m}} + \varepsilon[W,(\eta)_{\mathfrak{m}}] \\
= & \, [\mathrm{Ad}(g^{-1})\eta,X]_{\mathfrak{m}}-\varepsilon [\mathrm{Ad}(g^{-1})\eta,W]_\mathfrak{m} +  \varepsilon \big[(\mathrm{Ad}(g^{-1})\eta)_{\mathfrak{m}},X\big]_{\mathfrak{m}}  \\
= & \, [\mathrm{Ad}(g^{-1})\eta,X]_\mathfrak{m}.
\end{align*} Note that the last equality is deduced by observing $[\left(\mathrm{Ad}(g^{-1})\eta\right)_\mathfrak{a},W] = 0 $ and $[W,\left(\mathrm{Ad}(g^{-1})\eta\right)_\mathfrak{m}] = -[\left(\mathrm{Ad}(g^{-1})\eta\right)_\mathfrak{m},W]  $ Thus, \begin{align}
X_{P_\eta}(g,X) = g_*\vert_X\Big(-(\mathrm{Ad}(g^{-1})\eta)_{\mathfrak{m}},\, [\mathrm{Ad}(g^{-1})\eta,X]_{\mathfrak{m}}-\frac{1}{2}\big[(\mathrm{Ad}(g^{-1})\eta)_{\mathfrak{m}},X\big]_{\mathfrak{m}}\Big),\label{eq:coovectfie}
\end{align}
which is exactly the fundamental vector field of the left $G$-action for the magnetic form $\omega_\varepsilon$.  Using \eqref{eq:coovectfie}, the twisted Poisson bracket is \begin{align*}
\{P_\eta,P_\eta\}_\varepsilon(g,X)= dP_\eta \big(X_{P_\eta}(g,X)\big)= \left.\frac{d}{dt}\right\vert_{t=0} P_\eta \big(\Phi^{\eta}_t(g,X)\big),
\end{align*}
where $\Phi^{\eta}_t$ is the flow of $X_{P_\eta}$, i.e., the left action by $\exp(t\eta)$. Along this flow,
\begin{align*}
P^\sharp\bigl(\Phi_t^\eta(g,X)\bigr)=\mathrm{Ad}(\exp(-t\eta))P^\sharp(g,X),
\qquad
P\bigl(\Phi_t^\eta(g,X)\bigr)=\mathrm{Ad}^*(\exp(t\eta))P(g,X).
\end{align*} So,
\begin{align*}
\left.\frac{d}{dt}\right\vert_{t=0}P_{\eta'}\bigl(\Phi_t^\eta(g,X)\bigr)
=& \, \left.\frac{d}{dt}\right\vert_{t=0}\langle \mathrm{Ad}^*(\exp(t\eta))P(g,X),\eta'\rangle = \left.\frac{d}{dt}\right\vert_{t=0}\langle P(g,X),\mathrm{Ad}(\exp(t\eta))(\eta')\rangle\\
= & \,  \langle P(g,X),\,[\eta,\eta'] \rangle
=  P_{[\eta,\eta']}(g,X),
\end{align*}
where the last equality uses $\mathrm{Ad}$-invariance and the definition of $P_\eta$. This proves \eqref{eq:f1bracket}.   
\end{proof}

\begin{remark} 
\label{rem:F1poly}
(i)  Let $\mathcal{O}(T^*M) \subset C^\infty(T^*M)$ denote the real algebraic polynomial functions on $T^*M$. That is, $\mathcal{O}(T^*M): = \left(\mathbb{R}[G] \otimes S(\mathfrak{m})\right)^A \cong \mathbb{R}[G \times \mathfrak{m}]^A$.  Let $P:T^*M\to\mathfrak{g}^*$ be the left moment map defined in Definition \ref{def:maps}. For any $h\in S(\mathfrak{g} )$ with $I_h:=h\circ P$, the pullback  $P^*:S(\mathfrak{g} )\to \mathcal{O}(T^*M)$ is a Poisson algebra homomorphism: \begin{align*}
    I_{\{h_1,h_2\}} =  \{I_{h_1},I_{h_2}\}_\varepsilon.
\end{align*}  In particular, for linear functions $h_\eta(X)=B(\eta,X)$, one deduces from \cite[Proposition 3]{MR2141306} that
\begin{align*}
    \{P_\eta,P_{\eta'}\}_\varepsilon =  P_{[\eta,\eta']}, \qquad P_\eta:=h_\eta\circ P.
\end{align*}  

(ii) From \cite{MR2141306}, the Hamiltonian vector fields are $ X_{P_\eta}(g,X)$ in \eqref{eq:coovectfie} and \begin{align}
  X_{f_\theta}(g,X)  =   g_*\vert_X \left((\nabla \theta(\xi))_\mathfrak{m}, - \frac{1}{2} [(\nabla \theta(\xi))_\mathfrak{m},X]_\mathfrak{m} - \varepsilon [W,(\nabla \theta(\xi))_\mathfrak{m}]\right).
\end{align} 
\end{remark}

 In what follows, let $\xi := X - \varepsilon W \in \mathfrak{m} - \varepsilon W$. By construction, there exists a $W \in \mathfrak{g}$ that generates the adjoint orbit $M = \mathcal{O}_W = \mathrm{Ad}(G)\cdot W$. 

We first clarify two polynomial families given in \cite{MR2141306}, denoted by $F_1^{\mathrm{poly}}$ and $F_2^{\mathrm{poly}}$. The first family consists of polynomial functions on $T^*M$ obtained by pulling back polynomials on $\mathfrak{g}$ through the (magnetic) moment map $P$:
\begin{align}
F_1^{\mathrm{poly}}:=\{h\circ P: h\in S(\mathfrak{g})\}.
\end{align}

The second family comes from evaluating $A$-invariant polynomials on the affine shift $X-\varepsilon W\in\mathfrak{m}-\varepsilon W$. For each $\theta\in S(\mathfrak{m}-\varepsilon W)^A$, define
\begin{align}
    (\pi_{\mathfrak{m}}^*\theta)([g,X]) := \theta\bigl(X-\varepsilon W\bigr). \label{eq:ainv}
\end{align}
By Lemma \ref{lem:property}(iii), the function $\theta(X-\varepsilon W)$ is $A$-invariant on $G\times\mathfrak{m}$, hence it descends to $T^*M\cong (G\times\mathfrak{m})/A$ by Lemma \ref{lem:descent}. We then set
\begin{align}
F_2^{\mathrm{poly}}:=\left\{    \pi_\mathfrak{m}^*\theta :  \theta \in S(\mathfrak{m} -\varepsilon W)^A \right\}. 
\end{align}

Hence, based on the discussion in the previous discussion, we summarise all the results above into the following definition.
\begin{definition}
\label{def:families}
The sets of polynomial generators that Poisson commutes with $\mathcal{H}$  are given by: \begin{align}
F_1^{\mathrm{poly}}  := \left\{ f\circ P : f \in S(\mathfrak{g}) \right\}, \quad  F_2^{\mathrm{poly}}  := \left\{  \pi_\mathfrak{m}^*\theta : \theta \in S(\mathfrak{m} -\varepsilon W)^A \right\}. \label{eq:intm}
\end{align}    
\end{definition}


 In what follows, we first prove that the sets $F_1^{\mathrm{poly}}$ and $F_2^\mathrm{poly}$ are preserved by the magnetic geodesic flow. After establishing conservation, we return to their algebraic realization as pullbacks and to a clear notation for the resulting polynomial Poisson algebras.  Now, we denote the finitely generated algebras  \begin{align*}
    \mathfrak{F}_1^{\mathrm{poly}} : = \textbf{Alg} \left\langle F_1^{\mathrm{poly}}\right\rangle \text{ and }   \mathfrak{F}_2^{\mathrm{poly}} : = \textbf{Alg} \left\langle F_2^{\mathrm{poly}}\right\rangle .
\end{align*} 
It is clear that they are subalgebras of $\mathcal{O}(T^*M)$. 
We now show that the polynomials from the $P$ and $\pi_{\mathfrak{m}}^*$-pullbacks are conserved with respect to the Hamiltonian in \eqref{eq:Hamiltoneq}.


\begin{proposition} \label{prop:leftintegrals}
Let $\xi = X -\varepsilon W \in \mathfrak{m} -\varepsilon W$. For every $f\in S(\mathfrak{g})$, $ \{P^*f,  \mathcal{H}\}_\varepsilon = 0$. 
\end{proposition}

\begin{proof}
For any $\eta \in \mathfrak{g}$, by the definition of the Hamiltonian action, we have $\iota_{\hat{\eta}} \omega_\varepsilon = d \langle P,\eta \rangle$. Using $P(g,X) = \mathrm{Ad}(g)\xi \in \mathfrak{g}^*$, for all $(g,X) \in T^*M$, we rewrite $\langle P,\eta \rangle = B(\mathrm{Ad}(g)(\xi),\eta) = B(\xi, \mathrm{Ad}(g^{-1})\eta) $. Thus, $P_\eta = \langle P,\eta \rangle$ generates the fundamental vector field $\hat{\eta}$. Since the left $G$-action on $T^*M$ is Hamiltonian with the magnetic moment map $P$ (under $B$), and $\mathcal{H}(g,X)=\frac{1}{2} B(X,X)$ is left-invariant, we have $$\mathcal{L}_{\hat{\eta}}\mathcal{H} = \{d \langle P,\eta \rangle, \mathcal{H}\}_\varepsilon =  0.$$   By linearity and Leibniz's rule, every polynomial $f\in S(\mathfrak{g})$ yields a conserved function $P^*f$. So, $P^*f$ is indeed the first integral of $\mathcal{H}$.
\end{proof}

\begin{remark}
\label{rem:f1}
By Proposition \ref{prop:leftintegrals}, $\{P^*f,P^*h\}_\varepsilon=P^*\{f,h\}_{\mathfrak{g}}$. Hence, $\mathfrak{F}_1^{\mathrm{poly}}$ is not Abelian Poisson algebra unless $\mathfrak{g}$ is Abelian.
\end{remark}

\begin{lemma} 
\label{lem:lax}
Let $(g(t),X(t))$ be the solution of the Hamilton equations \eqref{eq:Hamiltoneq} for $\mathcal{H}$  with respect to $\omega_{\varepsilon}$. Let $\xi(t) = X(t) -\varepsilon W$ such that $\pi_\mathfrak{m}(t) = \xi(t)$. Then there exists a curve $\Omega(t) := \varepsilon W \in \mathfrak{g}$ such that $\pi_\mathfrak{m}(t) = X(t) -\varepsilon W$ satisfies \begin{align*}
\dot{\pi}_{\mathfrak{m}}(t)= [\pi_{\mathfrak{m}}(t),  \Omega(t)],
\end{align*} where $\dot{f} (t) = \dfrac{d}{dt} f(t)$ for any smooth function $f$. In particular, $\dot{\pi}_{\mathfrak{m}}(t)\in T_{\pi_{\mathfrak{m}}(t)}\big(\mathrm{Ad}(A)\cdot \pi_{\mathfrak{m}}(t)\big)$.
\end{lemma}
\begin{proof}
By Definition \ref{def:maps}, along any trajectory we have $\xi(t)=X(t)-\varepsilon W$, hence $\dot\xi(t)=\dot X(t)$ since $W$ is fixed. To compute $\dot{X}$, we use the explicit form of the magnetic symplectic structure at $(g,X)\in T^*M$. Let $X_\mathcal{H}$ denote the Hamiltonian vector field of $\mathcal{H}$. Recall that the local coordinate expression is given by \begin{align}
X_\mathcal{H}(g,X)=g_*\vert_X\left(v_\mathcal{H},\,-\frac{1}{2}[v_\mathcal{H},X]+w_\mathcal{H}\right),\quad v_\mathcal{H},w_\mathcal{H}\in\mathfrak{m}. \label{eq:hamilvec}
\end{align} Here $v_\mathcal{H} =v_\mathcal{H}(g,X)$ and $w_\mathcal{H} = w_\mathcal{H}(g,X)$ are $\mathfrak{m}$-components of $X_\mathcal{H}$ in the trivialization, which need to be determined from $d \mathcal{H} = \iota_{X_\mathcal{H}}\omega_\varepsilon$. Then, with the magnetic symplectic form in \eqref{eq:magform}, the magnetic Hamilton’s equation $d \mathcal{H} =\iota_{X_\mathcal{H}}\omega_\varepsilon$ evaluated at an arbitrary point $g_*\vert_X\bigl(v,-\frac{1}{2}[v,X]+w\bigr) \in T^*(T^*M)$ gives
\begin{align}
B(w,v_\mathcal{H})-B(w_\mathcal{H},v)+\varepsilon\,B\left(W,\,[v,v_\mathcal{H}]\right)\,=\,d\mathcal{H}(g,X)\left(g_*\vert_X\bigl(v,-\frac{1}{2}[v,X]+w\bigr)\right). \label{eq:hamiltonian01}
\end{align} Since $\mathcal{H}(g,X)=\frac{1}{2} B(X,X)$ depends only on $X$ and is quadratic, \eqref{eq:hamilvec} becomes   \begin{align*}
d\mathcal{H}(g,X)\left(g_*\vert_X\bigl(v,-\frac{1}{2}[v,X]+w\bigr)\right)=B(X,w).
\end{align*} Therefore, for all $v,w\in\mathfrak{m}$, comparing the identity of $d\mathcal{H}$ gives \begin{align}
B(w,v_\mathcal{H})&=B(X,w), \label{eq:b1}\\
-B(w_\mathcal{H},v)&=-\varepsilon\,B\left(W,\,[v,v_\mathcal{H}]\right). \label{eq:b2}
\end{align} From \eqref{eq:b1}, we obtain $v_\mathcal{H}=X$. Substituting into \eqref{eq:b2} and using $B$-invariance $B(W,[v,X]) = -B(v,[W,X]) = B(v,[X,W])$, \begin{align*}
-B(w_\mathcal{H},v)+\varepsilon\,B\left(W,\,[v,X]\right) &= -B(w_\mathcal{H},v)+\varepsilon\,B\left(v,\,[X,W]\right) \,=\,0\qquad \text{ for all } v\in\mathfrak{m}.
\end{align*} Hence,  $w_\mathcal{H}=\varepsilon[X,W]=-\varepsilon[W,X]\in\mathfrak{m}$.  Consequently, with all the discussion above, \eqref{eq:hamilvec} becomes $X_\mathcal{H}(g,X)=g_*\vert_X\bigl(X,\,-\varepsilon[W,X]\bigr)$, and the Hamilton equations are \begin{align}
\dot{g}=gX,\quad \dot{X}=-\varepsilon[W,X]. \label{eq:hamilton}
\end{align}

Finally, using \eqref{eq:hamilton}
\begin{align*}
\dot{\pi}_\mathfrak{m} =\dot{X}=-\varepsilon[W,X]=[X-\varepsilon W,\varepsilon W]=[\pi_\mathfrak{m}, \Omega],
\end{align*} with the (constant) choice $\Omega(t)\equiv \varepsilon W \in \mathfrak{g}$. This shows both the bracket form and that $\dot{\pi}_m$ is tangent to the $A$-orbit through $\pi_\mathfrak{m}$ as required.
\end{proof}

\begin{proposition}
\label{prop:rightintegrals}
For every $\theta\in S(\mathfrak{m}-\varepsilon W)^A$,  $ \{ \pi_{\mathfrak{m}}^*\theta, \mathcal{H}\}_\varepsilon = 0$. 
\end{proposition}

\begin{proof}
Since $\theta$ is $A$-invariant on the affine space $\mathfrak{m}-\varepsilon W$, its differential annihilates the tangent vectors to $A$-orbits. That is, for all $\xi$ and $\eta\in\mathfrak{a}$, differentiating $\theta$ along the flow gives \begin{align*}
0=\left.\frac{d}{dt}\right\vert_{t =0} \theta\big(\mathrm{Ad}(\exp(t\eta))\xi\big) = B\big(\nabla \theta(\xi),[\eta,\xi]\big)=B\big([\nabla \theta(\xi),\xi],\eta\big) 
\end{align*} with $\xi=\pi_{\mathfrak{m}}(0)$ and $\dot\xi=[\xi,\Omega]$. By Lemma \ref{lem:lax}, we obtain \begin{align*}
\left.\frac{d}{dt}\right\vert_{t= 0}\theta\big(\pi_{\mathfrak{m}}(t)\big) = B\big(\nabla \theta(\xi),[\xi,\Omega]\big)=B\big([\nabla \theta(\xi),\xi],\Omega\big)=0.
\end{align*} Hence, $\theta\circ\pi_{\mathfrak{m}}$ is constant along the flow. Equivalently, $\{\pi_{\mathfrak{m}}^*\theta,\mathcal{H}\}_\varepsilon=0$.
\end{proof}

\begin{remark}
Similar to Remark \ref{rem:f1}, we do \emph{not} claim  $\{\pi_{\mathfrak{m}}^*h_1,\pi_{\mathfrak{m}}^*h_2\}_\varepsilon=0$ for all $h_i$ with $i = 1,2$. In general, $\mathfrak{F}_2^{\mathrm{poly}}$ is not Poisson-commutative.
\end{remark}

To study the Poisson projection chain description of the algebra generated by these finitely-generating sets $F_1^{\mathrm{poly}}$ and $F_2^{\mathrm{poly}}$, we need to first show that $\mathfrak{F}_1^{\mathrm{poly}}$ and $\mathfrak{F}_2^{\mathrm{poly}}$ are Poisson algebras, and then describe the spectrum of these two Poisson algebras. Starting with the following lemma:

\begin{lemma}
\label{lem:poly}
  Let $F_1^{\mathrm{poly}}$ and $F_2^{\mathrm{poly}}$ be the sets defined in \eqref{eq:intm}. Then the polynomials by the pullbacks of $P$ and the descended invariant pullback $\pi_\mathfrak{m}^*$ $\pi_\mathfrak{m}$ induce the following subalgebras of $\mathcal{O}( (T^*M)) $: \begin{align*}
      \mathfrak{F}_1^{\mathrm{poly}}=P^*(S(\mathfrak{g})),\quad \mathfrak{F}_2^{\mathrm{poly}}=\pi_{\mathfrak{m}}^*\bigl(S(\mathfrak{m}-\varepsilon W)^A\bigr) ,
  \end{align*}  Here $\pi_{\mathfrak{m}}^*\theta$ is characterised by $q^*(\pi_{\mathfrak{m}}^*\theta) := \theta \circ \pi_{\mathfrak{m}} $. In particular, the pullback of $P$ is surjective, and the pullback of $\pi_\mathfrak{m}$ is injective.
\end{lemma}
\begin{proof}
 For $(g,X)\in T^*M \cong  (G\times\mathfrak{m})/A$, by Definition \ref{def:maps}, \begin{align*}
  \xi:=X-\varepsilon W\in \mathfrak{m}-\varepsilon W,  \quad P(g,X):=\mathrm{Ad}(g)\,\xi, \quad \pi_{\mathfrak{m}}(g,X):=\xi.
\end{align*} We now provide the construction separately.  
 
By the definition of $F_1^{\mathrm{poly}}$ in \eqref{eq:intm}, the generators in $F_1^{\mathrm{poly}}$ are of the form $(h\circ P)(g,X)=h\bigl(P(g,X)\bigr)$ with $h \in S(\mathfrak{g})$. Equivalently, $\mathfrak{F}_1^{\mathrm{poly}}: = \textbf{Alg} \left\langle F_1^{\mathrm{poly}}\right\rangle$ is the image of the algebra homomorphism given by the pullback $P^*$ as follows: \begin{align*}
  P^*:S(\mathfrak{g})\longrightarrow \mathcal{O}(T^*M),\quad h\longmapsto h\circ P.  
\end{align*} A direct calculation shows that $P^*$ preserves sums and products, its image is a subalgebra of $\mathcal{O}(T^*M)$, and $ \mathfrak{F}_1^{\mathrm{poly}} =P^*\bigl(S(\mathfrak{g})\bigr)$.  On the other hand, by the definition of $F_2^{\mathrm{poly}}$ in \eqref{eq:intm} and \eqref{eq:ainv}, every element of $F_2^{\mathrm{poly}}$ is of the form $ \theta ( \pi_\mathfrak{m}(g,X)) = \theta(X - \varepsilon X) $ with $ \theta\in S(\mathfrak{m}-\varepsilon W)^A$.  Thus, similar to the construction of $\mathfrak{F}_1^{\mathrm{poly}}$, we find that $\mathfrak{F}_2^{\mathrm{poly}}: = \textbf{Alg} \left\langle F_2^{\mathrm{poly}}\right\rangle$ is the image of the algebra homomorphism \begin{align*}
\pi_{\mathfrak{m}}^*:S(\mathfrak{m}-\varepsilon W)^A\longrightarrow \mathcal{O}(T^*M),\qquad \theta\longmapsto \pi_{\mathfrak{m}}^*\theta,\quad q^*(\pi_{\mathfrak{m}}^*\theta)=\theta\circ\widetilde{\pi}_{\mathfrak{m}},
\end{align*} which again preserves sums and products. Hence, its image is a subalgebra and $  \mathfrak{F}_2^{\mathrm{poly}} =\pi_{\mathfrak{m}}^*\bigl(S(\mathfrak{m}-\varepsilon W)^A\bigr)$. We then show that $\pi_\mathfrak{m}^*$ is injective. For any $\theta \in S(\mathfrak{m}-\varepsilon W)^A$, assume that $ \pi_\mathfrak{m}^*\theta = 0$. Then pulling back by $q$ gives $\theta(\pi_\mathfrak{m}(g,X)) = \theta(X - \varepsilon W) = 0$ on $G \times \mathfrak{m}$. Since $\widetilde{\pi}_\mathfrak{m}$ is surjective onto $\mathfrak{m} - \varepsilon W$, we find $\theta  = 0$, which implies that $\ker \pi_\mathfrak{m}^* = 0$. Therefore, $\pi_\mathfrak{m}^*$ is injective.
\end{proof}

 \begin{remark}
 \label{rem:ideal}
     (i) Note that this isomorphism is canonical, as $\mathcal{O}(T^*M)$ is the Poisson subalgebra.  For any real algebraic morphism $P : T^*M \rightarrow \mathfrak{g}^*$, the kernel of the pullback $P^*: S(\mathfrak{g}) \cong \mathbb{R}[x_1,\ldots,x_n] \twoheadrightarrow \mathcal{O}(T^*M)$ is the vanishing ideal of the real Zariski closure of the image: $\mathcal{I}(\overline{P((g,X))})$ for all $(g,X) \in T^*M$, where \begin{align*}
         \mathcal{I}(\overline{P((g,X))}) := \left\{ f\in S(\mathfrak{g}): f (P(g,X)) = 0,  \text{ for all } (g,X) \in T^*M \right\} = \ker (P^*)
     \end{align*} Hence, by the first isomorphism theorem, the image algebra is given by the following isomorphism \begin{align}
          \mathrm{Im} (P^*) \cong \frac{S(\mathfrak{g})}{\mathcal{I}(\overline{P(g,X)})} \cong \mathbb{R}[\overline{\mathrm{Ad}(G)(\mathfrak{m} - \varepsilon W)}] := S[\overline{\mathrm{Ad}(G)(\mathfrak{m} - \varepsilon W)}]. \label{eq:f1}
      \end{align} Here, $S[\overline{\mathrm{Ad}(G)(\mathfrak{m} - \varepsilon W)}]$ is the coordinate ring of the affine variety $\mathrm{Ad}(G)(\mathfrak{m} - \varepsilon W)$. Thus, $\mathfrak{F}_1^{\mathrm{poly}}: = \mathrm{Im} (P^*)$ is canonically Poisson isomorphic to the real coordinate ring of the closed subvariety $\overline{\mathrm{Ad}(G)(\mathfrak{m} - \varepsilon W)} \subset \mathfrak{g}$. For the rest of this manuscript, for abbreviation, we write $\mathbb{R}[ \mathrm{Ad}(G)(\mathfrak{m} - \varepsilon W)]$ to mean $\mathbb{R}[\overline{\mathrm{Ad}(G)(\mathfrak{m} - \varepsilon W)}]$.

      (ii) Although $\mathfrak{m}$ is a vector space, in general, it is not a Lie subalgebra of $\mathfrak{g}$. Therefore, $S(\mathfrak{m})$ does not admit a Lie-Poisson structure induced from $S(\mathfrak{g})$. In order to define a Poisson structure on it, we transport the Poisson bracket from $\pi_\mathfrak{m}^* S(\mathfrak{m}-\varepsilon W)^A\subset \mathcal{O}(T^* M)$ to $S(\mathfrak{m}-\varepsilon W)^A$ via the injective pullback $\pi_\mathfrak{m}^*$, and denote the resulting bracket by $\{\cdot,\cdot\}_2$ (see Proposition \ref{prop:slicepoisson} below).
 \end{remark}

We now show that $\mathfrak{F}_1^{\mathrm{poly}}$ and $\mathfrak{F}_2^{\mathrm{poly}}$ defined in Lemma \ref{lem:poly} are Poisson subalgebras of $\mathcal{O}(T^*M)$.
 
 \begin{proposition}
 \label{prop:slicepoisson}
Let $h_1,h_2 \in S(\mathfrak{g})$ and $\theta_1,\theta_2\in S(\mathfrak{m}-\varepsilon W)^A$. Then 
\begin{align*}
\{P^*h_1,P^*h_2\}_\varepsilon = P^*\left(\{h_1,h_2\}_1\right) \text{ and } \{\pi_\mathfrak{m}^*\theta_1,\pi_\mathfrak{m}^*\theta_2\}_\varepsilon=\pi_\mathfrak{m}^*\bigl(\{\theta_1,\theta_2\}_2\bigr),
\end{align*} where $\{\cdot,\cdot\}_1 = \{\cdot,\cdot\} $ is a Lie-Poisson bracket and $\{\cdot,\cdot\}_2$ is the Poisson bracket on $S(\mathfrak{m}-\varepsilon W)^A$ given by \begin{align*}
\{\theta_1,\theta_2\}_2(\xi) :=-B \left(\xi,\bigl[(\nabla\theta_1(\xi))_{\mathfrak{m}},(\nabla\theta_2(\xi))_{\mathfrak{m}}\bigr]\right),
\qquad \xi\in \mathfrak{m}-\varepsilon W .
\end{align*}
In particular, $\mathfrak{F}_2^{\mathrm{poly}}=\pi_\mathfrak{m}^* S(\mathfrak{m}-\varepsilon W)^A$
is a Poisson subalgebra of $\mathcal{O}(T^* M)$.
\end{proposition}

\begin{proof}
Fix $(g,X)\in T^* M$ and put $\xi:=X-\varepsilon W$. For $i=1,2$, write $v_i:=(\nabla\theta_i(\xi))_{\mathfrak{m}}$ and let $X_{\pi_\mathfrak{m}^*\theta_i}$ be as in \eqref{eq:ainv}. Using the $G$-invariant expression of $\omega_\varepsilon$ in \eqref{eq:magform}, we obtain \begin{align*}
\{\pi_\mathfrak{m}^*\theta_1,\pi_\mathfrak{m}^*\theta_2\}_\varepsilon(g,X) = & \, \omega_\varepsilon \bigl(X_{\pi_\mathfrak{m}^*\theta_1},X_{\pi_\mathfrak{m}^*\theta_2}\bigr) \\
= & \,-B(X,[v_1,v_2])+\varepsilon B(W,[v_1,v_2]) \\
= &\,-B(X-\varepsilon W,[v_1,v_2]) =-B \left(\xi,\bigl[v_1,v_2\bigr]\right),
\end{align*}  which depends only on $\xi$ and is $A$-invariant because each $\theta_i$ is $A$-invariant. Hence, it lies in $\pi_\mathfrak{m}^* S(\mathfrak{m}-\varepsilon W)^A$. This proves the formula and closure.
\end{proof}



We now have a look at the Poisson commutative relations between $\mathfrak{F}_1^{\mathrm{poly}}$ and $\mathfrak{F}_2^{\mathrm{poly}}$. As all generators $\pi_\mathfrak{m}^*\theta \in \mathfrak{F}_2^{\mathrm{poly}}$ are left $G$-invariant, we deduce the following lemma. 

\begin{lemma}
\label{lem:interzero}
  For $f\in S(\mathfrak{g})$ and $\theta\in S(\mathfrak{m}-\varepsilon W)^A$, we have \begin{align*}
\{ P^*f,\ \pi_{\mathfrak{m}}^*\theta \}_\varepsilon=0 .
\end{align*}  That is, $\left\{\mathfrak{F}_1^{\mathrm{poly}},\mathfrak{F}_2^{\mathrm{poly}}\right\}_\varepsilon = 0$.
\end{lemma} 

\begin{proof}
 Let $\langle \cdot,\cdot \rangle$ be the canonical pairing $\mathfrak{g}^* \times \mathfrak{g} \rightarrow \mathbb{R}$. For any $u \in \mathfrak{g}$, define the linear functional on $\mathfrak{g}^*$ by $u\mapsto \ell_u$, with $\ell_u(\eta):=\langle \eta,u\rangle $. Since $S(\mathfrak{g}) \cong \mathbb{R}[\mathfrak{g}^*]$, for any $f\in S(\mathfrak{g})$, we have \begin{align*}
f(\eta) = F\big(\ell_{u_1}(\eta),\ldots,\ell_{u_k}(\eta)\big),\qquad \eta \in\mathfrak{g}^*.
\end{align*} For $u\in\mathfrak{g}$, from \eqref{eq:momcoord}, \begin{align*}
P_u:=\ell_u\circ P=\langle P, u\rangle \in C^\infty(T^*M).
\end{align*} Moreover, since $P:T^*M\to \mathfrak{g}^*$ is polynomial in the fiber variables under the left trivialization, each $P_u$ is a polynomial function on $T^*M$, i.e. $P_u\in \mathfrak{F}_1^{\mathrm{poly}}=P^*S(\mathfrak{g})$.
 The pullback of $P$ gives \begin{align}
P^*f = f\circ P = F\big(P_{u_1},\ldots,P_{u_k}\big).
\label{eq:comp}
\end{align} 

We first show that the linear components in \eqref{eq:comp} satisfy the argument. By Lemma \ref{lem:property}, $\pi_{\mathfrak{m}}^*\theta$ is $G$-invariant. That is, $(\pi_{\mathfrak{m}}^*\theta)(L_a\cdot (g,X))=\theta(X)=(\pi_{\mathfrak{m}}^*\theta)((g,X))$, where $L_a$ denotes the left-translation and $(g,X) = [g,X]$. Hence $\pi_\mathfrak{m}^*\theta ( \exp(tu)\cdot (g,X)) = \pi_\mathfrak{m}^* \theta(g,X)$ for all $ t \in \mathbb{R}$. Note that here we use the conventions $\iota_{X_f} \omega_\varepsilon = df $ and $\iota_{\hat{u}} \omega_{\varepsilon} = d\langle P,u \rangle$. 
Hence, $X_{P_u} = \widehat{u}$. Therefore, for any $u\in\mathfrak{g}$, by the definition of the Lie derivative, \begin{align}
\{P_u,\pi_{\mathfrak{m}}^*\theta\}_\varepsilon =\omega_\varepsilon(X_{P_u},X_{\pi_{\mathfrak{m}}^*\theta}) =d(\pi_{\mathfrak{m}}^*\theta)(\hat{u}) =\mathcal{L}_{\hat{u}} (\pi_{\mathfrak{m}}^*\theta) 
= 0. \label{eq:lin}
\end{align} Here, in \eqref{eq:lin}, we use the Cartan formula $\mathcal{L}_{\hat{u}} = \iota_{\hat{u}}d +  d \iota_{\hat{u}}$ and the fact that $\pi_\mathfrak{m}^*\theta$ is $G$-invariant. We now focus on the higher degree components.  Note that the Poisson bracket $\{\cdot,\cdot\}_\varepsilon$ is bilinear and has a derivation in each slot. That is,
\begin{align*}
\{fg,h\}_\varepsilon=f\{g,h\}_\varepsilon+g\{f,h\}_\varepsilon,\qquad \{f,g+h\}_\varepsilon=\{f,g\}_{\varepsilon}+\{f,h\}_\varepsilon.
\end{align*}  Moreover, from the discussion in the previous paragraph, $S(\mathfrak{g})$ is generated by the linear functionals $\ell_u$ in sums and products. Since $\{P_{u_i},\pi_{\mathfrak{m}}^*\theta\}_\varepsilon=0$ for all $i$, an induction on the total degree of $F$ by using Leibniz and linearity, yields \begin{align*}
\{\,F(P_{u_1},\ldots,P_{u_k}),\ \pi_{\mathfrak{m}}^*\theta\,\}_\varepsilon=0.
\end{align*} Hence, we conclude  \begin{align*}
\{\,P^*f,\ \pi_{\mathfrak{m}}^*\theta\,\}_\varepsilon=0\qquad\text{ for all }\,f\in S(\mathfrak{g}),\ \theta\in S(\mathfrak{m}-\varepsilon W)^A 
\end{align*} as required.
\end{proof}
\begin{remark}
\label{rem:JY-pairing}
By the definition of $\mathfrak{F}_1^{\mathrm{poly}}$ and $\mathfrak{F}_2^{\mathrm{poly}}$, let $P_u = \ell_u \circ P$ and $f_\theta =  \pi_\mathfrak{m}^*\theta$. Then, from \cite{MR2141306}, the Hamiltonian vector fields are $ X_{P_u}(g,X)$ in \eqref{eq:coovectfie} and \begin{align}
  X_{f_\theta}(g,X)  =   g_*\vert_X \left((\nabla \theta(\xi))_\mathfrak{m}, - \frac{1}{2} [(\nabla \theta(\xi))_\mathfrak{m},X]_\mathfrak{m} - \varepsilon [W,(\nabla \theta(\xi))_\mathfrak{m}]\right).
\end{align} A direct calculation shows that $\{P_u,f_\theta\}_\varepsilon = \omega_\varepsilon(X_{P_u},X_{f_\theta}) = 0$. 
\end{remark}

In Lemma \ref{lem:lax}, we see that $F_2^{\mathrm{poly}}$ is also the set of first integrals as the curve $\pi_{\mathfrak{m}}(t)=X(t)-\varepsilon W$ evolves by $\dot\pi_{\mathfrak{m}}=[\pi_{\mathfrak{m}},\,\Omega(t)]$ with $\Omega(t)\equiv\varepsilon W \in \mathfrak{g}$. Hence, $\pi_{\mathfrak{m}}(t)$ remains on its $A$-orbit. Therefore, every $A$-invariant polynomial $\theta\in S(\mathfrak{m}-\varepsilon W)^A$ is constant along the flow and $\{\pi_{\mathfrak{m}}^*\theta,\mathcal{H}\}_\varepsilon=0$ by Proposition \ref{prop:rightintegrals}. In conclusion, with all the discussions above, we now present two families of the first integrals in the phase space $T^*M$ with magnetic twists. This can be concluded into the following theorem:
\begin{theorem} 
\label{thm:two-families} 
In $\big(T^*M,\omega_{\varepsilon}\big)$ with $\mathcal{H}(g,X)=\frac{1}{2} B(X,X)$, the two polynomial families \begin{align*}
\mathfrak{F}_1^{\mathrm{poly}}=P^*S(\mathfrak{g}), \qquad \mathfrak{F}_2^{\mathrm{poly}}=\pi_{\mathfrak{m}}^*\big(S(\mathfrak{m}-\varepsilon W)^A\big)
\end{align*} consist of first integrals: \begin{align*}
\left\{\mathfrak{F}_1^{\mathrm{poly}},\mathcal{H}\right\}_\varepsilon = 0, \qquad \left\{\mathfrak{F}_2^{\mathrm{poly}},\mathcal{H}\right\}_\varepsilon = 0.
\end{align*} 
\end{theorem}

\begin{proof}
    By the argument from Proposition \ref{prop:leftintegrals}, Proposition \ref{prop:rightintegrals}, Lemma \ref{lem:poly} and Lemma \ref{lem:interzero}, the proof follows automatically. 
\end{proof}

\begin{remark} 
When $W$ is irregular, $P^*$ may have a non-trivial kernel as its image lies in the cone $\mathrm{Ad}(G)(\mathfrak{m}-\varepsilon W)$. 
In contrast, although $\pi_{\mathfrak{m}}$ need not be a well-defined map on $T^*M$, the representative-level map $\widetilde{\pi}_{\mathfrak{m}}(g,X)=X-\varepsilon W$ is surjective onto $\mathfrak{m}-\varepsilon W$. Consequently, the descended pullback $\pi_{\mathfrak{m}}^*:S(\mathfrak{m}-\varepsilon W)^A\cong S(\mathfrak{m})^A\to C^{\infty}(T^*M)$ is injective: if $\pi_{\mathfrak{m}}^*\theta=0$, then $0=q^*(\pi_{\mathfrak{m}}^*\theta)=\theta\circ\widetilde{\pi}_{\mathfrak{m}}$, and surjectivity of $\widetilde{\pi}_{\mathfrak{m}}$ forces $\theta=0$. Thus, $\mathfrak{F}_2^{\mathrm{poly}}$ provides a block of conserved quantities independent of the regularity of $W$.
\end{remark}

\subsection{Poisson projection chains and superintegrabilities on \texorpdfstring{$T^*(G/A)$}{T*(G/A)}}
\label{subsec:Poisuper}

In this Subsection \ref{subsec:Poisuper}, we provide the full superintegrability on $T^*(G/A)$ by verifying the Poisson projection chains starting from $T^*(G/A)$. We first state the regularity of the subset on which we are working. Let $\mathfrak{g}_{\mathrm{reg}} \subset \mathfrak{g}$ denote the set of regular elements. That is, the elements whose centraliser has a dimension $\mathrm{rank} \, \mathfrak{g}$. Since $\mathfrak{g}_{\mathrm{reg}}$ is Zariski open and dense, its complement defines the zero locus of a nonzero polynomial. Let $f:\mathfrak{g}\to\mathbb{R}$ be the polynomial whose nonvanishing locus is the set of regular elements $\mathfrak{g}_{\mathrm{reg}}=\{\xi\in\mathfrak{g}:f(\xi)\neq0\}$. For fixed $W \in \mathfrak{g}$ and $\varepsilon \in \mathbb{R} \setminus \{0\}$,  define the open dense subsets \begin{align}
\mathfrak{m}_{\mathrm{reg}}^W:= \left\{X\in\mathfrak{m}: \xi =  X-\varepsilon W\in \mathfrak{g}_{\mathrm{reg}}\right\},\qquad U:=  \left\{(g,X)\in T^*(G/A): X\in\mathfrak{m}_{\mathrm{reg}}^W\right\}. \label{eq:regularity}
\end{align} Thus, $U\subset T^*M$ is Zariski open and dense. In the complement of $U$, 
In the complement of $U$, the rank of $dP$ and the Jacobian rank of a chosen set of descended invariant generators $\pi_\mathfrak{m}^*\theta_i$ can drop. Hence, all the geometric statements in the section below require regularity.

\begin{remark}  
The phase space is the symplectic manifold $\big(T^*M,\omega_\varepsilon\big)$.  The set $U$ is an open dense submanifold when constant rank or independence of differentials is required. Working on $U$ does not affect the integrability: $U$ is a symplectic submanifold, and the restrictions of $\mathfrak{F}_1^{\mathrm{poly}},\mathfrak{F}_2^{\mathrm{poly}} $ to $U$ have the same functional relations as in $T^*M$. 
\end{remark}

 We now show that working on $U$ does not change any geometric or algebraic properties of $\mathfrak{F}_1^{\mathrm{poly}}$ and $\mathfrak{F}_2^{\mathrm{poly}} $.

\begin{lemma} 
\label{lem:regu}
Let $Z_P:=\mathrm{Ad}(G)(\mathfrak{m}-\varepsilon W)$ and $(Z_P)_{\mathrm{reg}}:= Z_P \cap \mathfrak{g}_{\mathrm{reg}} = \mathrm{Ad}(G)\big(\mathfrak{m}_{\mathrm{reg}}^W\big)$. Then $(Z_P)_{\mathrm{reg}}$ is the Zariski open dense in $Z_P$ and \begin{align*}
\bigl\{f\vert_{Z_P}:f\in\mathbb{R}[\mathfrak{g}^*]\bigr\}=\left\{f\vert_{(Z_P)_{\mathrm{reg}}}:f\in\mathbb{R}[\mathfrak{g}^*]\right\}.
\end{align*} 
\end{lemma}
\begin{proof}
    Since $\mathfrak{g}_{\mathrm{reg}}$ is Zariski open dense, $\mathfrak{m}_{\mathrm{reg}}^W$ is Zariski open dense in $\mathfrak{m}$. Define the map $\phi: G \times \mathfrak{m} \xrightarrow{\ \rho\ } (G \times \mathfrak{m})/T \xrightarrow{\ P \ }  \mathfrak{g}$ by $\phi(g,X) = \mathrm{Ad}(g)(X - \varepsilon W)$ for all $g \in G$. By construction, $\phi$ is a regular morphism of affine varieties. Hence, the maps $\rho$, $P$, and $\phi$ are continuous.  Since $\phi$ is regular and $(G \times \mathfrak{m}_W^{\mathrm{reg}})/A$ is Zariski open dense in $T^*M$, we have $ \phi\left((G \times \mathfrak{m}_W^{\mathrm{reg}})/A\right) = \phi(T^*M) = Z_P$. 
    Hence, $\overline{(Z_P)_{\mathrm{reg}}} = Z_P$. 
\end{proof}
\begin{remark}
 On the regular stratum, \begin{align}
    U =\{(g,X)\in G\times_{A}\mathfrak{m}: \xi=X-\varepsilon W\ \text{is regular in }\mathfrak{g}\} \subset T^*M. \label{eq:regularst}
\end{align} Lemma \ref{lem:regu} states that $\mathrm{Ad}(G)(\mathfrak{m}_{\mathrm{reg}}^W)$ is Zariski dense in $\mathrm{Ad}(G)(\mathfrak{m} - \varepsilon W)$. All algebraic constructions for finding the finitely generated algebras $\mathfrak{F}_1^{\mathrm{poly}}, \mathfrak{F}_2^{\mathrm{poly}}$, and their intersection and union, are unchanged by replacing $\mathfrak{m} - \varepsilon W$ with $\mathfrak{m}_{\mathrm{reg}}^W$. 
\end{remark}

 Define the shift restriction by \begin{align}
     \mathrm{Res}_W: S(\mathfrak{g})^G \longrightarrow S(\mathfrak{m})^A  \text{ given by } \mathrm{Res}_W(C) (X) = C(X - \varepsilon W)
 \end{align}  for any $X \in \mathfrak{m}$. We have the following commutative diagram:   \[\begin{tikzcd}
	{S(\mathfrak{g})^G } & {S(\mathfrak{m})^A} \\
	{\mathcal{O}(T^*M)} & {\mathcal{O}(T^*M)}
	\arrow["{\mathrm{Res}_W}", from=1-1, to=1-2]
	\arrow["{P^*}"', from=1-1, to=2-1]
	\arrow["{\pi_\mathfrak{m}^*}", from=1-2, to=2-2]
	\arrow[from=2-1, to=2-2]
    \arrow[from=1-1, to=2-2, phantom, "\circlearrowright" {anchor=center, scale=1.5, rotate=90}]
\end{tikzcd}\] 
We first show that the restriction map is well-defined in the following lemma: \begin{lemma}
    \label{lem:welldefineres}
Let $\mathrm{Res}_W$ be the restriction mapping defined above. Then $\mathrm{Res}_W$ is a well-defined $\mathbb{R}$-algebra homomorphism. In particular, $\mathrm{Im Res}_W $ is a subalgebra of $S(\mathfrak{m})^A$.
\end{lemma}
\begin{proof}
    By Lemma \ref{lem:property}, $S(\mathfrak{m})^A \cong S(\mathfrak{m} - \varepsilon W)^A$. For any $C \in S(\mathfrak{g})^G$, the composition $C \circ \tau_{\varepsilon W} \in S(\mathfrak{m})$. Hence, $\mathrm{Res}_W(C) \in S(\mathfrak{m})$. Note that $\mathrm{Res}_W(C)  $ is also $A$-invariant. Indeed, for any non-zero $a \in A$, by the definition of the Poisson centralizer, \begin{align*}
    \mathrm{Res}_W(C)(\mathrm{Ad}(a)X) = C(\mathrm{Ad}(a)X - \varepsilon W) = C(\mathrm{Ad}(a)(X -\varepsilon W)) = C(X - \varepsilon W) = \mathrm{Res}_W(C)( X).
\end{align*} Hence, $  \mathrm{Res}_W(C) \in S(\mathfrak{m})^A$.

We now show that $ \mathrm{Res}_W(C)$ is a $\mathbb{R}$-algebra isomorphism. This can be verified as follows: For any $C,C' \in S(\mathfrak{g})^G$ and $\lambda \in \mathbb{R}$, we have \begin{align*}
     \mathrm{Res}_W(CC') = & \, C C' \circ \tau_{\varepsilon W} = \left(C \circ \tau_{\varepsilon W} \right) \left( C' \circ \tau_{\varepsilon W} \right) =\mathrm{Res}_W(C)\mathrm{Res}_W(C')  \\
       \mathrm{Res}_W(C+\lambda C') = & \,(C+\lambda C') \circ \tau_{\varepsilon W} = \left(C \circ \tau_{\varepsilon W}  + \lambda C' \circ \tau_{\varepsilon W} \right)  = \mathrm{Res}_W(C) + \lambda \mathrm{Res}_W(C').
\end{align*}

From these two paragraphs above, $\mathrm{Im Res}_W$ is a subalgebra of $S(\mathfrak{m})^A$.
\end{proof}
 \begin{remark}
\label{rem:ResW}
Fix $W \in \mathfrak{g}$ and $\varepsilon\in\mathbb{R}/\{0\}$. By Lemma \ref{lem:property}, we define the affine translation \begin{align*}
\tau_{-\varepsilon W} : \mathfrak{m} \longrightarrow \mathfrak{m}-\varepsilon W,\quad X\longmapsto X-\varepsilon W .
\end{align*} For $C\in S(\mathfrak{g})^G$, Lemma \ref{lem:welldefineres} shows that $\mathrm{Im Res}_W$ is a subalgebra and   $  \mathrm{Res}_W(C)(X):=C(X-\varepsilon W)=C\big(\tau_{-\varepsilon W}(X)\big)$.

By Lemma \ref{lem:property} (ii), the pullback along $\tau_{-\varepsilon W}$ induces an isomorphism of invariant algebras via its pullback as follows:
\begin{align*}
\tau_{-\varepsilon W}^* : S(\mathfrak{m}-\varepsilon W)^A  \xrightarrow{\ \cong\ } S(\mathfrak{m})^A,\quad (\tau_{-\varepsilon W}^*\theta)(X)=\theta(X-\varepsilon W).
\end{align*} With this notation $\mathrm{Res}_W = \tau_{-\varepsilon W}^*\circ \big(C\mapsto C\vert_{\mathfrak{m}-\varepsilon W}\big)$,  we may regard $\mathrm{Res}_W$ as a map $S(\mathfrak{g})^G\to S(\mathfrak{m}-\varepsilon W)^A$. The choice of $S(\mathfrak{m})^A$ as the codomain is simply the identification of the affine slice $\mathfrak{m}-\varepsilon W$ with the linear space $\mathfrak{m}$ via $\tau_{-\varepsilon W}$.

 If $\xi=X-\varepsilon W\in\mathfrak{m}-\varepsilon W$, then $\mathrm{Res}_W(C)(X)=C(\xi)$. Therefore, whenever expressions such as $\mathrm{Res}_W(C)(\xi)$ appear later, they are to be understood through the identification $\mathfrak{m}\xrightarrow{ \tau_{-\varepsilon W}\ } \mathfrak{m}-\varepsilon W$.
\end{remark}

From the construction in Subsection \ref{subsec:geomi}, we see that $\mathfrak{g}//G$ is the spectrum of the Poisson center $S(\mathfrak{g})^G$. Choosing homogeneous generators $C_1,\ldots,C_r \in S(\mathfrak{g})^G$, we define a canonical quotient map \begin{align}
    \chi: \mathfrak{g} \longrightarrow \mathfrak{g}//G \cong \mathbb{A}^r, \  \text{ } \zeta \longmapsto (C_1(\zeta),\ldots,C_r(\zeta)).  \label{eq:quotient}
\end{align} By the definition of the point in $\mathfrak{g}//A$, defined in \eqref{eq:pointsinquo}, for a non-zero regular value $\zeta \in \mathfrak{g}$, the fiber $\chi^{-1} \left(\chi(\zeta)\right)$ exactly defines the adjoint orbit $\mathcal{O}_\zeta$ and $\dim \mathcal{O}_\zeta = n -r $. That is, for a generic value $(a_1,\ldots,a_r) \in \mathrm{Im}\, \chi_{Z_P}$, the fiber \begin{align*}
    \left(\chi\vert_{Z_P}\right)^{-1}(a_1,\ldots,a_r)
\end{align*} is a dimension $n-r$ adjoint orbit. We first deduce the rank of the subalgebra $\mathrm{Im Res}_W$ from the following proposition. 
\begin{proposition}
\label{pro:dimension}
 Let $\mathrm{Res}_W$ be the same as defined above. Suppose that $C_1, \ldots,C_r$ are algebraically independent generators of $S(\mathfrak{g})^G$ such that $\chi: \mathfrak{g} \rightarrow \mathbb{A}^r$ is given by $\chi(\zeta) = \left(C_1(\zeta),\ldots,C_r(\zeta)\right)$. Here, $\mathbb{A}^r$ is a $r$-dimensional affine space. Define $s: = \mathrm{trdeg} \, \left(\mathrm{Im Res}_W\right)$ and suppose that $C$ restricts to an algebraically independent set on $\mathfrak{m} -\varepsilon W$. Then, for any $X \in \mathfrak{m}$ such that $\xi = X - \varepsilon W$, we have \begin{align*}
     \mathrm{rank} \, d(\mathrm{Res}_W)_X = \dim \left(\mathfrak{g}_\xi\right)_\mathfrak{m} = s,
 \end{align*} where $(\cdot)_\mathfrak{m}$ is the projection on the $\mathfrak{m}$-component.
\end{proposition}

\begin{remark}
\label{rem:dResW}
Strictly speaking, $\mathrm{Res}_W:S(\mathfrak{g})^G\to S(\mathfrak{m})^A$ is a homomorphism of coordinate rings. Hence, it induces a morphism of affine varieties in the opposite direction
\begin{align*}
    r_W: \mathfrak{m}//A=\mathrm{Spec}\bigl(S(\mathfrak{m})^A\bigr)\longrightarrow \mathfrak{g}//G=\mathrm{Spec}\bigl(S(\mathfrak{g})^G\bigr),\qquad r_W^*=\mathrm{Res}_W.
\end{align*}
When we write $d(\mathrm{Res}_W)_X$ in Proposition \ref{pro:dimension}, we mean the differential $d(r_W)_X$ at the point $X$. After choosing generators $C_1,\ldots,C_r$ of $S(\mathfrak{g})^G$, this differential is computed by the Jacobian of the polynomial map
\begin{align*}
    \varphi: \mathfrak{m}\to\mathbb{A}^r,\qquad \varphi(X)=\bigl(\mathrm{Res}_W(C_1)(X),\ldots,\mathrm{Res}_W(C_r)(X)\bigr),
\end{align*}
so that $\mathrm{rank}\,d(\mathrm{Res}_W)_X:=\mathrm{rank}\,d\varphi_X$.
\end{remark}

 \begin{proof}
 Fix $B$ as a $G$-invariant Killing form on $G$. Recall that $B\vert_\mathfrak{m}$ is also non-degenerate such that $T_X^*\mathfrak{m} \cong \mathfrak{m}$.    We prove this argument by computing the Jacobian of $\mathrm{Res}_W$. We first show that the gradients of invariant polynomials lie in the centraliser. For any $\zeta,Y \in \mathfrak{g}$ and $G$-invariant $C \in S(\mathfrak{g})^G$, using the Ad-invariant property, we deduce  \begin{align}
         0 = \left.\dfrac{d}{dt}\right\vert_{t = 0}C(\mathrm{Ad}(\exp(-tY))\zeta) = dC(\zeta) [Y,\zeta]  = B(\nabla C(\zeta),[Y,\zeta]) = B([\nabla C(\zeta),\zeta],Y). \label{eq:diffc}
     \end{align} Therefore, we observe that $dC_i(\zeta)$ annihilates $T_\zeta(G\cdot\zeta)$. In addition, the last equality in \eqref{eq:diffc} implies that $[\nabla C(\zeta),\zeta] = 0$ and $\nabla C(\zeta) \in \mathfrak{g}_\zeta$. Hence, via $B$, the differential $dC_i(\zeta)$ corresponds to the gradient $\nabla C_i(\zeta)$ defined by $dC(\zeta) (w) =  B(\nabla C(\zeta),w)$ for all $w \in \mathfrak{g}$. Without loss of generality, take $\zeta = \xi = X - \varepsilon W$. By definition of the adjoint action, \begin{align}
         T_\xi(G\cdot \xi) = [\mathfrak{g},\xi] := \{[Y,\xi]: Y \in \mathfrak{g}\} . \label{eq:tanorbitspace}
     \end{align}  Note that the annihilator of the tangent orbit space \eqref{eq:tanorbitspace} is given by $\widehat{[\mathfrak{g},\xi]} : = \{\lambda \in \mathfrak{g}^*: \text{ } \lambda([\mathfrak{g},\xi]) = 0\}$. For a $\lambda_{\nabla C(\xi)} \in \widehat{[\mathfrak{g},\xi]}$, let $\lambda_{\nabla C(\xi)}([\mathfrak{g},\xi]) = B(\nabla C(\xi),[\mathfrak{g},\xi])  $. Using \eqref{eq:diffc} again, we obtain $\mathfrak{g}_\xi \cong \widehat{[\mathfrak{g},\xi]}$. Hence, the set of gradients $\{\nabla C_i(\xi)\}_i$ forms the basis of $\mathfrak{g}_\xi$.

Now, consider the map \begin{align*}
    \varphi: \mathfrak{m}  \longrightarrow \mathbb{A}^r, \quad X \longmapsto \varphi(X) : =  (\mathrm{Res}_W(C_1)(X),\ldots,\mathrm{Res}_W(C_r)(X)) = \left((C_1)\vert_\mathfrak{m},\ldots,(C_r)\vert_\mathfrak{m}\right)
\end{align*} To find the rank of $\mathrm{Res}_W$, it is sufficient to compute the rank of $d \varphi_X$. For any $\delta X = u  \in \mathfrak{m}$, the differential of the restriction map in  \eqref{eq:diffc} is given as follows: \begin{align}
   d (\mathrm{Res}_W(C_i))_X(u) =  d C_i(\xi)(u) = B(\nabla C_i(\xi),u) = B(\left(\nabla C_i(\xi)\right)_\mathfrak{m},u). \label{eq:differenres}
\end{align} Thus, under the identification of $T_X^*\mathfrak{m}$ with $\mathfrak{m}$, the covector $d\varphi_i\vert_\xi$ is the vector $(\nabla C_i(\xi))_{\mathfrak{m}}\in\mathfrak{m}$. Equivalently, in any $B$-orthonormal basis of $\mathfrak{m}$, the $i$-th row of the Jacobian matrix $J_F(\xi)\in M_{r\times\dim\mathfrak{m}}(\mathbb{R})$ is the coordinate row of $(\nabla C_i(\xi))_{\mathfrak{m}}$. 

Since each $C_i$ is $G$-invariant, $\nabla C_i(\xi)\in\mathfrak{g}_\xi:=\{Y\in\mathfrak{g}:[Y,\xi]=0\}$, and for regular $\xi$, the $\nabla C_i(\xi)$ spans $\mathfrak{g}_\xi$ for all $i$. Hence, \begin{align*}
\mathrm{span}\big\{(\nabla C_i(\xi))_{\mathfrak{m}}\big\}_{i=1}^r=(\mathfrak{g}_\xi)_{\mathfrak{m}} \subset\mathfrak{m}.
\end{align*} Using the $B$-invariance and \eqref{eq:differenres}, for any $X \in \mathfrak{m}$, \begin{align}
\ker d\varphi_X &=\big\{u \in\mathfrak{m}:B\big((\nabla C_i(\xi))_{\mathfrak{m}},\,u \big)=0\ \text{ for all } i\big\} \label{eq:kervarph}\\
\nonumber
&=\big\{u \in\mathfrak{m}: B(Y,u)=0  \text{ for all } Y \in(\mathfrak{g}_\xi)_{\mathfrak{m}}\big\}\\
\nonumber
&=\mathfrak{m}\cap \big((\mathfrak{g}_\xi)^\perp\big) =\mathfrak{m}\cap[\mathfrak{g},\xi].
\end{align} Here, in the last line of \eqref{eq:kervarph}, we use the orthogonal decomposition $\mathfrak{g}=\mathfrak{g}_\xi\oplus[\mathfrak{g},\xi]$ with respect to $B$, which is standard for regular $\xi$ as $B\big([\mathfrak{g},\xi],\mathfrak{g}_\xi\big)=0$ and $\dim\mathfrak{g}_\xi+\dim[\mathfrak{g},\xi]=\dim\mathfrak{g}$. Hence, inside $\mathfrak{m}$, $((\mathfrak{g}_\xi)_\mathfrak{m})^\perp = \mathfrak{m} \cap [\mathfrak{g},\xi]$.  Therefore, by the rank-nullity theorem, \begin{align*}
\mathrm{rank}\, d\varphi_X =\dim\mathfrak{m}-\dim\big(\mathfrak{m}\cap[\mathfrak{g},\xi]\big) =\dim\big((\mathfrak{g}_\xi)_{\mathfrak{m}}\big)
\end{align*} as required.
 \end{proof}

 \begin{remark}
 \label{rem:dimension}
(i) Define $Z_P : = \mathrm{Ad}(G)(\mathfrak{m} - \varepsilon W)$. Choose any finitely-generating set $\{u_1,\ldots,u_s\}\subset\mathrm{Im}(\mathrm{Res}_W)$ and form the compressed map \begin{align*}
\varphi^{(s)}:Z_{P}\longrightarrow\mathbb{A}^s,\qquad \varphi^{(s)}(\xi):=\bigl(u_1(\xi),\ldots,u_s(\xi)\bigr).
\end{align*} Then, generically, \begin{align*}
\mathrm{rank}\,d\varphi^{(s)}=s,\qquad \ker d\varphi=\ker d\varphi^{(s)} 
\end{align*} due to the algebraic dependence on the generators in $\mathrm{Im}\,\varphi $. 


(ii) For arbitrary $X\in\mathfrak{m}$ with $\xi:=X-\varepsilon W$, we have  \begin{align*}
\mathrm{rank}\,d(\mathrm{Res}_W)_X \leq \dim\big((\mathfrak{g}_\xi)_{\mathfrak{m}}\big).
\end{align*} Indeed, the $i$-th row of the Jacobian of $(\mathrm{Res}_W(C_1),\ldots,\mathrm{Res}_W(C_r))$ at $X$ is $(\nabla C_i(\xi))_{\mathfrak{m}}\in(\mathfrak{g}_\xi)_{\mathfrak{m}}$ as $C_i$ is $G$-invariant and hence $\nabla C_i(\xi)\in\mathfrak{g}_\xi$. Therefore, the row space is contained in $(\mathfrak{g}_\xi)_{\mathfrak{m}}$, giving the inequality. On the regular locus $\xi\in\mathfrak{g}_{\mathrm{reg}}$, the rows span $(\mathfrak{g}_\xi)_{\mathfrak{m}}$. Thus, as presented in the proof above, the equality holds there. The locus where the rank is less than $s$ is Zariski closed (the vanishing of the $s\times s$ minors of the Jacobian) and is disjoint from a Zariski open dense subset of the regular slice.
 \end{remark}

 \begin{corollary}
    \label{coro:a=t} Let $A = T $ such that $\mathfrak{g} = \mathfrak{t} \oplus \mathfrak{m}$. Then $\mathrm{Res}_W: S(\mathfrak{g})^G \rightarrow S(\mathfrak{m})^T$ is injective. In particular, $\mathrm{Res}_W \left(S(\mathfrak{g})^G\right) \subset S(\mathfrak{m})^T$.
\end{corollary}

\begin{proof}
  By Chevalley's theorem, $S(\mathfrak{g})^G \cong \mathbb{A}^r$. Hence, fix such a finite set $\{C_1,\dots,C_r\}\subset S(\mathfrak{g})^G$ such that $S(\mathfrak{g})^G=\mathbb{R}[C_1,\dots,C_r]$. By assumption, $A=T$ and $W$ is regular. The restricted elements $$\mathrm{Res}_W(C_1), \dots, \mathrm{Res}_W(C_r) \in S(\mathfrak{m})^T$$ are algebraically independent.

Let $C\in S(\mathfrak{g})^G$ satisfy $\mathrm{Res}_W(C)=0$. Since $S(\mathfrak{g})^G = \mathbb{R}[C_1,\dots,C_r]$, there exists a polynomial $F\in \mathbb{R}[C_1,\dots,C_r]$ such that $C=F(C_1,\dots,C_r)$. Applying $\mathrm{Res}_W$ and Lemma \ref{lem:welldefineres}, we obtain \begin{align}
   0=\mathrm{Res}_W(C) =\mathrm{Res}_W \bigl(F(C_1,\dots,C_r)\bigr) =F\bigl(\mathrm{Res}_W(C_1),\dots,\mathrm{Res}_W(C_r)\bigr) \in S(\mathfrak{m})^T. 
\end{align} Since $\mathrm{Res}_W(C_1),\dots,\mathrm{Res}_W(C_r)$ are algebraically independent, the only polynomial $F$ with $$F(\mathrm{Res}_W(C_1),\dots,\mathrm{Res}_W(C_r))=0$$ is the zero polynomial. Hence, $F=0$, so $C=0$. Therefore, $\ker(\mathrm{Res}_W)=\{0\}$ and $\mathrm{Res}_W$ is injective.
\end{proof}
\begin{remark}
  Let us illustrate an example for Corollary \ref{coro:a=t}. Let $\mathfrak{g} = \mathfrak{su}(2)$, and $\mathfrak{t} \cong \mathbb{R}$ be Cartan subalgebra. We further assume that $(x,y,z)$ is the coordinate for $\mathfrak{su}^*(2) $ such that $ x$ is the coordinate function on $\mathfrak{t}^*$ and $(y,z)$ is the coordinate functions on $\mathfrak{m}^*$. Then the Casimir $C$ is defined by $x^2 + y^2 + z^2$, and $S(\mathfrak{su}(2))^{\mathrm{SU}(2)}$ and $S(\mathfrak{m})^T = \mathbb{R}[y^2 + z^2]$.  After the shifting, the restriction map is then given by \begin{align*}
        \mathrm{Res}_W (\Omega) (y,z) = \Omega((0,y,z) - \varepsilon W ) = x^2 + y^2 + \varepsilon^2||W||^2 \in S(\mathfrak{m} - \varepsilon W)^T \cong S(\mathfrak{m})^T,
    \end{align*} which is clearly injective.
\end{remark}

We now determine the rank of $\mathfrak{F}_1^{\mathrm{poly}}$ and $\mathfrak{F}_2^{\mathrm{poly}}$, which helps us prove the dimensional characterization of superintegrable systems.

\begin{lemma}
\label{lem:rankacc}
 Let  $ \mathfrak{F}_1^{\mathrm{poly}} $ and let $ \mathfrak{F}_2^{\mathrm{poly}}$ be the same as defined above, and let $r = \mathrm{rank} \,G$. Let $\mathrm{Res}_W: S(\mathfrak{g})^G \rightarrow S(\mathfrak{m}-\varepsilon W)^A$ be the restriction given by $\mathrm{Res}_W(C)(X) = C(X - \varepsilon W)$. Then\begin{align*}
  \mathrm{Spec}\, \mathfrak{F}_1^{\mathrm{poly}}=\mathrm{Ad}(G)(\mathfrak{m}-\varepsilon W),\quad \mathrm{Spec}\,\mathfrak{F}_2^{\mathrm{poly}}=(\mathfrak{m}-\varepsilon W)//A \cong  \mathfrak{m}//A .  
\end{align*} In particular, \begin{align*}
    \dim \mathrm{Spec} \,\mathfrak{F}_1^{\mathrm{poly}} = \left\{\begin{matrix}
        n & \text{ if } A = T \\
        n- r +  s & \text{ if } T \subsetneq A
    \end{matrix}\right. \text{ and } \dim \mathrm{Spec}\, \mathfrak{F}_2^{\mathrm{poly}} = \dim \mathfrak{m} - \dim A +  \dim A_y,
\end{align*} where $\dim A_y$ denotes the dimension of a stabiliser of $A$ acting on $\mathfrak{m}$.
\end{lemma}

\begin{proof}
 Let $ Z_P  = \mathrm{Ad}(G)(\mathfrak{m} -\varepsilon W) = \bigcup_{X \in \mathfrak{m}} \mathrm{Ad}(G)(X - \varepsilon W)$. In particular, $Z_P$ is $G$-invariant. That is, $\mathrm{Ad}(G)Z_P = Z_P$. We first look at the case where $ A \neq T$. By the definition of $P$, $ P^* : S(\mathfrak{g}) \twoheadrightarrow \mathfrak{F}_1^{\mathrm{poly}}$ is a ring homomorphism.  By Lemma \ref{lem:poly} and its remark, the pullback identifies \begin{align*}
      \mathfrak{F}_1^{\mathrm{poly}} \cong S(\mathfrak{g})\big/\mathcal{I}(Z_P).
  \end{align*} Hence, $\mathrm{Spec}\, \mathfrak{F}_1^{\mathrm{poly}}  = Z_P$\footnote{Formally, we should write $\mathrm{Spec} \, \mathfrak{F}_1^{\mathrm{poly}} = \overline{Z_P}$ as a Zariski closed subset. However, recall that from Remark \ref{rem:ideal}, we will omit the Zariski closure.}, and therefore, $\mathrm{trdeg} \, \mathfrak{F}_1^{\mathrm{poly}}=\dim Z_P$. It remains to compute $\dim Z_P$. By definition, $Z_P$ consists of the adjoint orbit $\mathrm{Ad}(G)(X -\varepsilon W)$ through the affine subspace $\mathfrak{m} - \varepsilon W$. Let $\chi:\mathfrak{g}\to\mathfrak{g}//G\cong \mathbb{A}^r$ be the canonical quotient defined by \eqref{eq:quotient}. Restrict $\chi$ to $Z_P$ and set $\Gamma:=\overline{\chi(Z_P)}\subset\mathbb{A}^r$.  Since $C_j\circ P=(\mathrm{Res}_W C_j)\circ\pi_{\mathfrak{m}}$ on $T^*M$, with  $\pi_{\mathfrak{m}}(g,X)=X-\varepsilon W$, the coordinate ring is identified as \begin{align*}
    \mathbb{R}[\Gamma]\cong \mathrm{Im}(\mathrm{Res}_W)\subset S(\mathfrak{m}-\varepsilon W)^A
\end{align*}  as $S(\mathfrak{g})^G=\mathbb{R}[C_1,\dots,C_r]$ and we restrict the generators to $\mathfrak{m} - \varepsilon W$. Hence, \begin{align*}
    \dim\Gamma = \mathrm{trdeg}\,\mathbb{R}[\Gamma] = \mathrm{trdeg}\,\mathrm{Im}(\mathrm{Res}_W) =  s.
\end{align*} By Proposition \ref{pro:dimension}, on the regular locus $U$, we also have  \begin{align*}
   s = \mathrm{rank} \, d (\mathrm{Res}_W)(X) = \dim (\mathfrak{g}_\xi)_\mathfrak{m}  ,
\end{align*} so $s$ is the generic differential rank of $\mathrm{Res}_W$.

On the Zariski open dense regular locus $\mathfrak{g}_{\mathrm{reg}}\subset\mathfrak{g}$, the Chevalley map $\chi$ is smooth, with all fibers of dimension $n-r$, and $\mathfrak{g}_{\mathrm{reg}}\cap Z_P\neq\emptyset$ by Lemma \ref{lem:regu}. Thus, for any $\xi\in Z_P\cap\mathfrak{g}_{\mathrm{reg}}$, the fiber $\chi^{-1}(\chi(\xi))=\mathrm{Ad}(G)\cdot\xi$ has dimension $n-r$ and lies in $Z_P$ since $Z_P$ is $\mathrm{Ad}(G)$-invariant. Therefore, the generic fiber of $\chi\vert_{Z_P}:Z_P\to\Gamma$ has dimension $n-r$. By the dimension theorem for morphisms \cite[Chapter II]{HartshorneAG}, we have \begin{equation}
\dim Z_P =  \dim\Gamma\ +\bigl(\text{generic fiber dimension of }\chi\vert_{Z_P}\bigr)\ =\ s+(n-r).\label{eq:dimZP}
\end{equation} It follows that $\dim Z_P=n-(r-s)$. Since $\mathrm{Spec}\, \mathfrak{F}_1^{\mathrm{poly}}=Z_P$, \begin{align*}
\mathrm{trdeg}\,\mathfrak{F}_1^{\mathrm{poly}} = \dim Z_P = n-(r-s),
\end{align*} as claimed.  

On the other hand, if $A = T$ with $W$ is regular, from the construction above, we find that the $G$-invariant generators $C_1,\ldots,C_r$ are in one-to-one correspondence with the restricted algebraically independent functions in $S(\mathfrak{m} -\varepsilon W)^T$. That is, $\mathrm{Res}_W(C_i)$ are algebraically independent for all $i = 1,\ldots,r$. Hence, $\chi(\mathfrak{m} -\varepsilon W) $ is Zariski dense in $\mathbb{A}^r$. In this case, by the definition of $\chi$, we have \begin{align*}
    \Gamma = \overline{\chi(Z_P)} = \mathbb{A}^r \text{ and } s = \mathrm{trdeg}\, \mathrm{Im \, Res}_W = r. 
\end{align*} Hence, \eqref{eq:dimZP} becomes $\dim Z_P= r +  (n-r) = n$. Since $\mathfrak{g}$ is $n$-dimensional,  any subset of dimension $n$ must be Zariski dense. Hence, $Z_P$ is Zariski dense in $\mathfrak{g}$. It follows that \begin{align*}
    \dim \mathrm{Spec}\, \mathfrak{F}_1^{\mathrm{poly}} = \dim Z_P = n \text{ if } A = T.
\end{align*}

We now focus on the spectrum of $\mathfrak{F}_2^{\mathrm{poly}}$ and its rank. The image of the map $P$ is $\mathrm{Ad}(G)(\mathfrak{m}-\varepsilon W)$, and thus its coordinate ring is equal to the quotient of $S(\mathfrak{g})$ by the vanishing ideal of that cone. Since $\pi_\mathfrak{m}^*$ is injective, we have an isomorphism of $\mathfrak{F}_2^{\mathrm{poly}}$ with the invariant ring $S(\mathfrak{m})^A$. Thus, $\mathrm{rank}\, \mathfrak{F}_2^{\mathrm{poly}} = \dim \mathrm{Spec} \,( \mathfrak{m}//A)$. \end{proof}
 
\begin{remark}
    (i) When $W$ is regular and $A=T$, as we presented in the proof, the image of the restriction $\mathrm{Res}$ (as a subalgebra) is of rank ($s=r$), hence $\mathrm{trdeg} \,\mathfrak{F}_1^{\mathrm{poly}}=n$, recovering the familiar regular case.

(ii) In the irregular case $s<r$, $Z_P \subset \mathfrak{g}$ with codimension $r-s$ is entirely accounted for by the invariant equations $C_j=\text{const}$, whose restrictions $\mathrm{Res}(C_j)$ are algebraically dependent or constant on the slice. No extra non-invariant equations contribute to the codimension, as $\chi(Z_P)=\Gamma$ and the generic fiber of $\chi\vert_{Z_P}$ already have the maximal orbit dimension $n-r$.

(iii) Later in Section \ref{sec:examples}, we will see an example with $G=\mathrm{SU}(3)$ and $A=\mathrm{S}(\mathrm{U}(1)\times \mathrm{U}(2))$. In this case, we have $n=8$, $r=2$, $s=1$. Hence, $\mathrm{trdeg}\,\mathfrak{F}_1^{\mathrm{poly}}=8-(2-1)=7$. This matches the explicit coordinate computations and the adjoint-quotient dimension count in Lemma \ref{lem:rankacc}.
\end{remark}

 In what follows, we focus on constructing the intersection and union of two algebras $\mathfrak{F}_1^{\mathrm{poly}}$ and $\mathfrak{F}_2^{\mathrm{poly}}$. 
 We will determine on the finitely-generated algebra consisting of generators from the set $F_1^{\mathrm{poly}} \cap  F_2^{\mathrm{poly}}$. Let $F,A, B$ be a commutative $\mathbb{R}$-algebra. Suppose further that $A,B$ are subalgebras of $F$ via fixed embeddings $  A \xhookrightarrow{\iota_A} F    \xhookleftarrow{\iota_B} B$. Then the intersection is the following algebra: \begin{align*}
    A \cap B = \left\{ f \in F: f \in A \text{ and } f \in B \right\},
\end{align*} which is canonically isomorphic to the fiber product as follows \begin{align}
    A \cap B \cong \left\{(a,b) \in A \times B : \iota_A (a) = \iota_B(b)\right\} : = A \times_F B. \label{eq:intersection}
\end{align}  This is summarised into the following lemma.

\begin{lemma} 
\label{lem:interfiber}
Let $\mathcal{R}$ be a commutative $\mathbb{R}$-algebra and let $\iota_A:A\hookrightarrow F$, $\iota_B:B\hookrightarrow F$ be injective $\mathbb{R}$-algebra homomorphisms. Identify $A,B$ with their images in $\mathcal{R}$. Then $A\cap B  \cong  A\times_F B  $ as $\mathbb{R}$-algebras, via $ \Phi:A\cap B\to A\times_F B$, $\,\Phi(f)=(f,f)$, with inverse $(a,b)\mapsto \iota_A(a)=\iota_B(b)\in F$.
\end{lemma}
\begin{proof}
If $f \in A\cap B$, then $(f,f)\in A\times_F B$ since the two images in $\mathcal{R}$ coincide. Conversely, for $(a,b)\in A\times_F B$ we have $\iota_A(a)=\iota_B(b)\in F$, so this common element lies in $A\cap B$. Both maps are algebra homomorphisms and inverses of each other.
\end{proof}

Now, in our case, set $  \mathcal{O}(T^*M)$, $A = \mathfrak{F}_1^{\mathrm{poly}} \subset \mathcal{O}(T^*M)$ and $ B = \mathfrak{F}_2^{\mathrm{poly}} \subset \mathcal{O}(T^*M)$, and let $\mathfrak{F}_1^{\mathrm{poly}} \cap \mathfrak{F}_2^{\mathrm{poly}}   $ be the intersection algebra. Using Lemma \ref{lem:interfiber}, we have the following formal definition of intersection algebra \begin{align}
\nonumber
    \mathfrak{F}_1^{\mathrm{poly}} \cap  \mathfrak{F}_2^{\mathrm{poly}} := & \, \left\{ f\in \mathcal{O}(T^*M): h \in S(\mathfrak{g})^{G}, \ \theta \in S(\mathfrak{m})^A \text{ such that } f = h \circ P = \theta \circ \pi_\mathfrak{m} \right\} \\
    \cong & \, \mathfrak{F}_1^{\mathrm{poly}} \times_{\mathcal{O}(T^*M)} \mathfrak{F}_2^{\mathrm{poly}}. \label{eq:algfirinte}
\end{align} 

We now focus on the identification of the intersection algebra. Starting with the following lemma:

\begin{lemma}
\label{lem:induce} \cite{MR1920389}
    Let $G$ be a compact Lie group, and let $W,V,T$ be the $G$-modules. Then the short exact sequence of $G$-modules \begin{align*}
        0 \longrightarrow W \xrightarrow{\,\iota\,} T \xrightarrow{\,q\,} V \longrightarrow 0
    \end{align*} induce the following $G$-invariants short exact sequence \begin{align*}
                0 \longrightarrow W^G \xrightarrow{\,\iota\,} T^G \xrightarrow{\,q\,} V^G \longrightarrow 0
    \end{align*}
\end{lemma}


\begin{proposition}
\label{prop:inter} Let $\mathfrak{F}_1^{\mathrm{poly}},\mathfrak{F}_2^{\mathrm{poly}}$ be as defined above, and let $Z_P = \mathrm{Ad}(G)(\mathfrak{m} - \varepsilon W) \subset \mathfrak{g}$. Let $\mathrm{Res}_W:S(\mathfrak{g})^G \rightarrow S(\mathfrak{m}-\varepsilon W)^A$ denote the restriction mapping. Then \begin{align}
    R_0:=\mathfrak{F}_1^{\mathrm{poly}} \cap \mathfrak{F}_2^{\mathrm{poly}}= P^* \left(\mathrm{Im \,Res_W}\right) = \left\{ h\circ P :\ h\in S(\mathfrak{g})^G, h\vert_{Z_P}\text{ well-defined} \right\}.
\end{align} Equivalently, $R_0=\{ f\in\mathcal{O}(T^*M):   \text{ there exist } C\in S(\mathfrak{g})^G,\ \theta\in S(\mathfrak{m}-\varepsilon W)^A \text{ such that }\ f=C \circ P= \pi_{\mathfrak{m}}^* \theta \}$.
\end{proposition}
\begin{proof}

Let $\varphi\in \mathfrak{F}_1^{\mathrm{poly}}\cap \mathfrak{F}_2^{\mathrm{poly}}$. By Definition \ref{def:maps}, there are $h\in S(\mathfrak{g})$ and $\theta\in S(\mathfrak{m}-\varepsilon W)^A$ with \begin{align}
\varphi=P^*h=\pi_{\mathfrak{m}}^*\theta. \label{eq:same}
\end{align} Evaluating at $(g,X)$ and letting $\xi=X-\varepsilon W$, using \eqref{eq:same}, we deduce  $h\bigl(\mathrm{Ad}(g)\,\xi\bigr)=\theta(\xi)$ for all $g \in G$. Hence, $h$ is constant on each $\mathrm{Ad}(G)$-orbit contained in $Z_P$, i.e., $h\vert_{Z_P}\in S(Z_P)^G$.  Consider the short exact sequence of $G$-modules
\begin{align}
0\longrightarrow \mathcal{I}(Z_P)\xrightarrow{\ \iota \ } S(\mathfrak{g})\xrightarrow{\ q\ } S(Z_P) \cong S(\mathfrak{g})/\mathcal{I} (Z_P)\longrightarrow 0, \label{eq:seq}
\end{align} where $\mathcal{I}(Z_P)$ is the ideal of $Z_P$ and $S(Z_P)\cong S(\mathfrak{g})/\mathcal{I}(Z_P)$. Note that $S(\mathfrak{g}) = \bigoplus_{j \geq 0} S^j(\mathfrak{g})$ with each $S^j(\mathfrak{g})$ is finite-dimensional. Moreover, since $G$ is compact and semisimple, all morphisms $\iota$ and $q$ are morphisms of graded $G$-modules. For a homogeneous polynomial $h \in S(\mathfrak{g})$, fixed $d: = \deg h \in \mathbb{N}_0$, we define \begin{align*}
    S_d(\mathfrak{g}) := \bigoplus_{k \leq d} S^k(\mathfrak{g}), \quad S_d(Z_p) := q \left(S_d(\mathfrak{g})\right), \quad \mathcal{I}_d(Z_p) := \mathcal{I}(Z_p) \cap S_d(\mathfrak{g}),
\end{align*} where $S^k(\mathfrak{g})$ consists of degree $k$ homogeneous polynomials.  Then, using \eqref{eq:seq}, we deduce a short exact sequence of graded pieces of $G$-algebra $S(\mathfrak{g})$ as follows: \begin{align}
0\longrightarrow \mathcal{I}_d(Z_P)\xrightarrow{\ \iota \ } S_d(\mathfrak{g})\xrightarrow{\ q\ } S_d(Z_P)  \longrightarrow 0 .\label{eq:finseq}
\end{align} Since $G$ is compact, by Lemma \ref{lem:induce}, the invariant of a finite-dimensional $G$-module will also induce a short exact sequence given by \begin{align}
    0\longrightarrow \mathcal{I}_d(Z_P)^G \xrightarrow{\ \iota \ } S_d(\mathfrak{g})^G \xrightarrow{\ q\ } S_d(Z_P)^G  \longrightarrow 0 \label{eq:finiteG}
\end{align} with $q:S_d(\mathfrak{g})^G \rightarrow S_d(Z_P)^G $ being surjective. 

We now use the degree filtration to extend the surjectivity of $q$ from the finite-dimensional truncations $S_d(\mathfrak{g})^G\to S_d(Z_P)^G$ to the full graded algebra map $q:S(\mathfrak{g})^G\to S(Z_P)^G$. Let $h' = h\vert_{Z_P} \in S(Z_P)^G$ for some $h \in S(\mathfrak{g})$. For fixed $d$ with $h' \in S_d(Z_P)^G$, since $q$ in \eqref{eq:finiteG} is surjective, there exists a $C_d \in S_d(\mathfrak{g})^G$ such that $q(C_d) = h'$. That is, $C_d\vert_{Z_P} = h'$. Using the inclusion $S_d(\mathfrak{g})^G \subset S(\mathfrak{g})^G$ and $S_d(Z_P)^G \subset S(Z_P)^G$, the identity $C_d\vert_{Z_P} = h'$ still holds in $S(Z_P)$. Then, taking $C = C_d \in S(\mathfrak{g})^G$, we have $q(C) = h' \in S(Z_P)^G$. Hence, $q:S(\mathfrak{g})^G\to S(Z_P)^G$ is surjective, and Lemma \ref{lem:induce} can be applied to the whole algebra. Therefore, the induced sequence \begin{align}
0\longrightarrow \mathcal{I}(Z_P)^G\xrightarrow{\,\iota\,} S(\mathfrak{g})^G \xrightarrow{\,q\, }S(Z_P)^G\longrightarrow 0 \label{eq:induceseq}
\end{align} is exact.  In particular, the restriction map $q$ is surjective, so there exists $C \in S(\mathfrak{g})^G$ with $C\vert_{Z_P}=h\vert_{Z_P} $.  Since $P(T^*M)=\mathrm{Ad}(G)(\mathfrak{m}- \varepsilon W)$, for all $(g,X) \in T^*M$, we have \begin{align*}
    C(P(g,X)) = C(\mathrm{Ad}(g)\xi) = h(\mathrm{Ad}(g)\xi) = \theta(\xi).
\end{align*} It follows that \begin{align*}
P^*C = P^*h = \varphi.
\end{align*} In particular, evaluating at $g =e$ yields $C(\xi)= h(\xi) = \theta(\xi)$ for all $\xi\in\mathfrak{m} - \varepsilon W$. We thus obtain $C \in S(\mathfrak{g})^G$ such that $C(\zeta) = h(\zeta)$ for all $\zeta \in Z_P$. In particular, $\mathrm{Res}_W(C) =\mathrm{Res}_W(h) $ on $\mathfrak{m} -\varepsilon W$. Hence, $\mathrm{Res}_W(C) = \theta$ and \begin{align*}
\varphi = P^*C= \pi_{\mathfrak{m}}^*\mathrm{Res}_W(C) \in P^*(\mathrm{Im}\mathrm{Res}_W).
\end{align*} Thus, $\mathfrak{F}_1^{\mathrm{poly}}\cap \mathfrak{F}_2^{\mathrm{poly}}\subseteq P^*(\mathrm{Im}\mathrm{Res}_W)$.

For the reverse inclusion, let $C \in S(\mathfrak{g})^G$ and put $\theta:=\mathrm{Res}(C)\in S(\mathfrak{m} -\varepsilon W)^A$. Then, for all $(g,X)$ with $\xi=X-\varepsilon W$, the pullback of these generators through the moment map $P$ is given by \begin{align}
P^*C(g,X) = C\bigl(\mathrm{Ad}(g)\,\xi\bigr)=C(\xi)=\theta(\xi)=\pi_{\mathfrak{m}}^*\theta(g,X), \label{eq:equalff}
\end{align}  so $P^*C\in \mathfrak{F}_1^{\mathrm{poly}}\cap \mathfrak{F}_2^{\mathrm{poly}}$, and therefore $P^*(\mathrm{Im}\mathrm{Res}_W)\subseteq \mathfrak{F}_1^{\mathrm{poly}}\cap \mathfrak{F}_2^{\mathrm{poly}}$. The equality \eqref{eq:equalff} implies that $\theta$ is the restriction of $C$ in the subspace $\mathfrak{m}$. Therefore, we have explicitly constructed a polynomial $C\in S(\mathfrak{g})^G$ such that $\theta=\mathrm{Res}(C)$, and thus $\theta\in\mathrm{Im}\mathrm{Res}_W$. Consequently, the original polynomial $\varphi$ is equal to $P^*(C) = P^*(\theta)$, which shows that all the elements $\varphi \in R_0$ are in the image of the pullback of the restriction map. Thus, we have shown inclusion $R_0 \subseteq P^*(\mathrm{Im} \mathrm{Res}_W)$.
\end{proof}
\begin{remark}
    From the proof of Proposition \ref{prop:inter}, the induced short exact sequence \eqref{eq:induceseq} shows that $$S(\mathfrak{g})^G/\mathcal{I}(Z_P)^G \cong S(Z_P)^G.$$ This also implies that $R_0 = S(Z_P)^G$.
\end{remark}
 \begin{corollary}\label{lem:domainstorsionfree}
Let $R_0,\mathfrak{F}_1^{\mathrm{poly}} $ and $\mathfrak{F}_2^{\mathrm{poly}}$ be the same as defined above.  Then $\mathfrak{F}_1^{\mathrm{poly}}$, $\mathfrak{F}_2^{\mathrm{poly}}$, and $R_0$ are finitely generated integral domains. Moreover, the inclusions \begin{align*}
\iota_1: R_0\hookrightarrow \mathfrak{F}_1^{\mathrm{poly}},\qquad \iota_2: R_0\hookrightarrow \mathfrak{F}_2^{\mathrm{poly}}
\end{align*} make each factor a torsion free $R_0$-module.
\end{corollary}

\begin{proof}
Write $Z_P =P(T^*M) = \mathrm{Ad}(G)(\mathfrak{m}-\varepsilon W) \subset \mathfrak{g}$. By Lemma \ref{lem:property}, we have a morphism \begin{align*}
\Phi :\ G \times(\mathfrak{m}-\varepsilon W) \longrightarrow \mathfrak{g},\quad (g,\xi)\longmapsto \mathrm{Ad}(g)\xi
\end{align*} that has an irreducible image.\footnote{Here irreducible is in the Zariski sense: $G\times(\mathfrak{m}-\varepsilon W)$ is irreducible, and the image of an irreducible variety under a morphism is irreducible (equivalently, its Zariski closure is irreducible).} Therefore, its ideal $\mathcal{I}(Z_P)\subset S(\mathfrak{g})$ is prime, and $\mathbb{R}[Z_P]:=S(\mathfrak{g})/\mathcal{I}(Z_P)$ is a domain. By Lemma \ref{lem:poly}, we have $\mathfrak{F}_1^{\mathrm{poly}}=P^*S(\mathfrak{g}) \cong \mathbb{R}[Z_P]$. Therefore, $\mathfrak{F}_1^{\mathrm{poly}}$ is a domain. On the other hand, the representative-level map $\widetilde{\pi}_{\mathfrak{m}}:G\times\mathfrak{m}\to\mathfrak{m}-\varepsilon W$ is surjective. Therefore, if $\pi_{\mathfrak{m}}^*\theta=0$ for $\theta\in S(\mathfrak{m}-\varepsilon W)^A$, then $0=q^*(\pi_{\mathfrak{m}}^*\theta)=\theta\circ\widetilde{\pi}_{\mathfrak{m}}$, and surjectivity gives $\theta=0$. Thus $\pi_{\mathfrak{m}}^*$ is injective on $S(\mathfrak{m}-\varepsilon W)^A$. By Lemma \ref{lem:poly} again, $\mathfrak{F}_2^{\mathrm{poly}}=\pi_{\mathfrak{m}}^*S(\mathfrak{m} -\varepsilon W)^A\cong S(\mathfrak{m})^A$, a subring of the domain $S(\mathfrak{m})$. Hence, $\mathfrak{F}_2^{\mathrm{poly}}$ is a domain.

By Proposition \ref{prop:inter}, we have \begin{align*}
R_0=\mathfrak{F}_1^{\mathrm{poly}}\cap\mathfrak{F}_2^{\mathrm{poly}} =P^* \big(\mathrm{Im}\mathrm{Res}\big) \cong \mathbb{R}[Z_P]^G.
\end{align*} Therefore, $R_0$ is a subring of the domain $\mathbb{R}[Z_P]^G$ and hence is itself a domain.

Finally, by the definition of torsion free: If $0\neq r\in R_0$ and $x\in \mathfrak{F}_i^{\mathrm{poly}}$ satisfy $rx=0$ in $\mathfrak{F}_i^{\mathrm{poly}}$, then since $\mathfrak{F}_i^{\mathrm{poly}}$ is a domain, we must have $x=0$. Thus, each $\mathfrak{F}_i^{\mathrm{poly}}$ is torsion free over $R_0$.
\end{proof}

 \begin{remark}
 \label{re:poisscenter}
By Proposition \ref{prop:inter} and Corollary \ref{lem:domainstorsionfree}, we have an identification  $ R_0      \cong \mathbb{R}[Z]^G$, i.e., $R_0$ consists of the restrictions to $Z_P=\mathrm{Ad}(G)(\mathfrak{m}-\varepsilon W)$ of the global $G$-invariants on $\mathfrak{g}^*$.\footnote{In particular, $R_0$ is a quotient of $S(\mathfrak{g})^G$ by the ideal of invariant polynomials vanishing on $Z$. It is not literally a subring of $S(\mathfrak{g})^G$.} With the block Poisson bracket on  \begin{align*}
    \widehat{\mathcal{A}} :=  \mathfrak{F}_1^{\mathrm{poly}} \otimes_{R_0} \mathfrak{F}_2^{\mathrm{poly}},  
\end{align*} every $r\in R_0$ Poisson commutes with both blocks. Thus,   $R_0  \subset \mathcal{Z}( \widehat{\mathcal{A}} )$.  Another side of inclusion $\mathcal{Z}( \widehat{\mathcal{A}} ) \subset R_0$ and equality  \begin{align*}
   \mathcal{Z}( \widehat{\mathcal{A}} ) =  R_0 
\end{align*} (on a Zariski open dense subset) is non-trivial and is proved later in Theorem \ref{thm:dimea}. 
\end{remark}

\begin{corollary}
\label{cor:A-torsion free}
The algebra $\mathcal{A}=\textbf{Alg} \left\langle \mathfrak{F}_1^{\mathrm{poly}}\cup \mathfrak{F}_2^{\mathrm{poly}} \right\rangle\subset \mathcal{O}(T^*M)$ is an integral domain. In particular, $\mathcal{A}$ is $R_0$ torsion free.
\end{corollary}

\begin{proof}
Let $L=\mathrm{Frac}(\mathfrak{F}_1^{\mathrm{poly}})$ and $M = \mathrm{Frac}(\mathfrak{F}_2^{\mathrm{poly}})$ inside $ \mathrm{Frac}\left(\mathcal{O}(T^*M)\right)$, and let $LM = \mathrm{Frac} \left(\textbf{Alg} \langle L \cup M \rangle \right)$ be their compositum\footnote{We denote by $LM$ the compositum of $L$ and $M$ as the smallest subfield containing both $L$ and $M$. That is, $LM := \bigcap \{K \subset \mathrm{Frac}(\mathcal{O}(T^*M)): \, K \text{ is a subfield and } L,M \subset K\}$.}. See, for instance, \cite[Chapter 5]{LangAlgebra} for more detailed discussion of the compositum of two fields. Since $\mathfrak{F}_1^{\mathrm{poly}}\subset L$ and $\mathfrak{F}_2^{\mathrm{poly}}\subset M$, the subalgebra of $\mathrm{Frac}\left(\mathcal{O}(T^*M)\right)$ generated by $\mathfrak{F}_1^{\mathrm{poly}}$ and $\mathfrak{F}_2^{\mathrm{poly}}$ satisfies \begin{align*}
\mathcal{A}=\textbf{Alg}\langle \mathfrak{F}_1^{\mathrm{poly}}\cup \mathfrak{F}_2^{\mathrm{poly}}\rangle \subset LM\,.
\end{align*} Since $LM$ is a field, any subring of $LM$ is an integral domain. Hence, $\mathcal{A}$ is an integral domain.

By Corollary \ref{lem:domainstorsionfree}, the inclusions $R_0\hookrightarrow \mathfrak{F}_i^{\mathrm{poly}}$ are injective for $i = 1,2$. Therefore, the induced map $R_0\hookrightarrow \mathcal{A}$ is injective, so we may regard $R_0$ as a subring of the domain $\mathcal{A}$. If $0 \neq r \in R_0$ and $a \in \mathcal{A}$ satisfy $ra = 0$ in $\mathcal{A}$, then $r\neq 0$ in $\mathcal{A}$ and, since $\mathcal{A}$ has no zero divisors, it follows that $a = 0$. Thus, $\mathcal{A}$ is $R_0$ torsion free.
\end{proof}

\medskip

Let $\mathcal{A} = \textbf{Alg} \left\langle F_1^{\mathrm{poly}} \cup F_2^{\mathrm{poly}} \right \rangle$. In what follows, we will provide a more detailed description of the union algebra $\mathcal{A}$.   

We begin by analyzing the fiber tensor algebra \begin{align*}
\widehat{\mathcal{A}}:=\mathfrak{F}_1^{\mathrm{poly}}\otimes_{R_0}\mathfrak{F}_2^{\mathrm{poly}}, \qquad  \mu:\widehat{\mathcal{A}}\longrightarrow \mathcal{O}(T^*M),
\end{align*} obtained via the canonical multiplication map. From Appendix \ref{appA}, we deduce that $\widehat{\mathcal{A}}$ is a Poisson algebra with a Poisson bracket $\{\cdot,\cdot\}_{\widehat{\mathcal{A}}}$.


Let \begin{align} \mu:\widehat{\mathcal{A}}\longrightarrow \mathcal{O}(T^*M),\qquad \mu\Bigl(\sum_i f_i\otimes\phi_i\Bigr)=\sum_i f_i\cdot\phi_i   \label{eq:defmu} \end{align} be an algebra homomorphism, where the product on the right is the pointwise product in $\mathcal{O}(T^*M)$. Here we consider $\mathfrak{F}_1^{\mathrm{poly}},\mathfrak{F}_2^{\mathrm{poly}}\subset \mathcal{O}(T^*M)$. We can further show that $\mu$ is a Poisson homomorphism. Although $\ker\mu$ could initially be nonzero, the arguments in Appendix \ref{appB} below will show that, in our setting, $\ker\mu=0$. 

\begin{proposition}
    The map $\mu$ is a homomorphism of Poisson algebras \begin{align*}
   \mu(g g')=\mu(g)\mu(g'),\qquad \mu\bigl(\{g,g'\}_{\widehat{\mathcal{A}}}\bigr)=\{\mu(g),\mu(g')\}_\varepsilon \quad\text{for all }g : = f \otimes \phi,g': = f' \otimes \phi' \in\widehat{\mathcal{A}}. 
\end{align*} In particular, its kernel $\mathcal{I} := \ker \mu$ is a Poisson ideal of $\left(\widehat{\mathcal{A}},\{\cdot ,\cdot \}_{\widehat{\mathcal{A}}}\right)$.
\end{proposition}
 
 \begin{proof}
 We first show the algebra homomorphism. For any $g = f \otimes \phi,g' = f' \otimes \phi' \in \widehat{\mathcal{A}}$, by the definition of \eqref{eq:defmu},  \begin{align*}
     \mu(gg') = \mu(f f' \otimes \phi \phi') = (ff')(\phi\phi') = \mu(g) \mu(g').
 \end{align*}  On the other hand, we show that the Poisson property holds. In the twisted Poisson bracket, \begin{align}
     \{\mu(g),\mu(g')\}_\varepsilon = \{f\phi,f'\phi'\}_\varepsilon = \{f,f'\}_\varepsilon \phi\phi' +  \{f,\phi'\}_\varepsilon \phi f' +  \{\phi,f'\}_\varepsilon f \phi' +  ff' \{\phi,\phi'\}_\varepsilon   \label{eq:twistedmu} 
 \end{align}  It suffices to check for simple tensors. By Lemma \ref{lem:interzero}, we have $\{\mathfrak{F}_1^{\mathrm{poly}},\mathfrak{F}_2^{\mathrm{poly}}\}_\varepsilon=0$. Then the identity \eqref{eq:twistedmu} becomes $ \{\mu(g),\mu(g')\}_\varepsilon =  \{f,f'\}_\varepsilon \phi\phi'  +  ff' \{\phi,\phi'\}_\varepsilon $. On the other hand, using the Leibniz rule for $\{\cdot ,\cdot \}_\varepsilon$ and restricted to $\{\cdot,\cdot\}\vert_{\widehat{\mathcal{A}}}$, we have \begin{align*}
    \mu\left(\{g, g'\}_{\widehat{\mathcal{A}}}\right) =\{f,f'\}_\varepsilon\,\phi\phi' +  ff'\,\{\phi,\phi'\}_\varepsilon =\{f\phi,\ f'\phi'\}_\varepsilon.
\end{align*} We then show that $\mathcal{I}$ is an ideal by using the definition of the Poisson ideal. Indeed, if $a\in\ker\mu$ and $b\in\widehat{\mathcal{A}}$, then
$\mu(\{a,b\}_{\widehat{\mathcal{A}}})=\{\mu(a),\mu(b)\}_\varepsilon=0$. 
 \end{proof} 
 
\begin{remark}
\label{rem:poisclosed}
  Let $\mathcal{A}:=\mu(\widehat{\mathcal{A}})\subset \mathcal{O}(T^*M)$. Define $\{\cdot ,\cdot \}_{\mathcal{A}}$ as the unique Poisson bracket on $\mathcal{A}$ such that the inclusion \begin{align*}
    \iota: (\mathcal{A},\{\cdot ,\cdot \}_{\mathcal{A}})\hookrightarrow (\mathcal{O}(T^*M),\{\cdot ,\cdot \}_\varepsilon)
\end{align*} is a Poisson algebra embedding. Equivalently, \begin{align*}
  (\mathcal{A},\{\cdot ,\cdot \}_{\mathcal{A}})\cong\bigl(\widehat{\mathcal{A}}/\ker\mu,\ \overline{\{\cdot ,\cdot \}}_{\widehat{\mathcal{A}}}\bigr).  
\end{align*} In particular, for any $h,h'\in S(\mathfrak{g})$ and $\theta,\theta'\in S(\mathfrak{m} - \varepsilon W)^A$, together with Lemma \ref{lem:poly} and Lemma \ref{lem:interzero}, we deduce that the Poisson bracket on $\mathcal{A}$ is closed in the following way:  \begin{align}
\nonumber
    \{P^*h, P^*h'\}_{\mathcal{A}}&=P^*\left(\{h,h'\}_1\right) ,\\[2mm]
\{\pi_{\mathfrak{m}}^*\theta, \pi_{\mathfrak{m}}^*\theta'\}_{\mathcal{A}}&=\ \pi_{\mathfrak{m}}^*\Bigl( \{\theta,\theta'\}_2 \Bigr) ,\\[2mm]
\nonumber
\{P^*h, \pi_{\mathfrak{m}}^*\theta\}_{\mathcal{A}}&=0.
\end{align}   Here $\{\cdot,\cdot\}_1$ defines the canonical Lie-Poisson bracket, and $\{\cdot,\cdot\}_2$ defines the Poisson structure on $S(\mathfrak{m} - \varepsilon W)^A$.
\end{remark}

To complete the argument and form a superintegrable chain, we aim to prove that $\ker \mu = \{0\}$, i.e., that the natural map $\mathcal{A}\to\widehat{\mathcal{A}}$ is injective and hence $\mathcal{A}=\widehat{\mathcal{A}}$.

A key point in what follows is that $\widehat{\mathcal{A}}$ is obtained from $\mathfrak{F}_1^{\mathrm{poly}}$ by tensoring over the Poisson center $R_0$ (for instance, $\widehat{\mathcal{A}} = \mathfrak{F}_1^{\mathrm{poly}}\otimes_{R_0}\mathfrak{F}_2^{\mathrm{poly}}$). Thus, to control $\ker\mu$, we need good functorial properties of the operation $-\otimes_{R_0}\mathfrak{F}_1^{\mathrm{poly}}$: in particular, we want tensoring over $R_0$ not to destroy injectivity.

This is why we now pause to understand $\mathfrak{F}_1^{\mathrm{poly}}$ not just abstractly as the coordinate ring of $Z_P$, but as an algebra over $R_0\cong \mathbb{R}[Z_P]^G$. Concretely, we will identify $\mathfrak{F}_1^{\mathrm{poly}}$ as the pullback (i.e. base change) of the adjoint quotient map. This description will be used later, and it has the immediate consequence that $\mathfrak{F}_1^{\mathrm{poly}}$ is flat\footnote{Let $R$ be a commutative ring and $M$ an $R$-module. Recall that $M$ is \emph{flat over $R$} if the functor $-\otimes_R M$ is exact, i.e. it sends every short exact sequence of $R$-modules to a short exact sequence (equivalently, it preserves injections).  See e.g. \cite[Chapter 2, Section 2]{AtiyahMacdonald}.} (indeed, free) over $R_0$.

We emphasize this flatness as it ensures that tensoring over $R_0$ preserves injections: if $M\hookrightarrow N$ is an injection of $R_0$-modules, then $M\otimes_{R_0}\mathfrak{F}_1^{\mathrm{poly}}\hookrightarrow N\otimes_{R_0}\mathfrak{F}_1^{\mathrm{poly}}$ is still injective. In particular, tensoring over $R_0$ does not introduce new relations, and if $R_0$ is a domain, then $\mathfrak{F}_1^{\mathrm{poly}}$ is $R_0$ torsion free.

\begin{proposition}
\label{prop:F1basechange}
Let $C:= S(\mathfrak{g})^G$, let $\chi:\mathfrak{g}\to \mathfrak{g}//G=\mathrm{Spec}\,C$ be the adjoint quotient, and let $\rho:C\twoheadrightarrow R_0$ be the quotient map induced by the restriction of invariants to $Z_P$. Let $\Gamma:= \mathrm{Spec}\,R_0$. 
Then there is a canonical $R_0$-algebra isomorphism
\begin{align*}
\mathfrak{F}_1^{\mathrm{poly}}\cong S(\mathfrak{g})\otimes_C R_0.
\end{align*}
In particular, $\mathfrak{F}_1^{\mathrm{poly}}$ is a free, hence flat, $R_0$-module.
\end{proposition}

\begin{proof}
By Remark \ref{rem:F1poly}, we have $\mathfrak{F}_1^{\mathrm{poly}}\cong \mathbb{R}[Z_P]$. Let $J:= \ker \rho \subset C$, so that $\rho$ induces an isomorphism $R_0\cong C/J$ and hence $\Gamma=\mathrm{Spec} \, R_0\cong \mathrm{Spec} (C/J)$ is the closed subvariety of $\mathfrak{g}//G= \mathrm{Spec}  C$.

Since all varieties here are affine, the fiber product is computed on coordinate rings by tensor product: \begin{align*}
\mathfrak{g}\times_{\mathfrak{g}//G}\Gamma &=\mathrm{Spec} \, S(\mathfrak{g})\times_{\mathrm{Spec} \, C}\mathrm{Spec} \, R_0  =\mathrm{Spec} \big(S(\mathfrak{g})\otimes_C R_0\big).
\end{align*} Using $R_0\cong C/J$, we obtain \begin{align*}
S(\mathfrak{g})\otimes_C R_0\cong S(\mathfrak{g})\otimes_C (C/J)\cong S(\mathfrak{g})/J S(\mathfrak{g}),
\end{align*} where the last isomorphism is the standard identification $A\otimes_C (C/J)\cong A/JA$ for any $C$-algebra $A$. Therefore, \begin{align*}
\mathfrak{g}\times_{\mathfrak{g}//G}\Gamma = \mathrm{Spec}\big(S(\mathfrak{g})\otimes_C R_0\big) = \mathrm{Spec}\big(S(\mathfrak{g})/J S(\mathfrak{g})\big).
\end{align*} Here $J S(\mathfrak{g})$ denotes the extension of the ideal $J\subset C=S(\mathfrak{g})^G$ to $S(\mathfrak{g})$. That is, \begin{align*}
    J S(\mathfrak{g}) := S(\mathfrak{g})\cdot J = \Big\{\sum_{i=1}^n f_i j_i : n\ge 1,\ f_i\in S(\mathfrak{g}),\ j_i\in J\Big\}.
\end{align*} We claim that the underlying closed set of this fiber product is exactly $Z_P$. The inclusion $Z_P\subset \chi^{-1}(\Gamma)$ is immediate. Conversely, let $x\in \chi^{-1}(\Gamma)$. Then $\chi(x)=\chi(\xi)$ for some $\xi\in \mathfrak{m}-\varepsilon W$. Because $G$ is compact, every adjoint orbit in $\mathfrak{g}$ is closed, and Section \ref{subsec:geomi} shows that the quotient map $\chi$ separates closed $G$-orbits. Hence $x\in \mathrm{Ad}(G)\cdot \xi\subset Z_P$. Therefore, $\chi^{-1}(\Gamma)=Z_P$.

By Kostant's theorem, the invariant ring $C=S(\mathfrak{g})^G$ is a polynomial algebra, and the adjoint quotient morphism $\chi:\mathfrak{g}\to \mathfrak{g}//G= \mathrm{Spec} \, C$ is flat. Equivalently, $S(\mathfrak{g})$ is a flat $C$-module. Moreover, each fiber $\chi^{-1}(y)$ is geometrically reduced, meaning that for every field extension $k\subset k'$, the base change $\chi^{-1}(y)\times_{\mathrm{Spec} \, k}\mathrm{Spec} \, k'$ is reduced. In particular, flatness together with the geometric reducedness of the fibers implies that base change along any reduced $C$-algebra preserves the reducedness of the total space. Since $R_0$ is reduced by Corollary \ref{lem:domainstorsionfree}, the base-changed fiber product $\mathfrak{g}\times_{\mathfrak{g}//G}\Gamma$ is reduced. Consequently, its coordinate ring is the coordinate ring of the reduced closed subset $Z_P$, so \begin{align*}
\mathbb{R}[Z_P]=\mathcal{O}\big(\mathfrak{g}\times_{\mathfrak{g}//G}\Gamma\big)=S(\mathfrak{g})\otimes_{S(\mathfrak{g})^G}R_0.
\end{align*}
Using $\mathfrak{F}_1^{\mathrm{poly}}\cong \mathbb{R}[Z_P]$ again, we obtain the asserted isomorphism. Since $S(\mathfrak{g})$ is free over $S(\mathfrak{g})^G$, base change along $S(\mathfrak{g})^G\to R_0$ shows that $\mathfrak{F}_1^{\mathrm{poly}}$ is free, hence flat, over $R_0$.
\end{proof}

\begin{corollary}
\label{cor:widehatAtofree}
The tensor product $\widehat{\mathcal{A}}:=\mathfrak{F}_1^{\mathrm{poly}}\otimes_{R_0}\mathfrak{F}_2^{\mathrm{poly}}$ is $R_0$ torsion free.
\end{corollary}

\begin{proof}
By Proposition \ref{prop:F1basechange}, the algebra $\mathfrak{F}_1^{\mathrm{poly}}$ is flat over $R_0$. By Corollary \ref{lem:domainstorsionfree}, the algebra $\mathfrak{F}_2^{\mathrm{poly}}$ is torsion free over $R_0$. Let $0\neq r\in R_0$. Then multiplication by $r$ gives an injective map
\begin{align*}
0\longrightarrow \mathfrak{F}_2^{\mathrm{poly}}\xrightarrow{\cdot r}\mathfrak{F}_2^{\mathrm{poly}}.
\end{align*}
Tensoring with the flat $R_0$-module $\mathfrak{F}_1^{\mathrm{poly}}$ preserves injectivity, so we obtain
\begin{align*}
0\longrightarrow \mathfrak{F}_1^{\mathrm{poly}}\otimes_{R_0}\mathfrak{F}_2^{\mathrm{poly}}\xrightarrow{\cdot r}\mathfrak{F}_1^{\mathrm{poly}}\otimes_{R_0}\mathfrak{F}_2^{\mathrm{poly}}.
\end{align*}
Thus multiplication by every nonzero $r\in R_0$ is injective on $\widehat{\mathcal{A}}$, which is exactly the $R_0$ torsion freeness of $\widehat{\mathcal{A}}$.
\end{proof}

 By Proposition \ref{prop:algdisjoint}, the canonical field-level multiplication map is injective:
\begin{align*} L\otimes_K M\longrightarrow \mathrm{Frac}\,\mathcal{O}(T^*M),\qquad \ell\otimes m\longmapsto \ell m . \end{align*} Hence, Corollary \ref{cor:norela} gives \begin{align*}
S^{-1}\ker\mu=0 .
\end{align*} 
By Corollary \ref{cor:norela} (applied with $R=R_0$, $A=\mathfrak{F}_1^{\mathrm{poly}}$, $B=\mathfrak{F}_2^{\mathrm{poly}}$, and $C=\mathcal{O}(T^*M)$), the injectivity of the field-level map $L\otimes_K M\to \mathrm{Frac}\,\mathcal{O}(T^*M)$ implies that $S^{-1}(\ker\mu)=0$. 
On the other hand, Proposition \ref{prop:F1basechange} and Corollary \ref{cor:widehatAtofree} show that $\widehat{\mathcal{A}}  =\mathfrak{F}_1^{\mathrm{poly}} \otimes_{R_0}\mathfrak{F}_2^{\mathrm{poly}}$ is $R_0$ torsion free. Hence, the only $R_0$ torsion submodule of $\widehat{\mathcal{A}}$ is $0$, so necessarily $\ker\mu=0$. Consequently,
\begin{align*}
\mathcal{A}\cong\widehat{\mathcal{A}}=\mathfrak{F}_1^{\mathrm{poly}}\otimes_{R_0}\mathfrak{F}_2^{\mathrm{poly}}.
\end{align*}
In particular, the multiplication map $\mu:\widehat{\mathcal{A}}\to \mathcal{O}(T^*M)$ is injective. By Definition \ref{def:maps}(i), we have $\mathcal{O}(T^*M)=\bigl(\mathbb{R}[G]\otimes S(\mathfrak{m})\bigr)^A\cong \mathbb{R}[G\times\mathfrak{m}]^A$. Since $G\times\mathfrak{m}$ is an irreducible affine variety, its coordinate ring $\mathbb{R}[G\times\mathfrak{m}]$ is an integral domain (see \cite[Chapter 1, Proposition 3.1]{hartshorne2013algebraic}). Taking $A$-invariants yields a subring, so $\mathcal{O}(T^*M)$ is again an integral domain (See, for instance,  \cite[Section 1]{AtiyahMacdonald}). Therefore, $\widehat{\mathcal{A}}$ is an integral domain as a subring of the integral domain $\mathcal{O}(T^*M)$. In particular, $\widehat{\mathcal{A}}$ is reduced. Hence, $\sqrt{\widehat{\mathcal{A}}}=\widehat{\mathcal{A}}$, and thus $\delta_{\mathrm{alg}}= \dim \ker \mu =0$ by Definition \ref{def:algetensorp}.
Moreover, Proposition \ref{prop:fieldtoring} together with Proposition \ref{prop:algdisjoint} and Corollary \ref{cor:widehatAtofree} gives \begin{align*}
\mathrm{trdeg}\,\mathcal{A}=\mathrm{trdeg}\,\mathfrak{F}_1^{\mathrm{poly}}+\mathrm{trdeg}\,\mathfrak{F}_2^{\mathrm{poly}}-\mathrm{trdeg}\,R_0.
\end{align*} In particular, $\delta_{\mathrm{alg}}= 0$. Together with all constructions above, we conclude the following tensor product statement. See Appendix \ref{appB} for the field-intersection and torsion freeness proof. Together with all constructions above, we conclude the following corrected form of the tensor product statement. 

\begin{theorem}
\label{thm:tensorgeneric}
 Let $G$ be a compact, simply connected semisimple Lie group, and let $A=\exp(\ker\mathrm{ad}(W))$ be defined as above for a fixed $W \in \mathfrak{g}$. 
Then \begin{align*} \mathrm{trdeg}\,\mathcal{A}=\mathrm{trdeg}\,\mathfrak{F}_1^{\mathrm{poly}}+\mathrm{trdeg}\,\mathfrak{F}_2^{\mathrm{poly}}-\mathrm{trdeg}\,R_0 . \end{align*}
Equivalently, $\delta_{\mathrm{alg}}=0$, 
and\begin{align*}
\mathcal{A}\cong \mathfrak{F}_1^{\mathrm{poly}}\otimes_{R_0}\mathfrak{F}_2^{\mathrm{poly}} . 
\end{align*}
\end{theorem}

With all the information above, we conclude the spectrum of the union of two Poisson polynomial algebras $\mathfrak{F}_1^{\mathrm{poly}}$ and $\mathfrak{F}_2^{\mathrm{poly}}$ in the following corollary. 
\begin{corollary}
\label{cor:union}
Let $\mathrm{Res}_W: S(\mathfrak{g})^G \rightarrow S(\mathfrak{m})^A $. Then there is an isomorphism of affine varieties $$ \mathrm{Spec}\,\mathcal{A}=\mathrm{Spec}(\mathfrak{F}_1^{\mathrm{poly}}\otimes_{R_0}\mathfrak{F}_2^{\mathrm{poly}})=\mathrm{Ad}(G)(\mathfrak{m}-\varepsilon W)\times_{\mathrm{Im}\mathrm{Res}_W} \left(\mathfrak{m} - \varepsilon W\right)//A , $$ with dimension $\dim \mathrm{Spec} \,\mathcal{A}=n-r+\rho_A$, where $n=\dim\mathfrak{g}$, $r=\mathrm{rank}(G)$, $s=\mathrm{trdeg} \, \mathrm{Im}\mathrm{Res}_W$, and $\rho_A=\mathrm{trdeg} \, S(\mathfrak{m})^A  = \mathrm{trdeg} \, S(\mathfrak{m} -\varepsilon W)^A$.
\end{corollary}

By Remark \ref{re:poisscenter}, we already have $R_0\subset \mathcal{Z}(\mathcal{A})$. The next theorem updates that argument to $\mathcal{Z}(\mathcal{A})=R_0$ on $U$. Here $U$ is the regular stratum in $T^*M$ defined in \eqref{eq:regularst}. In particular, Corollary \ref{re:poisscenter} below verifies the factor centrality used in Lemma \ref{lem:factor-central}. 
 Starting with the presentation of $R_0$. Let $J_1,\ldots,J_s\in \mathrm{Im}\mathrm{Res}_W$ be algebraically independent such that $s=\mathrm{trdeg} R_0$ with $R_0:=P^{*}(\mathrm{Im}\mathrm{Res}_W)$. (See the discussion of 
$\mathrm{Res}_W$ and Theorem \ref{thm:dimea}.)  These $J_i$ generate a polynomial subalgebra $\mathbb{R}[J_1,\ldots,J_s]\subset R_0$. In general, 
$R_0$ may be a finitely generated algebra over this subalgebra.  Equivalently, there exist additional generators $Y_1,\ldots,Y_m$ and relations $I$ such that \begin{align*}
    R_0 \;\cong\; \mathbb{R}[J_1,\ldots,J_s,Y_1,\ldots,Y_m]\big/ I.
\end{align*} in the regular torus case (e.g., $A=T$ and $W$ regular), we take $m=0$ and $I=0$ such that  $R_0=\mathbb{R}[J_1,\ldots,J_s]$. Define the  integral map  \begin{align*}
    \pi_1  :  T^*M \longrightarrow \mathrm{Spec} \, \mathcal{A} \quad\text{ by }\quad \pi_1^*(\widehat{f}) = f\quad \text{ for any }f \in \mathcal{A},
\end{align*}  where $\widehat{f}$ is the coordinate function of $\mathrm{Spec} \, \mathcal{A}$ corresponding to $f$. Equivalently, if we choose a generating list, \begin{align}
    \Phi =\bigl(\underbrace{P^*h_1,\ldots,P^*h_{n-r}}_{\text{generators in $P^*S(\mathfrak{g})$}};\;\underbrace{P^*C_{i_1},\ldots,P^*C_{i_s}}_{\text{generators in $R_0$}};\;\underbrace{\pi_\mathfrak{m}^*\theta_1,\ldots,\pi_\mathfrak{m}^*\theta_{\rho_A}}_{\text{generators in $\pi_\mathfrak{m}^*(S(\mathfrak{m})^A)$}}\bigr). \label{eq:generatorsmp}
\end{align}  Then $\pi_1=(\Phi_1,\ldots,\Phi_{n -r + s+ \rho_A})$. On the other hand, the inclusion $\iota: R_0 \hookrightarrow \mathcal{A}$ induces \begin{align*}
    \pi_2 :\mathrm{Spec} \, \mathcal{A} \longrightarrow \mathrm{Spec} \, R_0, \qquad \pi_2^*=\iota .
\end{align*} We define the composition maps by:  \begin{align*}
    q:= \pi_2 \circ \pi_1 :T^*M \longrightarrow \mathrm{Spec} \, R_0 .
\end{align*} On the regular locus $U$ defined in \eqref{eq:regularst}, the differential $d\pi_1$ has a constant rank, and $\ker dq$ is a smooth distribution.

\begin{remark}
Let us remark that, at this stage, we only use the inclusion $R_0\hookrightarrow\mathcal{A}$ to define \begin{align*}
   q:T^*M \xrightarrow{\ \pi_1\ } \mathrm{Spec}\,\mathcal{A}
\xrightarrow{\ \pi_2\ } \mathrm{Spec}\, R_0, 
\end{align*}  where $\pi_2$ is the morphism induced by $R_0\hookrightarrow\mathcal{A}$. We do not yet use any information about the Poisson center $\mathcal{B}=Z(\mathcal{A})$. In Theorem \ref{thm:dimea}, we will show that $\mathcal{B}=R_0$ on the regular locus, so that afterward, $\mathrm{Spec} R_0 = \mathrm{Spec}\mathcal{B}$ and $q$ coincide with the central projection in the sense of Definition \ref{def:superge}.
\end{remark}

 \begin{theorem}
 \label{thm:dimea}
        Let $G$ be a simply connected, compact and semisimple Lie group of dimension $n$ and rank $r$. Fix a closed subgroup $A\subset G$ with Lie algebra $\mathfrak{a}$ and decompose $\mathfrak{g}=\mathfrak{a}\oplus\mathfrak{m}$ orthogonally with respect to the Killing form $B$. Let $U$ be the regular stratum in $T^*M$ defined in \eqref{eq:regularst}. Fix $W \in \mathfrak{g}$ and put $\xi:=X-\varepsilon W\in\mathfrak{m} -\varepsilon W\subset\mathfrak{g}$.  Let \begin{align*}
    \mathrm{Res}_W: S(\mathfrak{g})^G\to S(\mathfrak{m})^A,\quad \mathrm{Res}_W(C)(X):=C(X-\varepsilon W), \qquad R_0:=P^*\big(\mathrm{Im}\mathrm{Res}_W\big),
\end{align*} and put $s =\mathrm{trdeg}(\mathrm{Im}\mathrm{Res}_W)$. Then, on $U$, \begin{align*} 
\mathcal{B} =R_0=P^*\big(\mathrm{Im}\mathrm{Res}_W\big)\quad\text{and}\quad \mathrm{rank}\, R_0=s.
\end{align*}  Moreover, $$\dim\mathrm{Spec} \,\mathcal{A}=(n-r)+\rho_{A}=2\dim\mathfrak{m}-s,\qquad \rho_{A}:=\mathrm{trdeg} \, S(\mathfrak{m})^A. $$ In particular, we have the following:

(i) if $A=T$ and $W$ is regular, then $s=r$, $\rho_{T}=\dim\mathfrak{m}-r$, $\dim\mathrm{Spec} \,\mathcal{A}=2\dim\mathfrak{m}-r$, and $\mathrm{rank}\, \mathcal{B}=r$;

(ii) if $A\neq T$ with $T\subsetneq A$ and $W$ is irregular while $\xi$ is regular, then $s< r$, $\rho_{A}=n-2\dim A+r-s$, $\dim\mathrm{Spec} \,\mathcal{A}=2\dim\mathfrak{m}-s$, and $\mathrm{rank}\, \mathcal{B}=s$.    
 \end{theorem}

 \begin{proof}
    (1) \textbf{Inclusion $R_0\subset \mathcal{B}$.}   For $C\in S(\mathfrak{g})^G$, by the definition of the Poisson center, $P^*C$ Poisson commutes both with $\mathfrak{F}_1^{\mathrm{poly}}$ and, as shown in Lemma \ref{lem:interzero}, with $\mathfrak{F}_2^{\mathrm{poly}}$. Thus, $P^*C\in \mathcal{B}$. If $\mathrm{Res}_W(C)=0$, then $P^*C\equiv 0$ on $\mathfrak{m}-\varepsilon W$. If $\mathrm{Res}_W(C)\neq0$, then for all $(g,X) \in T^*M$, we deduce the following: \begin{align*}
        P^*C(g,X)=C(\mathrm{Ad}(g)\xi)=C(\xi)=\mathrm{Res}_W(C)(X)=\pi_{\mathfrak{m}}^*\big(\mathrm{Res}_W(C)\big)(g,X)\in R_0.
    \end{align*} Hence $R_0\subset \mathcal{B}$.

 (2) \textbf{Rank verification.}
Fix $(g,X)\in U$ and put $\xi:=X-\varepsilon W$. Choose $Z_1,\ldots,Z_{n-r}\in\mathfrak{g}$ such that $\{[Z_i,\xi]\}_{i=1}^{n-r}$ is a basis of $[\mathfrak{g},\xi]=T_\xi(\mathrm{Ad}(G)\cdot\xi)$, and define the linear functions $h_i(Y):=B(Z_i,Y)$ such that $\nabla h_i\equiv Z_i$ is constant. Choose a transcendence basis $\theta_1,\ldots,\theta_{\rho_A}$ of $S(\mathfrak{m})^A$ and basic $G$-invariants $C_{i_1},\ldots,C_{i_s}\in S(\mathfrak{g})^G$ with \begin{align*}
\theta_\ell=\mathrm{Res}_W(C_{i_\ell})\qquad \ell=1,\ldots,s ,
\end{align*} such that $ \theta_1,\ldots,\theta_s $ are algebraically independent. Recall that \begin{align}
\Phi=\bigl(P^*h_1,\ldots,P^*h_{n-r}\,;\, P^*C_{i_1},\ldots,P^*C_{i_s}\,;\, \pi_{\mathfrak{m}}^*\theta_1,\ldots,\pi_{\mathfrak{m}}^*\theta_{\rho_A}\bigr). \label{eq:global}
\end{align} To conclude the proof, we compute the Jacobian of $\Phi$.

We now write down the Hamiltonian vector fields with respect to $\omega_\varepsilon$. For $f_\theta:=\pi_{\mathfrak{m}}^*\theta$ with $\theta\in S(\mathfrak{m} -\varepsilon W)^A$ (so $f_\theta([g,X])=\theta(\pi_\mathfrak{m}(\xi))$), solving  $df_\theta=\omega_\varepsilon(\cdot,X_{f_\theta})$ gives \cite{MR2141306} \begin{align}
X_{f_\theta}(g,X) = g_*\vert_X\left( (\nabla\theta(\xi))_{\mathfrak{m}} ,  -\frac{1}{2}[(\nabla\theta(\xi))_{\mathfrak{m}},X]_{\mathfrak{m}}-\varepsilon[W,(\nabla\theta(\xi))_{\mathfrak{m}}]  \right). \label{eq:vectorf2}
\end{align} For $I_h:= h \circ P$ with $h\in S(\mathfrak{g})$, put $\zeta:=\nabla h(\mathrm{Ad}(g)\xi)$. Recall that from Lemma \ref{lem:F1polypbrack}, for a linear functional $P_h(Y) = \langle \eta,Y\rangle$ with $\eta \in \mathfrak{g}^*$ and $Y \in \mathfrak{g}$, we obtain \begin{align}
    d (h \circ P)_{(g,X)}\left(v,w - \frac{1}{2}[v,X]\right) =  B\left(\zeta_\mathfrak{m},w -\frac{1}{2}[v,X]\right) - B([\zeta,w],v) \label{eq:differefirst}
\end{align} Then solving $dI_h = \omega_\varepsilon(\cdot, X_{I_h})$ yields the following: \begin{align}
X_{I_h}(g,X) = g_*\vert_X\left(- (\mathrm{Ad}(g^{-1})\zeta)_{\mathfrak{m}}\ ,\ [\mathrm{Ad}(g^{-1})\zeta,X]_{\mathfrak{m}}-\frac{1}{2}[(\mathrm{Ad}(g^{-1})\zeta)_{\mathfrak{m}},X]_{\mathfrak{m}}\right). \label{eq:vectorf1}
\end{align} In particular, if $h$ is $G$-invariant, then $\zeta=\mathrm{Ad}(g)\nabla h(\xi)$. Therefore, $\zeta = \nabla h(\xi)$, and the first component is $(\nabla h(\xi))_{\mathfrak{m}}$.  Restricting to base variations $(v,0)$ tangent to $G \cdot \xi$, by the chain rule and \eqref{eq:differefirst}, using $\nabla h_i\equiv Z_i$ with $i = 1,\ldots,n-r$, we deduce that \begin{align*}
d(P^*h_i)_{(g,X)}(v,0) = B\big(\nabla h_i(\mathrm{Ad}(g)\xi),\,\mathrm{Ad}(g)[v,\xi]\big) = -B\big([Z_i,\xi],\,v\big) \text{ for all } v \in \mathfrak{m}.
\end{align*} Thus, if $\sum_i a_i\,d(P^*h_i)_{(g,X)}(v,0) = 0$ for all $v \in \mathfrak{m}$ and $a_i \in \mathbb{R}$, then $\left(\sum_i a_i[Z_i,\xi]\right)_\mathfrak{m} = 0$, and by the choice of $Z_i$, we get $a_i=0$ for all $i$. Therefore, \begin{align}
\mathrm{rank}\bigl(d(P^*h_1),\ldots,d(P^*h_{n-r})\bigr) = n-r.
\end{align}  (i.e. the differentials are linearly independent, which means the Jacobian has full rank.) For the fiber part, letting $\Theta=(\theta_1,\ldots,\theta_{\rho_A})$, we have for $(0,\delta X)\in \mathfrak{m}\oplus\mathfrak{m}$, \begin{align*}
d(\pi_{\mathfrak{m}}^*\theta_j)_{(g,X)}(0,\delta X) = \left.d\theta_j\right\vert_X(\delta X).
\end{align*} Since $\theta_1,\ldots,\theta_{\rho_A}$ are algebraically independent polynomials on $\mathfrak{m}$, the Jacobian criterion  implies that, on a Zariski open subset of $U$, \begin{align}
\mathrm{rank}\bigl(d(\pi_{\mathfrak{m}}^*\theta_1),\ldots,d(\pi_{\mathfrak{m}}^*\theta_{\rho_A})\bigr)=\rho_A . \label{eq:rankid}
\end{align} Combining the orbit and fiber parts yields \begin{align*}
\mathrm{rank}\,d\Phi = (n-r)+\rho_A \quad\text{on a Zariski open subset of }U.
\end{align*} The only tautological relations among the components of $\Phi$ are the $s$ identities \begin{align*}
P^*C_{i_\ell} - \pi_{\mathfrak{m}}^*\theta_\ell = 0,\qquad \ell =1, \ldots,s .
\end{align*} Therefore, $\dim\ker(d\Phi)=s$.

(3) \textbf{Inclusion $\mathcal{B}\subset R_0$ and $\mathrm{rank} \, \mathcal{B}=s$.}  
From step (2), on the Zariski open regular set $U$, computing the rank of the Jacobian matrix gives $\mathrm{rank} \, d\Phi=(n-r)+\rho_{A}$, and the only relations among the components are $P^*C_{i_\ell}-\pi_\mathfrak{m}^*\theta_{\ell}=0$ for all $\ell=1,\dots,s$. In particular, from \eqref{eq:vectorf1} and \eqref{eq:vectorf2}, the Hamiltonian vector fields generated by the two polynomial blocks \begin{align*}
    \mathcal{D}:=\mathrm{span}\,\underbrace{\{X_{h_j \circ P}\}}_{\text{$n-r$ elements}} + \, \mathrm{span}\,\underbrace{\{X_{\pi^*_\mathfrak{m} \theta_k}\}}_{\text{$\rho_A$ elements}} 
\end{align*}
are tangent to the fibers of $q:=\pi_2\circ\pi_1:T^*M\to\mathrm{Spec} \, R_0$. Since $\mathrm{rank}\,\mathcal{D}=(n-r)+\rho_A$ on $U$ and $\dim\ker(dq)=\dim(T^*M)-\dim\mathrm{Spec}\,R_0=2\dim\mathfrak{m}-s=(n-r)+\rho_A$ on the same locus, we have \begin{align*}
    \mathcal{D}=\ker(dq)\qquad\text{on }U.
\end{align*}

Now, if $f\in \mathcal{B}$, then $\{f,g\}_\varepsilon=0$ for each component $g$ of $\Phi$, so $df(X_g)=\omega_\varepsilon(X_f,X_g)=0$ for all $X_g\in\mathcal{D}=\ker(dq)$. Hence $f$ is constant along the $q$-fibers on $U$. Thus, $f\vert_U$ is constant along the $\mathcal{D}$-leaves and therefore factors through $\mathrm{Spec} \, R_0$. Equivalently, there exists $\widetilde{F}\in R_0$ such that \begin{align*}
    f\vert_U=\widetilde{F}\big(J_1,\dots,J_s\big).
\end{align*} Since $U$ is Zariski open dense and all functions involved are regular, the identity extends to all  $T^*M$. Hence, there do not exist any other algebraic relations among the components of $\Phi$ beyond $P^*C_{i_\ell} = \pi_\mathfrak{m}^* \theta_\ell$ and $f\in R_0$. This proves $\mathcal{B} \subset R_0$.  

Since $R_0\subset \mathcal{B}$ is from Proposition \ref{prop:inter}, we obtain $\mathcal{B}=R_0$ in $U$. The algebraic independence of $\{J_1,\ldots,J_s\}$ on $U$ implies $\mathrm{rank} \, \mathcal{B} =\mathrm{trdeg} \,\widetilde{F}(J_1,\ldots,J_s)=s$.

(4) \textbf{Dimension.} Let $\rho_{A}:=\mathrm{trdeg} \, S(\mathfrak{m})^A$. The Jacobian calculation shows that, modulo the $s$ relations $u_{\ell}=0$, the differentials of the generators in $\Phi$ are independent of $U$. Hence,  $ \dim\mathrm{Spec} \,\mathcal{A}=\mathrm{trdeg}\mathcal{A}=(n-r)+\rho_{A}$.  For generic $\xi\in(\mathfrak{m}-\varepsilon W)\cap\mathfrak{g}_{\mathrm{reg}}$, $\mathfrak{g}_{\xi}$ is of dimension $r$ as $\xi$ is regular, and \begin{align*}
    \mathfrak{g}_{\xi}=(\mathfrak{a}\cap\mathfrak{g}_{\xi})\oplus(\mathfrak{g}_{\xi})_{\mathfrak{m}},\qquad \dim(\mathfrak{g}_{\xi})_{\mathfrak{m}}=s \text{ by Proposition } \ref{pro:dimension},
\end{align*} so $\dim(\mathfrak{a}\cap\mathfrak{g}_{\xi})=r-s$. Hence $\dim(A \cdot \xi)=\dim\mathfrak{a}-(r-s)$, and by the rank formula \begin{align*}
    \rho_{A}=\dim\mathfrak{m}-\dim(A \cdot \xi) =(n-\dim A)-\bigl(\dim A-(r-s)\bigr) =n-2\dim A+r-s.
\end{align*} Therefore, $$\dim\mathrm{Spec} \,\mathcal{A}=(n-r)+\rho_{A}=2(n-\dim A)-s=2\dim\mathfrak{m}-s.$$ The two stated specializations follow by substituting $\dim A=r$, $s=r$ when $A=T$ and $W$ are regular, and by allowing $s\leq r$ when $A\neq T$ and $W$ are irregular while $\xi$ remains regular.  
 \end{proof}


As a consequence of Theorem \eqref{thm:dimea}, we will directly show that $R_0$ is indeed a Poisson center in the following corollary through direct computation. 
 
\begin{corollary}
\label{coro:check}
    For every $r\in R_0$, we have \begin{align*}
\{r,f\}_1=0\quad \text{ for all } f\in \mathfrak{F}_1^{\mathrm{poly}},\qquad \{r,\phi\}_2=0\quad \text{ for all } \phi\in \mathfrak{F}_2^{\mathrm{poly}}.
\end{align*} Equivalently, $\iota_i(R_0)\subset \mathcal{Z}(\mathfrak{F}_i^{\mathrm{poly}})$ for $i=1,2$.
\end{corollary} 
\begin{proof}
Fix $\ell \in R_0$. By the definition of the intersection, there exist $h\in S(\mathfrak{g})$ and $\theta\in S(\mathfrak{m} - \varepsilon W)^A$ such that, as functions on $T^*M$,\begin{align*}
\ell=h\circ P=\theta\bigl(X-\varepsilon W\bigr).
\end{align*} Let $f = h'\circ P\in \mathfrak{F}_1^{\mathrm{poly}}$ for some $h'\in S(\mathfrak{g})$. Using the expression of $\ell$ coming from the second block and Lemma \ref{lem:interzero}, \begin{align*}
\{\ell,f\}_1=\bigl\{\theta(X-\varepsilon W),\,h'\circ P\bigr\}_\varepsilon=0.
\end{align*}
Similarly, let $\phi=\theta'(X-\varepsilon W)\in \mathfrak{F}_2^{\mathrm{poly}}$ for some $\theta'\in S(\mathfrak{m})^A$. The same reason shows that \begin{align*}
\{\ell,\phi\}_2=\bigl\{h\circ P,\,\theta'(X-\varepsilon W)\bigr\}_\varepsilon=0.
\end{align*} Thus, $\ell$ Poisson commutes with each factor.
\end{proof}

Since $(\mathcal{A},\{\cdot,\cdot\}_\mathcal{A})$ is a finitely generated Poisson algebra, $\mathrm{Spec}\,\mathcal{A}$ is a Poisson variety. Then the evaluation map $\pi_1: (T^*M,\omega_\varepsilon) \rightarrow \left(\mathrm{Spec} \, \mathcal{A},\{\cdot,\cdot\}_{\mathcal{A}}\right)$ is Poisson. Moreover, another canonical projection $\pi_2: \left(\mathrm{Spec} \, \mathcal{A},\{\cdot,\cdot\}_{\mathcal{A}}\right) \rightarrow \left(\mathrm{Spec}\, R_0,\{\cdot,\cdot\}_0\right)$ is also Poisson, as $R_0$ is a Poisson center. Together with everything above, by Definition \ref{def:superge}, we deduce the following superintegrable systems. 

\begin{theorem}
\label{thm:superin}
Let $G$ be a compact, simply-connected, semisimple Lie group, and let $A=\exp(\ker\mathrm{ad}(W))\supseteq T$. Suppose that $(T^*M,\omega_\varepsilon)$ is the twisted symplectic manifold. On the Zariski open dense regular locus $U \subset T^*M$, we then have the following arguments:

(i) if $W$ is irregular and $T \subsetneq A$. Then the system of the Poisson chain \begin{align}
    T^*M\xrightarrow{\pi_1}\mathrm{Ad}(G)(\mathfrak{m} - \varepsilon W)\times_{R_0}(\mathfrak{m} - \varepsilon W)//A \xrightarrow{\pi_2}\mathrm{Spec} R_0  \label{eq:irr}
\end{align} is superintegrable.

(ii) If $W$ is regular, then $R_0  = P^*\left(S(\mathfrak{g})^G\right)$ and \begin{align}
    T^*M\xrightarrow{\pi_1}  \mathfrak{g}  \times_{\mathfrak{g}//G} (\mathfrak{m} - \varepsilon W)//T \xrightarrow{\pi_2}\mathfrak{g}//G  \label{eq:re}
\end{align} is superintegrable.
\end{theorem}

\begin{remark} 
Theorem \ref{thm:superin} can be viewed as a Poisson projection chain approach for the two natural families of integrals arising from a Hamiltonian $G$-action: the pullback of functions on $\mathfrak{g}^*$ along the magnetic moment map and the $G$-invariant functions. In our homogeneous magnetic setting, the second family admits a canonical algebraic model via $S(\mathfrak{m}-\varepsilon W)^A$, and the common central part $R_0$ controls the reduction to the second stage of the Poisson chain.
\end{remark}

 \section{Superintegrable systems on the reductive homogeneous spaces for \texorpdfstring{$G = \mathrm{SU}(3)$}{G = SU(3)}}
 \label{sec:examples}

With the Poisson projection chain construction discussed in Section \ref{sec:superconstruction}, we will implement and verify these results in Section \ref{sec:examples} on the cotangent bundle of the reductive homogeneous space $\mathrm{SU}(3)/A$. We will work out explicit generators, the twisted Poisson bracket, and the rank account. Throughout Section \ref{sec:examples}, let $G = \mathrm{SU}(3)$. Recall that, from the classification of adjoint orbits of compact Lie groups \cite{MR1920389} and $M = G/A$, they are classified into regular and irregular cases based on the position of the stabiliser in the root diagram. 
Consider the complexification of $\mathfrak{su}(3)$. The simple roots are given by $\alpha_1 = \varepsilon_1 - \varepsilon_2$ and $\alpha_2 = \varepsilon_2 - \varepsilon_3$.  Writing a diagonal element as $H = \mathrm{diag}(c_1,c_2,c_3)$ with $c_1+c_2+c_3 = 0$, we have
\begin{align*}
\alpha_1(H) = c_1 - c_2,\qquad \alpha_2(H) = c_2- c_3,\qquad (\alpha_1+\alpha_2)(H)= c_1- c_3.
\end{align*} By Proposition \ref{prop:characterizationirrarr}, $H$ is regular if all three diagonal entries are pairwise distinct.

Section \ref{sec:examples} is divided into two parts. In Subsection \ref{subsec:regular}, we begin by choosing a fixed $W \in \mathfrak{t}$ as a regular value. That is, the reductive homogeneous space is the full flag manifold $\mathrm{SU}(3)/T$, and we use Chevalley coordinates to describe the $T$-invariant polynomials and their Poisson algebra. In this way, superintegrability and its action-angle coordinates are established. Then, in Subsection \ref{subsec:irregular}, with an irregular $W$, we will follow the same route to demonstrate superintegrability through the Poisson projection chains. In this case, rather than using the Chevalley basis, we implement the Gell-Mann basis to proceed with the computation. Both examples illustrate how the general moment map construction produces explicit first integrals and magnetic geodesic flows on homogeneous symplectic manifolds. 

\subsection{Superintegrable systems on \texorpdfstring{$\mathrm{SU}(3)/T$}{SU(3)/T}}
\label{subsec:regular}

In Subsection \ref{subsec:regular}, we consider the full flag manifold of $\mathbb{C}^3$. That is, $  \mathrm{SU}(3)/T  $ with the maximal torus $T  \subset \mathrm{SU}(3)$. The Weyl group is $\mathcal{W} \cong S_3$, and $\mathfrak{su}(3)$ admits a Cartan decomposition\footnote{It is also a reductive decomposition, so the construction in Section \ref{sec:superconstruction} applies here.} $\mathfrak{t} \oplus \mathfrak{m}$, where $\mathfrak{t} = \mathrm{Lie}(T)$ and \begin{align*}
\mathfrak{m} = \bigoplus_{k =1}^3\bigl(\mathbb{R}\,E_{\alpha_k}\oplus\mathbb{R}\,E_{-\alpha_k}\bigr),\qquad \alpha_3 =\alpha_1 + \alpha_2.
\end{align*}  To compute explicit brackets, choose a Chevalley basis ${H_i, E_{\alpha}, E_{-\alpha}}$ for $\mathfrak{su} (3)$ with $i=1,2$ as the simple roots $\alpha_i$, and $\alpha_3$ as the highest root. In the dual space $\mathfrak{su}^*(3)$, we have coordinate functions $p_i$ and $x_{\alpha}, x_{-\alpha}$ corresponding to $H_i$ and $E_{\pm \alpha_i}$ that generate the symmetric algebra $S(\mathfrak{su}(3))$.  The Lie-Poisson brackets are determined by: For simple roots $i,j=1,2$, \begin{align*}
    \{h_i,h_j\}=0 , \text{ }  \{h_i,x_{\alpha_j}\} = a_{ij}x_{\alpha_j}  \text{ and }  \{h_i,x_{-\alpha_j}\} = -a_{ij}x_{-\alpha_j}, \text{ } \{x_{\alpha_i},x_{-\alpha_j}\} = \delta_{ij}h_i.
\end{align*}   Here, $(a_{ij})$ is the Cartan matrix with $a_{ii}=2$, $a_{12}=a_{21}=-1$ for $A_2$. Let \begin{align*}
    H_1 =\mathrm{i} \begin{pmatrix} 1& & \\
    &-1& \\ 
    & & 0\end{pmatrix},\qquad  H_2= \frac{\mathrm{i}}{\sqrt3} \begin{pmatrix} 1& & \\
&1& \\
& &-2
\end{pmatrix}
\end{align*} span the Cartan subalgebra $\mathfrak{t} \subset \mathfrak{g}$, where $\mathrm{i} = \sqrt{-1}$. We also denote the root vector by $E_{\pm\alpha_1} = E_{12},E_{21}$, $E_{\pm\alpha_2}= E_{23},E_{32}$, and $E_{\pm\alpha_3} = E_{13},E_{31}$. Here $\alpha_1,\alpha_2,\alpha_3 \in \Phi^+$. Then for any $X \in \mathfrak{g}$, we have \begin{align*}
X= h_1 H_1+ h_2H_2 + \sum_{k = 1}^3\bigl(x_kE_{\alpha_k} + y_kE_{-\alpha_k}\bigr) \in\mathfrak{g}, \qquad  z_k = x_k + \mathrm{i}\,y_k  , \text{ } k=1,2,3,
\end{align*} where $(h_1, h_2,x_1,x_2, x_3, y_1, y_2, y_3)$ are the real coordinates, and $z_k$ is the complex coordinate corresponding to the root space $E_{\pm \alpha_k}$.  We now fix a regular element $W=\frac{1}{2}(H_1 + H_2) \in \mathfrak{t}$. Fix the $G$-invariant inner product $B$ in $\mathfrak{g}$ such that  $B(E_{\alpha},E_{-\alpha}) =1$, $B(\mathfrak{t},\mathfrak{m})=0$, and $B$ is $\mathrm{Ad}$-invariant. Throughout Subsection \ref{subsec:regular}, we use $B(X,Y)=-\frac{1}{2}\mathrm{tr}(XY)$. 
 From the discussion in Section \ref{sec:superconstruction},  $ T^*(\mathrm{SU}(3)/T)\cong \mathrm{SU}(3)\times_T\mathfrak{m}^*$  with the magnetic form $\omega_\varepsilon=\omega_{\mathrm{can}}+\varepsilon\,\pi^*\omega_{\mathrm{KKS}}$, where the base form $\omega_{\mathrm{KKS}}$ at $[e] := eT\in M=\mathrm{SU}(3)/T $ is $\omega_{\mathrm{KKS}}([v_1],[v_2])=B(W,[v_1,v_2])$ for any $[v_1],[v_2] \in T_{[e]}M$.
The magnetic moment map is as follows:  \begin{align*}
P(g,X)=\mathrm{Ad}(g)\bigl(X-\varepsilon W\bigr)\in\mathfrak{g}^*,\qquad P_j:=\langle P,\xi_j\rangle
\end{align*} for a fixed $B$-orthonormal basis $\{\xi_j\}_{j=1}^8$ of $\mathfrak{g}$, where $\xi_j$ is the relabelling of the real coordinates of $\mathfrak{su}(3)$ given above. 

\subsubsection{Constructing superintegrable systems on $T^*(\mathrm{SU}(3)/T)$}

Let us first consider the linear coordinate of $\mathfrak{su}^*(3)$ in complex coordinates.
 
\begin{lemma} 
\label{lem:root-coord}
Let $\mathfrak{g}_\mathbb{C}=\mathfrak{sl}_3(\mathbb{C})$ with the Cartan subalgebra $\mathfrak{t}_\mathbb{C}$ and its root space decomposition $\mathfrak{g}_\mathbb{C} = \mathfrak{t}_\mathbb{C}  \oplus \bigoplus_{\alpha\in\Delta}\mathfrak{g}_\alpha$ with $  \dim \mathfrak{g}_{\pm\alpha}=1$. Choose root vectors $E_{\pm\alpha_k}\in\mathfrak{g}_{\pm\alpha_k}$ for the three positive roots $\alpha_1,\alpha_2,\alpha_3=\alpha_1+\alpha_2$, and the Killing form $B$ such that
\begin{align}
B(E_{\alpha_k},E_{-\alpha_\ell})=\delta_{k\ell},\qquad B(\mathfrak{t}_\mathbb{C},\mathfrak{g}_{\pm\alpha})=0,\qquad B(\mathfrak{g}_\alpha,\mathfrak{g}_\beta)=0 \text{ if }\alpha+\beta\neq 0. \label{eq:normal} 
\end{align} Take $E_{\pm \alpha_k}$ in $\mathfrak{sl}(3,\mathbb{C})$ such that the normalization in \eqref{eq:normal} holds. Let $\mathfrak{g}=\mathfrak{su}(3)$ be the compact real form with the conjugation $E_{-\alpha_k}=\overline{E_{\alpha_k}}$. Then every $\zeta\in\mathfrak{g}$ admits a unique expansion \begin{align*}
\zeta = H +  \sum_{k=1}^3\Big(z_k\,E_{\alpha_k}+\overline{z_k}\,E_{-\alpha_k}\Big), \qquad H\in\mathfrak{t},\ z_k\in\mathbb{C},
\end{align*} and the coefficients are: \begin{align}
  z_k:=  z_k(\zeta) = B(\zeta,E_{-\alpha_k}),\qquad \overline{z_k} : = \overline{z_k}(\zeta)= B(\zeta,E_{\alpha_k}) \quad  k=1,2,3. \label{eq:complexcoe}
\end{align}
\end{lemma}

\begin{proof}
By root decomposition and orthogonality, $B(\mathfrak{t}_\mathbb{C},\mathfrak{g}_{\pm\alpha})=0$ and $B(\mathfrak{g}_{\alpha},\mathfrak{g}_{\beta})=0$ hold unless $\beta=-\alpha$. Using the normalization $B(E_{\alpha_k},E_{-\alpha_\ell})=\delta_{k\ell}$, we compute \begin{align*}
B(\zeta,E_{-\alpha_j}) = B\Big(H+\sum_{k=1}^3 \big(z_k E_{\alpha_k}+\overline{z_k} E_{-\alpha_k}\big),\,E_{-\alpha_j}\Big) = \sum_{k=1}^3 z_k\,B(E_{\alpha_k},E_{-\alpha_j}) = z_j.
\end{align*} The formula for $\overline{z_j}$ is analogous.
\end{proof}

\begin{remark}
    Let $\sigma \in \mathrm{Aut}(\mathfrak{g}_\mathbb{C})$ with $\sigma(E_\alpha) = -E_{-\alpha}$ be a $B$-invariant conjugation such that $\sigma^2 = \mathrm{id}$. Then \begin{align*}
        \overline{B(\zeta,E_\alpha)} = B(\sigma(\zeta),\sigma(E_\alpha)) = B(\zeta, -E_{-\alpha}) = -B(\zeta,E_{-\alpha}).
    \end{align*} Hence $B(\zeta,E_{-\alpha}) = -  \overline{B(\zeta,E_\alpha)}$.
\end{remark}

We first provide the generators in the algebras $\mathfrak{F}_1^{\mathrm{poly}}$ and $\mathfrak{F}_2^{\mathrm{poly}}$. 
By Lemma \ref{lem:rankacc}, we have $\mathrm{rank} \, \mathfrak{F}_1^{\mathrm{poly}}= \dim \mathfrak{su}(3) = 8$. Hence, we can select eight basis elements, resulting in eight coordinate functions denoted by $ P_k=P_{\xi_k}$, where $ k=1,\ldots,8  $. In particular, the linearity of the generators ensures algebraic independence. Consequently, we define the Poisson polynomial algebra as $$ \mathfrak{F}_1^{\mathrm{poly}} =\mathbb{R}[P_1,\ldots,P_8], \qquad\text{with}\ \deg P_k=1. $$  Summarized in the following lemma.
\begin{lemma}
\label{lem:firf1}
    We have $\mathfrak{F}_1^{\mathrm{poly}} = P^*(S(\mathfrak{su}(3)) = \mathbb{R}[P_1,\ldots,P_8]$ with the commutator relations $\{P_i,P_j\}_1 = c_{ij}^kP_k$.
\end{lemma}

Let us now focus on the structure of the algebra $\mathfrak{F}_2^{\mathrm{poly}}$, specifically considering the root coordinates $z_k$. By Lemma \ref{lem:root-coord}, the complex root coordinate functions can be expressed in the following way: \begin{align}
    z_k(g,X)= \langle P(g,X),E_{\alpha_k}\rangle    =\exp(-\mathrm{i}\langle\alpha_k,\phi(g)\rangle)\,(x_k+iy_k),\quad k=1,2,3, \label{eq:coordinate}
\end{align}  where $\phi(g)\in\mathfrak{t}$ satisfies $g=\exp(\phi(g)) \in T $, and $\mathrm{i} = \sqrt{-1}$ is a complex number. We refer the reader to \cite{campoamor2023algebraic} the Cartan centralizer for the complexified case $\mathfrak{sl}(n,\mathbb{C})$, where the Racah algebra is embedded in the Cartan centralizer of type $A_n$ \cite{MR3884869}. For the general classification of the explicit expression of the Cartan commutant for $B_n,C_n$ and $D_n$, see \cite{campoamor2024superintegrable,abc}.  
 \begin{lemma}
 \label{alg1}
The Poisson algebra $S(\mathfrak{m})^T$ of the $T$-invariant polynomials with a Poisson structure $\{\cdot,\cdot\}_2$ induced by $\{\cdot,\cdot\}_\varepsilon$ is minimally generated by the five polynomials as follows:  \begin{align*}
    u_k=|z_k|^2\,, \text{ } k=1,2,3 ,\qquad v=\Re(z_1z_2\bar z_3),\qquad w=\mathrm{Im}(z_1z_2\bar z_3),
\end{align*} subject to the single cubic relation $ u_1u_2u_3=v^2+w^2 $.  Then   \begin{align*}
    \mathfrak{F}_2^{\mathrm{poly}} =  \pi_\mathfrak{m}^* \left(\mathbb{R}[u_1,u_2,u_3,v,w]\Big/ \langle u_1u_2u_3-v^2-w^2\rangle\right), \qquad \deg u_k=2,\,\deg v=\deg w=3.
\end{align*}  The Poisson brackets among these generators are  \begin{align}
 \nonumber
     &\{u_1,u_2\}_2=2v,\ \{u_2,u_3\}=2v,\ \ \ \{u_3,u_1\}_2=2v,\\
     \nonumber
&\{u_1,v\}_2=u_1(u_3-u_2)-c_1w,\quad    \{u_2,v\}_2=u_2(u_1-u_3)-c_2w,\\
&\{u_3,v\}_2=u_3(u_2-u_1)-c_3w,\quad   \{u_k,w\}_2=c_k v,\\
  \nonumber
&\{v,w\}_2=-\frac{1}{2}\big(c_3u_1u_2-c_2u_1u_3-c_1u_2u_3\big),
 \end{align}  where $ 1\leq  k \leq 3$ and $c_k = \varepsilon B(W,H_{\alpha_k})$. 
\end{lemma}

\begin{proof}
We first compute the $T$-invariant generators in $S(\mathfrak{m})$. By Lemma \ref{lem:root-coord}, we deduce  \begin{align}
    u_k:= |z_k|^2=x_k^2+y_k^2\quad \text{ with } k=1,2,3,\qquad v+\mathrm{i}w:=z_1\,z_2\,\overline{z_3} . \label{eq:tinvari}
\end{align} We first show that the polynomials defined in \eqref{eq:tinvari} are $T$-invariant. For the degree two generators, it is obvious for $u_k$ as $\alpha_k +  (-\alpha_k) = 0$. Moreover, for $z_1 z_2 \overline{z_3}$, we use the root identity $\alpha_1+\alpha_2-\alpha_3=0$ since the total weight is zero. Hence, it is $T$-invariant. We then take its real and imaginary parts $v=\Re(\cdot)$, $w=\mathrm{Im}(\cdot)$. A general $T$-invariant monomial is $z_1^{a_1}z_2^{a_2}z_3^{a_3}
\overline{z_1}^{b_1}\overline{z_2}^{b_2}\overline{z_3}^{b_3}$ with weight $(a_1-b_1)\alpha_1+(a_2-b_2)\alpha_2+(a_3-b_3)\alpha_3$. That is, by the definition of the Cartan commutant \begin{align*}
  (a_1-b_1)\alpha_1+(a_2-b_2)\alpha_2+(a_3-b_3)(\alpha_1+\alpha_2)=0,  
\end{align*} which is equivalent to the linear system $ (a_1-b_1)+(a_3-b_3)=0$ and  $ (a_2-b_2)+(a_3-b_3)=0$. Hence, $a_1-b_1=-(a_3-b_3)$ and $a_2-b_2=-(a_3-b_3)$. Writing $d:=a_3-b_3$, we can decompose any invariant monomial as a product of the powers of $|z_k|^2=u_k$ (which accounts for the symmetric part) and $(z_1z_2\overline{z_3})^d$ (which accounts for the weight $d$), together with its complex conjugate, if needed, to keep the coefficients real. Thus, the invariant ring is generated by $u_1,u_2,u_3$ and $z_1z_2\overline{z_3}$ (hence, by $u_1,u_2,u_3,v,w$).  In addition, a routine computation shows that there is an algebraic relation between these five invariants: \begin{align}
 u_1\,u_2\,u_3=|z_1|^2\,|z_2|^2\,|z_3|^2 =|z_1 z_2 \overline{z_3}|^2=v^2+w^2 . \label{eq:plrelation}
\end{align}   

   We then calculate the Poisson bracket relations from all the $5$ generators $\{u_1,u_2,u_3,w,v\}$. By Proposition \ref{prop:slicepoisson}, the Poisson bracket defined on $S(\mathfrak{m})^T$ is given by $\{\theta_1,\theta_2\}_2(\xi) :=-B \left(\xi,\bigl[(\nabla\theta_1(\xi))_{\mathfrak{m}},(\nabla\theta_2(\xi))_{\mathfrak{m}}\bigr]\right)$, where $(\cdot)_\mathfrak{m}$ is projection to $\mathfrak{m}$-component. In other words, we need to evaluate the gradient of the coordinate functions. By Lemma \ref{lem:root-coord}, the linear functional of coordinates is given by $z_k(\zeta)=B(\zeta,E_{-\alpha_k})$, and $\overline{z_k}(\zeta)=B(\zeta,E_{\alpha_k})$ are linear functionals. A direct calculation shows that  $\nabla z_k=E_{-\alpha_k}$ and $\nabla \overline{z_k}=E_{\alpha_k}$. Therefore, \begin{align}
\nabla u_k=\nabla(z_k\overline{z_k})=\overline{z_k}\,E_{\alpha_k}+z_k\,E_{-\alpha_k}.\tag{A}
\end{align} Writing $F:=z_1z_2\overline{z_3}$ such that $v=\frac{1}{2}(F+\overline{F})$, $w=\frac{1}{2\mathrm{i}}(F-\overline{F})$, \begin{align}
\nabla F=z_2\overline{z_3}\,E_{-\alpha_1}+z_1\overline{z_3}\,E_{-\alpha_2}+z_1z_2\,E_{\alpha_3},\qquad \nabla\overline{F}=\overline{z_2}z_3\,E_{\alpha_1}+\overline{z_1}z_3\,E_{\alpha_2}+\overline{z_1}\,\overline{z_2}\,E_{-\alpha_3}.\tag{B}
\end{align} From (A) and (B), we first deduce $(\theta)_\mathfrak{m} =\theta$ for any $\theta \in S(\mathfrak{m})^T$.

We now compute the Poisson bracket $\{u_1,u_2\}_2$.  Using (A) and (B), \begin{align*}
[\nabla u_1,\nabla u_2] = z_1z_2\,[E_{-\alpha_1},E_{-\alpha_2}]+\overline{z_1}\,\overline{z_2}\,[E_{\alpha_1},E_{\alpha_2}] =  \overline{z_1}\,\overline{z_2}\,E_{\alpha_3}-\,z_1z_2\,E_{-\alpha_3}.
\end{align*} Therefore, by Lemma \ref{lem:root-coord} and its remark, we have $B(\zeta,E_{\alpha_3}) = -z_3$ and $B(\zeta,E_{-\alpha_3}) = \overline{z_3}$, and hence \begin{align*}
\{u_1,u_2\}_2 =-B\big(\zeta,\,[\nabla u_1,\nabla u_2]\big) = -\Big( \overline{z_1}\,\overline{z_2}\,B(\zeta,E_{\alpha_3}) -z_1z_2\,B(\zeta,E_{-\alpha_3})\Big) = \overline{z_1}\,\overline{z_2}\,z_3 +z_1z_2\,\overline{z_3} = 2\,\Re(z_1z_2\overline{z_3}) = 2\,v.
\end{align*}

 Now, we focus on the Poisson bracket $\{u_1,F\}_2$. Similar to what we did above, a direct computation shows that \begin{align*}
[\nabla u_1,\nabla F] &= \left[\,\overline{z_1}E_{\alpha_1}+z_1E_{-\alpha_1},\ z_2\overline{z_3}E_{-\alpha_1}+z_1\overline{z_3}E_{-\alpha_2}+z_1z_2E_{\alpha_3}\right]\\
&= \overline{z_1}\,z_1z_2\,[E_{\alpha_1},E_{\alpha_3}]+\overline{z_1}\,z_1\overline{z_3}\,[E_{\alpha_1},E_{-\alpha_2}] + \overline{z_1}\,z_2\overline{z_3}\,[E_{\alpha_1},E_{-\alpha_1}]\\
&\quad +  z_1\,z_1\overline{z_3}\,[E_{-\alpha_1},E_{-\alpha_2}] +  z_1\,z_1z_2\,[E_{-\alpha_1},E_{\alpha_3}] +  z_1\,z_2\overline{z_3}\,[E_{-\alpha_1},E_{-\alpha_1}]\\
&=  \underbrace{z_2\overline{z_3}\,\overline{z_1}}_{\;=\ \overline{z_1}z_2\overline{z_3}}\,H_{\alpha_1} - \underbrace{z_1\overline{z_1}\,\overline{z_3}}_{\;=\ |z_1|^2\overline{z_3}}\,E_{-\alpha_3} + \underbrace{z_1\overline{z_1}\,z_2}_{\;=\ |z_1|^2 z_2}\,E_{\alpha_2}.
\end{align*}
Then the twisted Poisson bracket $\{u_1,F\}_2$ is calculated as follows: \begin{align*}
\{u_1,F\}_2 =-\,B(\zeta,[\nabla u_1,\nabla F]) = -\Big(\overline{z_1}z_2\overline{z_3}\,B(\zeta,H_{\alpha_1}) - |z_1|^2\overline{z_3}\,B(\zeta,E_{-\alpha_3}) +  |z_1|^2 z_2\,B(\zeta,E_{\alpha_2})\Big).
\end{align*} Using $B(\zeta,H_{\alpha_1})=B(X,H_{\alpha_1})-\varepsilon\,c_1$ and $B(\zeta,E_{-\alpha_3})=z_3,\ B(\zeta,E_{\alpha_2})=\overline{z_2}$, we obtain \begin{align}
\{u_1,F\}_2 = -\,\overline{z_1}z_2\overline{z_3}\, \left(B(X,H_{\alpha_1})-\varepsilon c_1 \right) -|z_1|^2\big( |z_2|^2-|z_3|^2 \big). \label{eq:a1}
\end{align}

 Finally, we compute the twisted Poisson bracket $\{u_1,\overline{F}\}_2$. A similar calculation from (A) and (B) gives \begin{align*}
[\nabla u_1,\nabla \overline{F}] = -\,z_1\overline{z_2}z_3\,H_{\alpha_1}+|z_1|^2 z_3\,E_{\alpha_3}-|z_1|^2\overline{z_2}\,E_{-\alpha_2}.
\end{align*} Hence, \begin{align}
\{u_1,\overline{F}\}_2 = z_1\overline{z_2}z_3\,\left(B(X,H_{\alpha_1})-\varepsilon c_1\right) -|z_1|^2\big(\,|z_3|^2-|z_2|^2\,\big). \label{eq:a2}
\end{align}

To obtain the corresponding Poisson bracket, we then need to transform from $F,\overline{F}$ to $v,w$.  Since $v=\frac{1}{2}(F+\overline{F})$ and $w=\frac{1}{2\mathrm{i}}(F-\overline{F})$, \begin{align*}
\{u_1,v\}_2 =\frac{1}{2}\big(\{u_1,F\}_2+\{u_1,\overline{F}\}_2\big),\qquad \{u_1,w\}_2 =\frac{1}{2\mathrm{i}}\big(\{u_1,F\}_2-\{u_1,\overline{F}\}_2\big).
\end{align*} Adding \eqref{eq:a1} and \eqref{eq:a2} together cancels the terms $|z_1|^2$ and yields \begin{align*}
\{u_1,v\}_2 = -\,\Re\Big(\overline{z_1}z_2\overline{z_3}-z_1\overline{z_2}z_3\Big)\left(B(X,H_{\alpha_1})-\varepsilon c_1\right) = -\,\varepsilon c_1\,\Im(z_1z_2\overline{z_3}) +|z_1|^2\big(|z_3|^2-|z_2|^2\big).
\end{align*} Recognizing $w=\Im(z_1z_2\overline{z_3})$ and $u_k=|z_k|^2$, we have \begin{align}
 \{u_1,v\}_2=u_1(u_3-u_2)-c_1\,w . \label{eq:b}
\end{align} Similarly, \begin{align}
\{u_1,w\}_2 =\frac{1}{2\mathrm{i}}\big(\{u_1,F\}_2-\{u_1,\overline{F}\}_2\big) =  \frac{c_1}{2}\big(\overline{z_1}z_2\overline{z_3}+z_1\overline{z_2}z_3\big) =  c_1 v . \label{eq:c}
\end{align} A cyclic permutation of $(1,2,3)$ yields $\{u_i,u\}_2$ and $\{u_i,w\}_2 $. Finally, \begin{align*}
\{v,w\}_2 =\frac{1}{4\mathrm{i}}\big(\{F+\overline{F},F-\overline{F}\}_2\big) =\frac{1}{2\mathrm{i}}\,\{F,\overline{F}\}_2- B \big(\zeta,\,[\nabla F,\nabla\overline{F}]\big).
\end{align*} Here, we deduce that $B\big(\zeta,\,[\nabla F,\nabla\overline{F}]\big) = \frac{\mathrm{i}}{2}\,\Big(c_3 u_1u_2-c_2 u_1u_3-c_1 u_2u_3\Big) $ such that \begin{align}
 \{v,w\}_2=-\frac{1}{2}\Big(c_3 u_1u_2-c_2 u_1u_3-c_1 u_2u_3\Big).
\end{align} This completes the proof.
\end{proof}
\begin{remark}  \label{rem:magneticcoup}
(i) Throughout Lemma \ref{alg1}, we use normalization  \begin{align}
    c_k :=\varepsilon\,B \left(W,H_{\alpha_k}\right)\qquad k=1,2,3,
\end{align} where $H_{\alpha_k}$ are the coroots (with $\alpha_3=\alpha_1+\alpha_2$) and $B$ is the invariant bilinear form. With this choice, the Poisson twisted brackets appearing in \eqref{eq:a1}, \eqref{eq:a2}-\eqref{eq:c} are written with the symbols $c_k$ and without an extra factor of $\varepsilon$. 
If one prefers to keep the magnetic term explicit, set instead \begin{align*}
    \widehat{c}_k :=B \left(W,H_{\alpha_k}\right),
\end{align*} in which case the same relations become \begin{align*}
\{u_i,v\}_2 &=u_i(u_k-u_j)-\varepsilon\,\widehat{c}_i\,w,\\
\{u_i,w\}_2 &=\varepsilon\,\widehat{c}_i\,v,\\
\{v,w\}_2 &=-\frac{1}{2} \left(\varepsilon\,\widehat{c}_3 u_1u_2-\varepsilon\,\widehat{c}_2 u_1u_3-\varepsilon\,\widehat{c}_1 u_2u_3\right),
\end{align*} and similarly $\{u_i,u_j\}_2=2v$. Thus, Lemma \ref{alg1} is consistent with either convention. Hence, we adopt $c_k=\varepsilon B(W,H_{\alpha_k})$ precisely to avoid an additional factor $\varepsilon$. 
 
  (ii)   From the above discussion, we also deduce that the generators in the coordinate ring $S(\mathfrak{m}-\varepsilon W)^T$ describe the coordinates on $ \mathfrak{m}//T= \mathrm{Spec}\,S(\mathfrak{m})^T$. Note that the affine quotient $\mathfrak{m}//T$ can be identified with the $5$-dimensional real hypersurface \begin{align}
    \mathfrak{m}//T:= \{(u_1,u_2,u_3,v,w)\in\mathbb{R}^{5}: u_1u_2u_3 = v^2 + w^2,  u_i \geq 0 \}. \label{eq:afb}
\end{align} It contains a singularity in the locus $\{v = w  = 0\}$ and $\{u_iu_ju_k = 0\}$.   The map $\kappa:\mathfrak{m}\to X$ sending $$X\mapsto(u_1(X),u_2(X),u_3(X),v(X),w(X))$$ is the canonical quotient map on invariants. Note that for any $X_1,X_2 \in \mathfrak{m}$, $\pi(X_1) = \pi(X_2)$ implies that there exists a $t \in T$ such that $X_2 = \mathrm{Ad}(t^{-1})X_1$. The $T$-invariant map $\kappa\circ\widetilde{\pi}_{\mathfrak{m}}:G\times\mathfrak{m}\to\mathfrak{m}//T$ descends to a well-defined quotient-valued map $T^*(\mathrm{SU}(3)/T)\to\mathfrak{m}//T$, and its coordinate functions are exactly the generators in $S(\mathfrak{m})^T$. Geometrically, we could define $r_k : = |Z_k| = \sqrt{u_k}$ and $\phi : = \mathrm{arg} \left(\frac{u_1 u_2}{u_3}\right)$. Then \begin{align*}
    v = r_1 r_2 r_3 \cos \phi \text{ and } w   = r_1 r_2 r_3 \sin \phi.
\end{align*} Hence, the local coordinate on $\mathfrak{m}//T$ is given by $(r_1,r_2,r_3,\phi)$.
\end{remark}

The intersection consists of those elements of the $P^*(\mathfrak{su}(3))$ and $T$-invariant generators after restriction to $\mathfrak{m}$. By Chevalley’s theorem, the $G$-invariant polynomials in $\mathfrak{su}(3)$ are generated by the quadratic and cubic Casimirs  \begin{align}
    C_2(Y)=\sum_{i=1}^8Y_i^2,\qquad  C_3(Y)=\sum_{i,j,k}^8d_{i,j,k}Y_iY_jY_k, \label{eq:Casimir}
\end{align}   and their restrictions $ R_2 : = \mathrm{Res}_W (C_2)  =u_1+u_2+u_3,$ and $ R_3 : = \mathrm{Res}_W (C_3)=2v $ are clearly non-zero.  Hence, \begin{align*}
R_0 = \mathfrak{F}_1^{\mathrm{poly}}\cap \mathfrak{F}_2^{\mathrm{poly}} =P^* \bigl(S(Z_P)^G\bigr) =\pi_\mathfrak{m}^* \bigl(\mathbb{R}[R_2,R_3]\bigr) =\mathbb{R}[J_2,J_3],    
\end{align*}  where $J_2:=P^*C_2=\pi_\mathfrak{m}^*R_2$ and $J_3:=P^*C_3=\pi_\mathfrak{m}^*R_3$. We now focus on the joint of two algebras. From the discussion in Section \ref{sec:superconstruction}, it is sufficient to study the algebra $\mathfrak{F}_1^{\mathrm{poly}} \otimes_{R_0} \mathfrak{F}_2^{\mathrm{poly}}$. 

\begin{proposition}
\label{prop:uniona}
  Let $\mathfrak{F}_1^{\mathrm{poly}}$ and $\mathfrak{F}_2^{\mathrm{poly}}$ be the same as defined in Lemma \ref{lem:firf1} and Lemma \ref{alg1}, respectively. Then $\mathcal{A} = \mathfrak{F}_1^{\mathrm{poly}} \otimes_{R_0} \mathfrak{F}_2^{\mathrm{poly}}$ is given by \begin{align}
    \mathcal{A} =  \frac{\mathbb{R} \bigl[P_1,\ldots,P_8,u_1,u_2,u_3,v,w\bigr]}{\left\langle u_1u_2u_3-v^2-w^2,C_2 - \left(u_1 +  u_2 +  u_3\right) , C_3 - 2 v\right\rangle}     . \label{eq:A1}
\end{align}   In particular, $(\mathcal{A},\{\cdot,\cdot\}_\mathcal{A})$ is a Poisson algebra with the following commutator relations: \begin{align}
\nonumber
\{P_k,P_j\}_{\mathcal{A}} &= c_{kj}^i\,P_i,\\
\nonumber
\{P_k,u_k\}_{\mathcal{A}} &= \{P_k,v\}_{\mathcal{A}}=\{P_k,w\}_{\mathcal{A}}=0,\\
\nonumber
\{u_i,u_j\}_{\mathcal{A}} &= 2\,v  , \label{eq:commure}\\
\{u_i,v\}_{\mathcal{A}} &= u_i\,(u_k-u_j)-c_i\,w ,\\
\nonumber
\{u_i,w\}_{\mathcal{A}} &= c_i\,v  ,\\
\nonumber
\{v,w\}_{\mathcal{A}} &= -\frac{1}{2}\bigl(c_3 u_1u_2-c_2 u_1u_3-c_1 u_2u_3\bigr),
\end{align} where $ 1\leq i,j,k \leq 3$ and $c_i = \varepsilon B(W,H_{\alpha_i})$.
\end{proposition}
\begin{proof}
  Due to the property of left $G$-equivariance, it suffices to compute at $g=e$. Write $\xi=X-\varepsilon W$, and let $P_Y:=\langle P,Y \rangle$ be a coordinate functional such that $P_i=P_{\xi_i}$.  By Lemma \ref{lem:F1polypbrack}, we deduce that \begin{equation}
X_{P_Y}(e,X)= \left. \dfrac{d}{dt}\right\vert_{t =0} \left(e,X + t (Y)_\mathfrak{m} + tW -\frac{t}{2}[(Y)_\mathfrak{m},X]\right)  =\Big(-(Y)_{\mathfrak{m}}\,,\,\frac{1}{2}[(Y)_{\mathfrak{m}},X]\,-\,[Y,X]_{\mathfrak{m}}\Big). \label{eq:XYfield}
\end{equation}
 
Now, to conclude \eqref{eq:commure}, it is sufficient to show that $\{\mathfrak{F}_1^{\mathrm{poly}},\mathfrak{F}_2^{\mathrm{poly}}\}_\varepsilon = 0$. For any $\theta(\xi) \in \mathfrak{F}_2^{\mathrm{poly}}$, by the definition of the twisted Poisson bracket, $\{P_Y,\theta(\xi)\}_\varepsilon=d(\theta(\xi))\big(X_{P_Y}\big)$. Together with the Hamiltonian vector field \eqref{eq:XYfield}, the differential of $\theta$ in the coordinate $\left(v, w - \frac{1}{2}[v,X]\right)$ is given by
\begin{align*}
\{P_Y,\theta(\xi)\}_\varepsilon = & \,B\left(\nabla\theta(\xi),\,\Big(\frac{1}{2}[(Y)_{\mathfrak{m}},X]-[Y,X]_{\mathfrak{m}}\Big) -\frac{1}{2}\big[-(Y)_{\mathfrak{m}},X\big]\right) \\
= & \, B\big(\nabla\theta(\xi),\, -[Y,X]_{\mathfrak{m}}\big).
\end{align*} The decomposition of $Y=(Y)_{\mathfrak{t}}+(Y)_{\mathfrak{m}}$ shows $[Y,X]_{\mathfrak{m}}=[(Y)_{\mathfrak{t}},X]$ as $[(Y)_{\mathfrak{m}},X]\in\mathfrak{t}\oplus(\text{sum of root lines})$ has no $\mathfrak{m}$-component along the $T$-invariant gradient. Hence,
\begin{equation}\label{eq:carduction}
\{P_Y,\theta(\xi)\}_\varepsilon=B\big(\nabla\theta(\xi),\, -[(Y)_{\mathfrak{t}},X]\big).
\end{equation}

Now, for any $\theta \in S(\mathfrak{m})^A$, compute \eqref{eq:carduction} using the $A_2$ root relations. First, for $u_k(\xi)=|z_k|^2$, as presented in the proof of Lemma \ref{alg1}, we have $\nabla u_k(\xi)=z_k\,E_{\alpha_k}+\overline{z_k}\,E_{-\alpha_k}$. For any $H\in\mathfrak{t}$, we have \begin{align*}
[H,X]=\sum_{i=1}^3\alpha_i(H)\,z_i\,E_{\alpha_i}-\sum_{i=1}^3\alpha_i(H)\,\overline{z_i}\,E_{-\alpha_i}.
\end{align*} Hence, using $B(E_{\alpha_i},E_{-\alpha_j})=\delta_{ij}$ and $B(E_{\pm\alpha_i},E_{\pm\alpha_j})=0$,
\begin{align*}
B\big(\nabla u_k(\xi),[H,X]\big)
= -\,\alpha_k(H)\,|z_k|^2+\alpha_k(H)\,|z_k|^2=0.
\end{align*}
Substituting $H=(Y)_{\mathfrak{t}}$ into \eqref{eq:carduction} gives $\{P_Y,u_k\}_\varepsilon=0$.

Next, let $F:=z_1z_2\overline{z_3}$ and note that $v=\mathrm{Re} F$, $w=\mathrm{Im} F$. The gradient with respect to $B$ is \begin{align*}
\nabla F(\xi)=z_2\overline{z_3}\,E_{-\alpha_1}+z_1\overline{z_3}\,E_{-\alpha_2}+z_1z_2\,E_{\alpha_3}.
\end{align*} Therefore, \begin{align*} 
B\big(\nabla F(\xi),[H,X]\big) &=\ \alpha_1(H)\,z_2\overline{z_3}\,z_1+\alpha_2(H)\,z_1\overline{z_3}\,z_2-\alpha_3(H)\,z_1z_2\overline{z_3}\\
&=\ \big(\alpha_1(H)+\alpha_2(H)-\alpha_3(H)\big)\,z_1z_2\overline{z_3}=0, 
\end{align*} since $\alpha_3=\alpha_1+\alpha_2$. Taking real and imaginary parts yields  $B\big(\nabla v(\xi),[H,X]\big)=B\big(\nabla w(\xi),[H,X]\big)=0$.  With $H=(Y)_{\mathfrak{t}}$, \eqref{eq:carduction} gives $\{P_Y,v\}_\varepsilon=\{P_Y,w\}_\varepsilon=0$. Since $P_{\xi_i}=P_i$, the three identities read \begin{align*}
\{P_j,u_k\}_\varepsilon=\{P_j,v\}_\varepsilon=\{P_j,w\}_\varepsilon=0\qquad j=1,\dots,8; \text{ }  k=1,2,3.
\end{align*} Therefore, \eqref{eq:commure} is verified.
 \end{proof}

We now construct the superintegrable system defined in $T^*(\mathrm{SU}(3)/T)$.  Define the map of affine real varieties 
\begin{align} 
\nonumber
\pi_1 : T^*(\mathrm{SU}(3)/T)&\longrightarrow\mathrm{Spec} \,\mathcal{A},\\
(g,X)\quad&\longmapsto\,
\bigl(P_1(g,X),\ldots ,P_8(g,X),u_1(g,X),u_2(g,X),u_3(g,X),v(g,X),w(g,X)\bigr).
\end{align} Here, by Corollary \eqref{cor:union}, $\mathrm{Spec} \, \mathcal{A} = \mathfrak{su}(3) \times_{\mathfrak{su}(3)//\mathrm{SU}(3)} (\mathfrak{m} - \varepsilon W)//T$. By construction, every coordinate of $\pi_1$ lies in the Poisson commutative subalgebra $\mathcal{A} \subset \mathcal{O}(T^*M)$. Hence $\pi_1$ is a Poisson morphism.  Generic fibers of $\pi_1$ have dimensions $\dim T^*M-\mathrm{trdeg} \, \mathcal{A}=12-10=2$ and are therefore isotropic.   The canonical inclusion $R_0\hookrightarrow \mathcal{A}$ induces an affine variety morphism \begin{align*}
\pi_2:\mathrm{Spec} \,\mathcal{A}\longrightarrow\mathrm{Spec} \, R_0
          \,\cong\,\mathbb{R}^2_{(J_2,J_3)}, 
  \end{align*} given by \begin{align}
     \pi_2(P_1,\ldots ,P_8,u_1,u_2,u_3,v,w)    =  \bigl(J_2,J_3\bigr)    =\left(\sum_{k=1}^8P_k^2,\,2v\right). 
  \end{align} The composite \begin{align*}
      T^*(\mathrm{SU}(3)/T)\,\xrightarrow{\,\pi_1\,}\,\mathrm{Spec} \,\mathcal{A}
        \,\xrightarrow{\,\pi_2\,}\,\mathrm{Spec} \, R_0
  \end{align*}  exhibits the following structure.
  
\begin{proposition}
\label{prop:supersu3t}
Write $\xi=X-\varepsilon W$.   Let $  (h_1,h_2,x_k,y_k)_{k=1}^3\in\mathbb{R}^8 $ be the real coordinates given before, and let $$ \pi_1\colon T^*(\mathrm{SU}(3)/T)\longrightarrow \mathrm{Spec} \,\mathcal{A},\qquad \pi_2\colon\mathrm{Spec} \,\mathcal{A}\longrightarrow\mathrm{Spec} \, R_0 $$ be a chain of Poisson submersions, where $$ \mathcal{A}=\mathfrak{F}_1^{\mathrm{poly}} \otimes_{R_0} \mathfrak{F}_2^{\mathrm{poly}},\qquad R_0=\mathfrak{F}_1^{\mathrm{poly}}\cap \mathfrak{F}_2^{\mathrm{poly}}=\mathbb{R}[J_2,J_3] $$ and $\mathfrak{F}_1^{\mathrm{poly}},\mathfrak{F}_2^{\mathrm{poly}}$ are the polynomial families.  Then \begin{align}
      T^*(\mathrm{SU}(3)/T) \xrightarrow{\,\pi_1\,} \mathrm{Spec} \,\mathcal{A}     \xrightarrow{\,\pi_2\,} \mathrm{Spec} \, R_0
  \end{align} is a superintegrable system.
\end{proposition}

\begin{proof}
Using Lemma \ref{lem:firf1} and Lemma \ref{alg1}, we see that $\mathfrak{F}_1^{\mathrm{poly}}$ and $\mathfrak{F}_2^{\mathrm{poly}}$ are indeed Poisson polynomial algebra closed under the twisted Poisson bracket $\{\cdot,\cdot\}_\varepsilon$. By Proposition \ref{prop:uniona}, $\mathcal{A}=\mathfrak{F}_1^{\mathrm{poly}}\otimes_{R_0}\mathfrak{F}_2^{\mathrm{poly}}$ is a Poisson algebra with Poisson brackets as follows:
\begin{align*}
\{f\otimes \theta,f'\otimes \theta'\}_{\mathcal{A}}=\{f,f'\}_\varepsilon\otimes \theta\theta' +  ff'\otimes\{\theta,\theta'\}_\varepsilon,\qquad f,f'\in\mathfrak{F}_1^{\mathrm{poly}},\,\theta,\theta'\in\mathfrak{F}_2^{\mathrm{poly}},
\end{align*} and extended by bilinearity. By construction, each coordinate of the map \begin{align*}
\pi_1:T^*M\to\mathrm{Spec}\,\mathcal{A},\qquad(g,X)\mapsto(P_1,\ldots,P_8,u_1,u_2,u_3,v,w),
\end{align*} belongs to $\mathcal{A}$, and the pullback of the coordinate functions on $\mathrm{Spec}\,\mathcal{A}$ reproduces the corresponding elements of $\mathcal{A}\subset \mathcal{O}(T^*M)$. The defining rule for $\{\cdot,\cdot\}_{\mathcal{A}}$ therefore gives\begin{align*}
\{f,g\}_{\mathcal{A}}\circ\pi_1=\{f\circ\pi_1,g\circ\pi_1\}_\varepsilon \quad\text{for all}\,f,g\in \mathcal{O}(\mathrm{Spec}\,\mathcal{A}) \subset C^\infty(\mathrm{Spec}\, \mathcal{A}),
\end{align*}so $\pi_1$ is a Poisson map.

We now show that $d\pi_1$ is surjective via direct computation. In other words, we show that $\mathrm{rank} \, d \pi_1 (g,X) = 10$ on a Zariski open dense subset of $T^*(\mathrm{SU}(3)/T)$. For each $k$, the differentiation of the coordinate functions is given by
\begin{align}
 dz_k(v,\delta X)= \langle [\mathrm{Ad}(g^{-1})v,\mathrm{Ad}(g^{-1})\xi]+\mathrm{Ad}(g^{-1})\delta X\,,\,E_{\alpha_k}\rangle , \label{eq:diffmp}
\end{align}
with $du_k=2\Re(\bar z_k\,dz_k)$, and similarly for $dv,dw$. Recall that $v$ and $\delta X$ represent the vertical and horizontal components of the tangent vector of $T^*M$. 
In the open regular set of $T^*(\mathrm{SU}(3)/T)$ with $\xi$ being regular and $z_1z_2\neq 0$, the map $v\mapsto [v,\xi]$ maps $\mathfrak{g}/\mathfrak{t}$ isomorphically onto $[\mathfrak{g},\xi]$, and the directions $[\mathfrak{g},\xi]$ together with $\mathfrak{m}$ generate $\mathfrak{g}$. That is, $\mathfrak{g} = [\mathfrak{g},\xi] \oplus \mathfrak{m}$. Hence, the differentials $\{dP_k\}_{k=1}^8$ span an $8$-dimensional subspace of covectors. In the same locus, the two differentials \begin{align}
    du_1 = 2 x_1 dx_1 + 2 y_1 dy_1, \quad du_2 = 2 x_2 dx_2 + 2 y_2 dy_2  \label{eq:uvdif}
\end{align} are linearly independent with $u_1=x_1^2+y_1^2$ and $u_2=x_2^2+y_2^2$ in the root coordinates, and they are not contained in the span of $\{dP_k\}_{k=1}^8$. 

To see this more clearly, we now introduce real coordinates $(\kappa_j,x_k,y_k)$ on $T^*(\mathrm{SU}(3)/T)$ and express the differential $d \pi_1$ computed above by terms of its Jacobian matrix. Here, $(\kappa_1,\ldots,\kappa_6)$ are real coordinates on $\mathrm{SU}(3)/T$. By the definition of the Jacobian matrix, we have \begin{align}
    J_{\pi_1}(g,X)= \begin{pmatrix}
\dfrac{\partial P_k}{\partial \kappa_j}(g,X) & \dfrac{\partial P_k}{\partial x_k}(g,X)&\dfrac{\partial P_k}{\partial y_k}(g,X) \\[6pt]
\dfrac{\partial u_1}{\partial \kappa_j}(g,X) & \dfrac{\partial u_1}{\partial x_k}(g,X)&\dfrac{\partial u_1}{\partial y_k}(g,X) \\[6pt]
\dfrac{\partial u_2}{\partial \kappa_j}(g,X) & \dfrac{\partial u_2}{\partial x_k}(g,X)&\dfrac{\partial u_2}{\partial y_k}(g,X) 
\end{pmatrix}_{10\times12}. \label{eq:Jacobian}
\end{align} 

Explicit calculations yield the following:

(1) Recall that the differentiation of $P$ is given in \eqref{eq:coordinate}.  For the coordinates of the map of $P_k$, we have the following: \begin{align*}
    \frac{\partial P_k}{\partial x_k}(g,X)=\left\langle\mathrm{Ad}(g) E_{\alpha_k},\xi_i\right\rangle,\quad \frac{\partial P_k}{\partial y_k}(g,X)=\left\langle\mathrm{Ad}(g) E_{-\alpha_k},\xi_i\right\rangle.
\end{align*}  Since $\{\mathrm{Ad}(g)(E_{\pm\alpha_k})\}_{k = 1}^3$ forms an orthonormal basis of the root space for the generic $g$, the $8\times6$ sub-block $\left(\dfrac{\partial P_k}{\partial x_k},\dfrac{\partial P_k}{\partial y_k}\right)$ has full rank $6$.

(2) The derivatives of $P_k$ with respect to the group parameters $\kappa_j$ are computed as \begin{align*}
    \dfrac{\partial P_k}{\partial \kappa_j}(g,X)= \left\langle[E_{\kappa_j},\mathrm{Ad}(g)(X-\varepsilon W)),\xi_i\right\rangle .
\end{align*} 

At generic points, the adjoint orbit of $X-\varepsilon W$ has maximal dimension, which ensures at least two independent directions from these group variations.

(3) From \eqref{eq:uvdif}, the differentials of $u_1,u_2$ are defined in a way that shows no dependence on the group parameters $\kappa_j$. Thus, these two rows are clearly linearly independent, giving us full rank $2$ to $J_{\pi_1}$. Due to the generality assumption, the coordinates satisfy $x_i,y_i\neq 0$.

Combining each step, the Jacobian $J_{\pi_1}(g,X)$ in \eqref{eq:Jacobian} is reduced to the following block form \begin{align*}
    J_{\pi_1}(g,X)=
\begin{pmatrix}
\dfrac{\partial P_k}{\partial \kappa_j}(g,X)&\frac{\partial P_k}{\partial(x_k,y_k)}(g,X)\\[6pt]
\textbf{0}_{2\times6}&\mathrm{diag}(2x_1,2y_1,2x_2,2y_2,0,0)
\end{pmatrix}_{10\times12}.
\end{align*} Let us analyze $J_{\pi_1}$ block by block. Since the $8\times6$ block $\partial P_k/\partial(x_k,y_k)$ has a rank $6$ from computation (1), and from computation (3), the $2\times6$ block for $(u_1,u_2)$ is a rank $2$ matrix, we have $6+2=8$. Additionally, the group derivatives $\partial P_k/\partial\kappa_j$ provide at least two independent columns, due to the maximal dimension of the adjoint orbits at generic points. Thus, the total rank is explicitly computed from the calculations above. That is, $
\mathrm{rank}\, J_{\pi_1}(g,X)=10 $. By standard differential geometry, this confirms that the differential $d\pi_1(g,X)$ is surjective onto the tangent space of $\mathrm{Spec}\,\mathcal{A}\cong\mathbb{R}^{10}$. Hence, on a dense open Zariski subset, $\mathrm{rank} \, J_{\pi_1}(g,X) = 10$, and $d\pi_1$ is surjective. By the continuity of rank along the manifold, $\pi_1$ is a Poisson submersion on an open dense set.

The inclusion $R_0\hookrightarrow\mathcal{A}$ endows $R_0=\mathbb{R}[J_2,J_3]$ with the induced (trivial) Poisson bracket, since $J_2,J_3$ are Casimirs. The coordinate map
\begin{align*}
\pi_2:\mathrm{Spec}\,\mathcal{A}\longrightarrow\mathrm{Spec}\,R_0,\qquad
(p_1,\ldots,p_8,u_1,u_2,u_3,v,w)\longmapsto\Bigl(\sum_{k=1}^8 P_k^2,\,2v\Bigr),
\end{align*}
is therefore a Poisson map whose differential has constant rank $2$. Hence $\pi_2$ is a Poisson submersion.

Finally, $\mathcal{A}$ is generated by the $13$ variables, subject to the three independent relations $u_1u_2u_3=v^2+w^2$, $J_2=\sum_k P_k^2$, and $J_3=2v$, so $\mathrm{trdeg}\,\mathcal{A}=10$. Since $\dim T^*M=12$, the generic fiber of $\pi_1$ has dimension $12-10=2$. By Definition \ref{def:superge}, the chain $T^*M\xrightarrow{\pi_1}\mathrm{Spec}\,\mathcal{A}\xrightarrow{\pi_2}\mathrm{Spec}\,R_0$ is superintegrable, with $\pi_1$ and $\pi_2$ being Poisson, and such fibers are coisotropic in $(T^*M,\omega_\varepsilon)$.
\end{proof}

\subsection{Superintegrable systems on \texorpdfstring{$\mathrm{SU}(3)/A$}{SU(3)/A}.}
\label{subsec:irregular}
 In this Subsection \ref{subsec:irregular}, we again take $\mathrm{SU}(3)$, and consider the set of traceless $3 \times 3$ anti-Hermitian complex matrices as its Lie algebra, denoted by $\mathfrak{su}(3)$. We use the Gell-Mann (GM) basis $\{\lambda_i\}_{i=1}^8$ for $\mathfrak{su}(3)$ with structure constants $[ \lambda_i,\lambda_j]=2\mathrm{i}\sum_{k = 1}^8c_{ij}^k\lambda_k$ as it simplifies component calculations in the irregular case. The explicit matrix representations of these basis elements are \begin{align*}
     \lambda_1 = & \, \begin{pmatrix}
         0 & 1 & 0\\
        1 & 0 & 0 \\
      0 & 0 & 0  \\
     \end{pmatrix}, \text{ }  \lambda_2 =     \begin{pmatrix}
         0 & -\mathrm{i} & 0\\
        \mathrm{i} & 0 & 0 \\
      0 & 0 & 0  \\
     \end{pmatrix},\text{ } \lambda_3 =     \begin{pmatrix}
         1 & 0 & 0\\
        0 & -1 & 0 \\
      0 & 0 & 0  \\
     \end{pmatrix}\text{ } \lambda_4 =     \begin{pmatrix}
         0 & 0 & 1\\
        0 & 0 & 0 \\
      1 & 0 & 0  \\
     \end{pmatrix} ;\\
     \lambda_5 = & \, \begin{pmatrix}
         0 & 0 & -\mathrm{i}\\
        0 & 0 & 0 \\
      \mathrm{i} & 0 & 0  \\
     \end{pmatrix}, \text{ }  \lambda_6 =     \begin{pmatrix}
         0 & 0 & 0\\
       0 & 0 & 1 \\
      0 & 1 & 0  \\
     \end{pmatrix},\text{ } \lambda_7 =     \begin{pmatrix}
        0 & 0 & 0\\
        0 & 0 & -\mathrm{i} \\
      0 & \mathrm{i} & 0  \\
     \end{pmatrix}\text{ } \lambda_8 =    \frac{1}{\sqrt{3}} \begin{pmatrix}
         1 & 0 & 0\\
        0 & 1 & 0 \\
      0 & 0 & -2 \\
     \end{pmatrix} .
 \end{align*} For further details, including the explicit form of the GM basis $\lambda_i$ and its transformation from the Chevalley basis, see, for example, \cite[Chapter II]{Cahn1984}. Now, let \begin{align}
 \nonumber
   H_1 &= \mathrm{i} \lambda_3,  \qquad   H_2 = - \frac{\mathrm{i}}{2}\left( \lambda_3 -  \sqrt{3}   \lambda_8 \right);\\
   \nonumber
    E_{12} = &\, \frac{1}{2}(\lambda_1 +  \mathrm{i}\lambda_2), \text{ } E_{21} = \frac{1}{2}(\lambda_1 - \mathrm{i}\lambda_2); \\
    E_{13} = &\, \frac{1}{2}(\lambda_4 +  \mathrm{i}\lambda_5), \text{ } E_{31} = \frac{1}{2}(\lambda_4 - \mathrm{i}\lambda_5); \label{eq:GMbas} \\
    \nonumber
    E_{23} = &\, \frac{1}{2}(\lambda_6 +  \mathrm{i}\lambda_7), \text{ } E_{32} = \frac{1}{2}(\lambda_6 - \mathrm{i}\lambda_7). 
\end{align}  
Throughout this subsection, $\mathrm{i} =\sqrt{-1}$ is a complex number.  We note that this basis transformation is computationally convenient. One can simply check the commutator relations, for example, by \begin{align*}
    [E_{12},E_{21}]  = H_1,  \qquad  [E_{23},E_{32}] = H_2,  \qquad  [E_{13},E_{31}] = H_1 + H_2.
\end{align*}  
Under the basis \eqref{eq:GMbas}, we may write $$H = c_1 \lambda_3 + c_2 \lambda_8 =\mathrm{diag} \left(c_1 + \frac{c_2}{\sqrt{3}}, -c_1 + \frac{c_2}{\sqrt{3}}, -\frac{2c_2}{\sqrt{3}}\right)$$ for some $c_1,c_2 \in \mathbb{R}$. By the definition of regularity, there exists a root $\alpha$ such that $\alpha(H) = 0$. A direct computation shows that the following element  \begin{align}
    W =\frac{3}{2} \lambda_3+\frac{\sqrt3}{2} \lambda_8 =\mathrm{diag}(2,-1,-1)\label{eq:irrw}
\end{align} is irregular, such that the centraliser is given by  \begin{align*}
    \mathfrak{a}:=\ker\mathrm{ad}W     =\mathrm{span} \{\lambda_1,\lambda_2,\lambda_3,\lambda_8\}.
\end{align*} 
Moreover, $T \subsetneq A$. The orthogonal complement, with respect to the Killing form $B$, is  $ \mathfrak{m}:= \mathrm{span}\{\lambda_4,\lambda_5,\lambda_6,\lambda_7\}  $, where $$  B(X,Y) := - \frac{1}{2}\mathrm{tr}(XY).$$   
With the irregular element in \eqref{eq:irrw}, we now construct the stabiliser. By definition, $\mathrm{SU}(3) = \{A \in M_3(\mathbb{C}):\, \det A = 1 \}$, which represents the group action on $\mathbb{C}^3$.   Work with the defining representation of $G=\mathrm{SU}(3)$ in $\mathbb{C}^3$ and fix $W=\mathrm{diag}(2,-1,-1)$. Since $W$ is normal, we have $\mathbb{C}^3 \cong E_2 \oplus E_{-1}$, where $E_2$ is a one-dimensional $2$-eigenspace and $E_{-1}$ is a two dimensional $-1$-eigenspace. The corresponding matrix relative to this decomposition yields $1\times 1$ and $2\times 2$ block matrices. The centralizer \begin{align*}
A=C_G(W)=\{a\in \mathrm{SU}(3):\, aW=Wa\}
\end{align*} consists of those unitary matrices that preserve each eigenspace of $W$. Hence, a simple computation shows that $A$ is block-diagonal on this basis: \begin{align}
a=\begin{pmatrix} s & 0\\[2pt] 0 & U\end{pmatrix}\qquad\text{with }s\in \mathrm{U}(1),\ U\in \mathrm{U}(2). \label{eq:generatora}
\end{align} The condition $\det a=1$ forces $s\cdot \det(U)=1$, i.e., $a=\mathrm{diag}\bigl(\det(U)^{-1},\,U\bigr)\in \mathrm{SU}(3)$. Hence, in the matrix form,  \begin{align*}
    A = C_{\mathrm{SU}(3)}(W) = \left\{\mathrm{diag}\left(\det(U)^{-1},\,U\right): \, U \in \mathrm{U}(2)\right\} \cong \mathrm{S}(\mathrm{U}(2) \times \mathrm{U}(1)).
\end{align*}  In Lie algebra setting, we therefore have a reductive decomposition $\mathfrak{su}(3)=\mathfrak{s}(\mathfrak{u}(2)\oplus\mathfrak{u}(1))\oplus \mathfrak{m}$, where \begin{align}
    \mathfrak{m}=\left\{\begin{pmatrix}
    0&0&u_1\\
    0&0&u_2\\
    -\overline{u}_1&-\overline{u}_2&0\end{pmatrix}: u_1,u_2\in\mathbb{C}\right\}. \label{eq:matrix}
\end{align} Thus, in what follows, we identify $\mathfrak{m}$ with $\mathbb{C}^2$ by associating $(u_1,u_2) \in \mathbb{C}^2$ with the traceless matrix in \eqref{eq:matrix} with $u_1 = x_4  - \mathrm{i} x_5$ and $u_2 = x_6 - \mathrm{i} x_7$.

We now build the twisted symplectic structure on it. Recall that we have the following affine slice  $  \mathfrak{m}-\varepsilon W  =\{X-\varepsilon W:X\in\mathfrak{m}\} $, which can therefore be written as \begin{align*}
\xi  =x_4\lambda_4+x_5\lambda_5+x_6\lambda_6+x_7\lambda_7-\varepsilon W, \qquad x_4,\ldots ,x_7\in\mathbb{R} .
\end{align*} Under the choices of matrix representation given in \eqref{eq:matrix}, this can be explicitly written as a 3$\times$3 traceless skew-Hermitian matrix as follows: \begin{align}
  \xi:=  X-\varepsilon W= \begin{pmatrix}
-2\varepsilon &0& x_4-\mathrm{i}x_5\\[6pt]
0 &\varepsilon & x_6-\mathrm{i}x_7\\[6pt]
x_4+\mathrm{i}x_5 & x_6+\mathrm{i}x_7 & \varepsilon
\end{pmatrix}\in\mathfrak{su}(3). \label{eq:ximatrixsu(3)}
\end{align}

 Near the identity, choose a section \begin{align*}
    \sigma: U\subset\mathfrak{m}\xrightarrow{\,\exp\,}\exp(U)\subset G   .
\end{align*} Elements of $G/A$ are represented by $m=\sigma(y) = \exp(y)$ with $y=(y_4,\ldots ,y_7) \in \mathbb{R}^{4}$. Identify $\mathfrak{m}^*$ with $\mathfrak{m}$ in the Killing form $B$. A point in the local trivialization of $T^*M \cong (G \times \mathfrak{m})/A $ is therefore given by:  \begin{align}
    (g,X):=\bigl(\sigma(y),\,X\bigr), \qquad X=\sum_{i= 4}^{7}x_i\lambda_i\in\mathfrak{m}.
\end{align}    The canonical Liouville 1-form is $\varsigma=B( X,\,\sigma^{-1}d\sigma)$, where $\sigma^{-1}d\sigma$ is the pullback along $\sigma$ of the left Maurer-Cartan form on $\mathrm{SU}(3)$. See \cite[Proof of Proposition 2.12]{Davis2012HomogeneousSymplecticGalilei} for more detailed on this construction. Adding the magnetic 2-form $\varepsilon\,\pi_1^*\omega_{\mathrm{KKS}}$ with $W =\frac{3}{2}\lambda_3+\frac{\sqrt3}{2}\lambda_8$, the resulting symplectic 2-form is \begin{align*}
   \omega_{\varepsilon}  =d\varsigma + \varepsilon\,B( W,\,[\sigma^{-1}d\sigma,\sigma^{-1}d\sigma]) . 
\end{align*} In the chosen coordinates $(x_4,\ldots,x_7,y_4,\ldots,y_7)$, the non-vanishing Poisson brackets are \begin{align*}
    \{x_i,y_j\}_\varepsilon=2\,\delta_{ij}, \qquad \{x_i,x_j\}_\varepsilon=-2\varepsilon\,c_{ij}^k B( W,\lambda_k), \qquad \{y_i,y_j\}_\varepsilon=0  
\end{align*} for any $i,j= 4,\ldots ,7$ and $k = 1,\ldots,8$.

\subsubsection{Building superintegrable systems on $T^*(\mathrm{SU}(3)/A)$}

Recall that the magmatic moment map is  $P(g,X)=\mathrm{Ad}(g)\bigl(X -\varepsilon W\bigr)\in\mathfrak{su}^*(3)$. Evaluating the representative $\xi $ of the orbit $Z_P =\mathrm{Ad}(G)(\mathfrak{m}-\varepsilon W)$, we obtain the eight linear generators as follows: \begin{align*}
P_i:=\ell_i\circ P(g,X) =\ell_i\bigl(x_4\lambda_4+\ldots + x_7\lambda_7-\varepsilon W\bigr), 
\end{align*} where $\ell_i(P) = \langle P,\lambda_i \rangle$. A direct trace calculation yields
\begin{align*}
P_1& \, =P_2=0,  \quad P_3 =\frac{3\varepsilon}{2}, \quad P_4 =-x_4,\\
P_5& \,=-x_5, \text{ }  P_6 =-x_6, \text{ } P_7 =-x_7, \text{ }  P_8 =\frac{\sqrt3\,\varepsilon}{2} .
\end{align*} Thus, the linear coordinates $(x_4,x_5,x_6,x_7)$ are $(-P_4,-P_5,-P_6,-P_7)$. Define the coordinate functions in $\mathfrak{su}(3)$, \begin{align*}
   P_i(Y)=-\frac{1}{2}\mathrm{tr}(Y\lambda_i),\quad i=1,\ldots,8. 
\end{align*} In particular, \begin{align*}
    P (g,X)=\sum_{i=1}^8P_i(g,X)\lambda_i = \begin{pmatrix}
    \frac{P_3}{2} + \frac{P_8}{2 \sqrt{3}} & -\frac{1}{2}(P_1 - \mathrm{i}P_2) & -\frac{1}{2}(P_4 - \mathrm{i}P_5) \\
    -\frac{1}{2}(P_1 + \mathrm{i}P_2) & - \frac{P_3}{2} +  \frac{P_8}{2 \sqrt{3}} & -\frac{1}{2}(P_6 - \mathrm{i}P_7) \\
     -\frac{1}{2}(P_4 + \mathrm{i}P_5) & -\frac{1}{2}(P_6 + \mathrm{i}P_7)  & - \frac{P_8}{2 \sqrt{3}} \\ 
\end{pmatrix},\quad R(g,X)=P_4^2+\ldots+P_7^2.
\end{align*} Thus  \begin{align*}
    Z_P=\{P ((g,X)):(g,X)\in T^*M\},\quad \dim  Z_P =7.
\end{align*} 

In the following example, we will provide the commutator relations of the first integrals in $\mathfrak{F}_1^{\mathrm{poly}}$ under the twisted Poisson bracket. 
 
\begin{example} \label{ex:GellMann}
Let $\{\lambda_i\}_{i=1}^8$ be the GM basis defined above with $[\lambda_i,\lambda_j]=2\mathrm{i}\,c_{ij}^k\,\lambda_k$.  Choose $B(X,Y):=-\frac{1}{2}\,\mathrm{tr}(XY)$.  Then $P_i =\langle P,\lambda_i \rangle$ and \begin{align}
\{P_i,P_j\}_\varepsilon= \langle P,[\lambda_i,\lambda_j] \rangle= \langle P,2\mathrm{i}\,c_{ij}^k\,\lambda_k \rangle = 2\mathrm{i}\,c_{ij}^k\, \langle P,\lambda_k\rangle= 2\mathrm{i}c_{ij}^k\,P_k. \label{eq:linearbracket}
\end{align}  
For instance, taking the structure constant to be $c_{12}^3=1$. Then the twisted Poisson brackets are given by \begin{align*}
\{P_1,P_2\}_\varepsilon= 2\mathrm{i} P_3,\qquad \{P_1,P_3\}_\varepsilon=- 2\mathrm{i} P_2,\qquad \{P_2,P_3\}_\varepsilon= 2\mathrm{i} P_1,
\end{align*} and similarly for the other nonzero structure constants $c_{ij}^k$:
\begin{align*}
c_{14}^7=c_{24}^6=c_{25}^7=c_{34}^5=c_{51}^6=c_{63}^7=\frac{1}{2},\qquad c_{45}^8=c_{67}^8=\frac{\sqrt{3}}{2} 
\end{align*} with antisymmetry in the lower indices giving the rest of the relations. These identities match \eqref{eq:f1postr}.
\end{example}
 
\begin{remark} 
All $\varepsilon$-dependence of the magnetic deformation is absorbed into the shift $P(g,X)=\mathrm{Ad}(g)(X-\varepsilon W)$.  The calculation above uses only the $\mathrm{Ad}$-equivariance of $P$ and the fact that $X_{P_\xi}$ is the fundamental left field. Recall that by Proposition \ref{prop:slicepoisson}, the bracket induced from $P$ is therefore the ordinary Lie-Poisson bracket in $\mathfrak{g}$, without explicit $\varepsilon$ in \eqref{eq:f1bracket}. Thus, we verified again, on $\mathfrak{F}_1^{\mathrm{poly}}$, we have $\{P_i,P_j\}_\varepsilon$ exactly as in the $\varepsilon = 0$ case, while the $\mathfrak{F}_2^{\mathrm{poly}}$ brackets carry the $\varepsilon$ terms.
\end{remark}

Recall that from \eqref{eq:ximatrixsu(3)}, we have the matrix form $\xi=\begin{pmatrix}\varepsilon\,\mathbf{1}_2 & z\\[2pt] 
z^* & -2\varepsilon\end{pmatrix}$. A direct block multiplication gives the following: \begin{align*}
\xi^2=\begin{pmatrix}\varepsilon^2\mathbf{1}_2+zz^* & - \varepsilon z\\[2pt]
- \varepsilon z^* & z^*z+4\varepsilon^2\end{pmatrix}.
\end{align*} Hence, \begin{align*}
\mathrm{tr}(\xi^2)=6\varepsilon^2+2\left(|z_1|^2  + |z_2|^2\right), \text{ and } C_2(\xi)=\frac{1}{2}\,\mathrm{tr}(\xi^2)=3\varepsilon^2+\left(|z_1|^2  + |z_2|^2\right),
\end{align*} where $C_2(\xi)$ is the restricted Casimir from $S(\mathfrak{g})$ into $S(\mathfrak{m}-\varepsilon W)$. Using trace normalizations, the quadratic and cubic Casimir are given by $C_2(\xi) = \frac{1}{2} \mathrm{tr}(\xi^2)$ and $C_3(\xi) = - \mathrm{tr}(\xi^3)$. Recall that from \eqref{eq:Casimir},  these Casimir can be rewritten by $  C_2:=\sum_{i=1}^8P_i^2$ and   $C_3:=\sum_{i,j,k=1}^8d_{ijk}P_iP_jP_k$. Under the real coordinate expression, we further rewrite $C_2(\xi)$ into the following form: \begin{align*}
 C_2(\xi)=x_4^2+x_5^2+x_6^2+x_7^2+3\varepsilon^2.
\end{align*} Define the $A$-invariant quadratic monomial by $R:=x_4^2+x_5^2+x_6^2+x_7^2$. Then $C_2(\xi)=R+3\varepsilon^2$.  A straightforward 
symbolic computation gives\begin{align}
    C_3(\xi)=3\varepsilon \bigl(2\varepsilon^2+x_4^2+x_5^2-2x_6^2-2x_7^2\bigr).  \label{eq:3C}
\end{align} Write $w:=x_4^2+x_5^2 $ and $t:=x_6^2+x_7^2$ such that $R=w+t$. Then \eqref{eq:3C} becomes \begin{align}
C_3(\xi)=3\varepsilon\bigl(2\varepsilon^2+w-2t\bigr)
=3\varepsilon\bigl(2\varepsilon^2 +R - 3t\bigr).
\end{align} Since $w-t$ is not $A$-invariant, the cubic Casimir does not survive the restriction map $\mathrm{Res}_W:S(\mathfrak{g})^G\rightarrow S(\mathfrak{m})^A$. Therefore, by the definition of the restriction mapping, we deduce  \begin{align*}
\mathrm{Im}\mathrm{Res}_W=\mathbb{R}[R]\, \Longrightarrow \, s=1  .
\end{align*}   Substituting the restricted Casimir above, we eliminate the slice coordinates $x_4,\ldots,x_7$ and obtain a single cubic relation among the eight moment coordinates  \begin{align*}
      \Phi(P):= C_3(\xi) - 3\varepsilon \Bigl(2\varepsilon^2 + P_4^2+P_5^2-2P_6^2-2P_7^2\Bigr) = 0.
  \end{align*}  No quadratic relation appears as $C_2(\xi)$ is free on $\mathrm{Ad}(\mathrm{SU}(3))(\mathfrak{m} - \varepsilon W)$. Hence, \begin{align*}
      \mathfrak{F}_1^{\mathrm{poly}}  =P^*S(\mathfrak{g})   \cong\mathbb{R}[P_1,\ldots  ,P_8]\big/ \langle\Phi\rangle \text{ and } \mathrm{rank}\,\mathfrak{F}_1^{\mathrm{poly}}=7 .
  \end{align*} Concluding from this result, we present the following lemma. 
 
 \begin{lemma}
 \label{lem:F1poly2}
Let $G=\mathrm{SU}(3)$, $A\subset G$ be the same as defined above. Let  $P:T^*M\rightarrow\mathfrak{g}^*$ given by $P (g,X)=\mathrm{Ad}(g)\bigl(X-\varepsilon W\bigr)$ be the magnetic moment map for the left action, and for $\xi\in\mathfrak{g}$ set\begin{align*}
P_{\xi}(g,X):=B\bigl(P(g,X),\xi\bigr),\qquad P_k:= P_{\xi_k} 
\end{align*} with $k=1,\dots,8$.   Then\begin{align*}
\mathfrak{F}_1^{\mathrm{poly}}=P^*S(\mathfrak{g})=\mathbb{R}[P_1,\dots,P_8]\big/\left\langle \Phi\right\rangle \quad \text{ with } \deg P_k=1,
\end{align*} and the Poisson commutator relations in coordinates are \begin{equation}
\{P_i,P_j\}_\varepsilon = \sum_{k=1}^8 c_{ij}^k\,P_k,\qquad [\xi_i,\xi_j]=\sum_{k=1}^8 c_{ij}^k\,\xi_k. \label{eq:f1postr}
\end{equation}
\end{lemma}


We now compute the generators of $\mathfrak{F}_2^{\mathrm{poly}}$. In other words, it is equivalent to deducing the generators of $S(\mathfrak{m})^A$. By construction, the adjoint orbit $\mathrm{SU}(3)/\mathrm{S}(\mathrm{U}(2) \times \mathrm{U}(1))$ is identified with a homogeneous space $\mathbb{CP}^2$. Here, we will use the standard tangent space description for the Grassmannians of Hermitian spaces. For the $2$-dimensional eigenspace of $W$, set a line $E \subset \mathbb{C}^2$ such that $[E] \in \mathbb{CP}^2$. In particular, there exists a canonical identification as follows: $T_{[E]}M \cong \mathrm{Hom}(E,E^\perp) \cong E^\perp \otimes E^* \cong \mathbb{C}^2$.  With respect to the same block decomposition, 
every $X\in \mathfrak{m}$ is unique in form of
\begin{align*}
X=\begin{pmatrix} 0 & z^{ *}\\[2pt] z & 0\end{pmatrix}\quad\text{for a unique column }z\in\mathbb{C}^2,
\end{align*} so $\mathfrak{m}$ is naturally identified with $\mathbb{C}^2$ through the off-block columns $\begin{pmatrix}
    0\\
    z
\end{pmatrix}$ and $\begin{pmatrix}
    z^* \\
    0
\end{pmatrix}$, where $z^* = \overline{z}^T$ and $\overline{z}$ is the conjugation of the complex. Hence, under the identification $\mathfrak{m}\cong\mathbb{C}^2$, the induced linear action of $A$ on $\mathfrak{m}$ is \begin{align*}
z\longmapsto \rho(A)z: = \det(U)\,U z\qquad(U\in \mathrm{U}(2)).
\end{align*}   Recall that from \eqref{eq:generatora}, an element $a\in A$ is a block diagonal matrix \begin{align*}
    a=\begin{pmatrix} U & 0\\ 0 & \det(U)^{-1}\end{pmatrix},\qquad U\in \mathrm{U}(2). 
\end{align*} For $a=\mathrm{diag}(\det(U)^{-1},U)\in A$, a direct computation shows that, through conjugation actions, we have  
\begin{align}
\mathrm{Ad}(a)(X)=aXa^{-1} =\begin{pmatrix} 0 & \det(U)^{-1}\,z^{ *} U^{-1}\\[2pt] \det(U)\,U z & 0\end{pmatrix}. \label{eq:Aaction}
\end{align} The adjoint action on the slice is $X-\varepsilon W\mapsto a(X-\varepsilon W)a^{-1}$. The usual matrix multiplication with the block form $$ X-\varepsilon W = \begin{pmatrix} D & z \\ z^* & \varepsilon\end{pmatrix},
\quad D=\begin{pmatrix}-2\varepsilon&0\\0&\varepsilon\end{pmatrix},
$$ gives  \begin{align}
  \mathrm{Ad}(a)(\xi)=a\,\xi\,a^{-1} = \begin{pmatrix}
U & 0\\[2pt]
0&\det(U)^{-1}
\end{pmatrix} \begin{pmatrix}
\varepsilon\,\mathbf{1}_2 & z\\[2pt]
z^* & -2\varepsilon
\end{pmatrix} \begin{pmatrix}
U^{-1} & 0\\[2pt]
0&\det(U)
\end{pmatrix} = \begin{pmatrix}
\varepsilon\,\mathbf{1}_2 & \det(U)\,U z\\[2pt]
\det(U)^{-1} z^* U^{-1} & -2\varepsilon
\end{pmatrix}. \label{eq:adaaction}
\end{align}  
This can be summarised into the following proposition.

\begin{proposition}
\label{prop:Adaction}
Fix $G=\mathrm{SU}(3)$, $\mathfrak{g}=\mathfrak{su}(3)$, the Killing form $B$, and
\begin{align*}
W=\mathrm{diag}(2,-1,-1),\qquad A=C_G(W)\cong  \mathrm{S}(\mathrm{U}(2)\times \mathrm{U}(1)).
\end{align*}
Let $\mathfrak{g}=\mathfrak{a}\oplus\mathfrak{m}$ be the reductive $B$-orthogonal decomposition with $\mathfrak{a}=\ker(\mathrm{ad}_W)$. Identify $\mathfrak{m}$ with $\mathbb{C}^2$ by \begin{align*}
z=(z_1,z_2)^T\longleftrightarrow X(z)=\begin{pmatrix}0 & z^*\\[2pt] z & 0\end{pmatrix}\in\mathfrak{m},
\end{align*} such that for $a=\mathrm{diag}(\det(U)^{-1},U)\in A$ with $U\in \mathrm{U}(2)$, we have
\begin{align*}
\mathrm{Ad}(a)\bigl(X(z)\bigr)=X\bigl(\rho(U)z\bigr),\qquad \rho(U) =\det(U)\,U.
\end{align*}
\end{proposition}

 We now focus on finding the generators of the finitely generated polynomial algebra $S(\mathfrak{m})^A$. For any $F\in S(\mathfrak{m})$, set $f (z,\overline{z}):=F\circ X\in\mathbb{R}[z_1,z_2,\bar z_1,\bar z_2]$. By definition,  for any $a \in A$, we have $(a\cdot F)(X):=F\bigl(\mathrm{Ad}(a^{-1})X\bigr)$. In other words, by Proposition \ref{prop:Adaction}, we deduce \begin{align}
    (a\cdot f)(z,\bar z)=f \bigl(\rho(U)^{-1}z,\ \overline{\rho(U)^{-1}z}\bigr).
\end{align}  We then process the coordinate computation in this example.

\begin{lemma}
\label{lem:liealge}
Let $A=\mathrm{S}(\mathrm{U}(2) \times \mathrm{U}(1))$ act on the off-block column $z=(z_1,z_2)^T\in\mathbb{C}^2$ by the unitary representation \begin{align}
\rho(U)z=\det(U)\,U z\qquad \text{ for any } U\in \mathrm{U}(2).
\end{align} Suppose that $f=f(z,\overline{z}) \in S(\mathfrak{m})^A$. Then, for every $H\in\mathfrak{u}(2)$, we have the linear first-order PDE \begin{align*}
\sum_{j=1}^2 \frac{\partial f}{\partial z_j}\,\bigl((H+\mathrm{tr}H\,\mathbf{1}_2)z\bigr)_j +  \sum_{j=1}^2 \frac{\partial f}{\partial \overline{z}_j}\,\bigl((\overline{H+\mathrm{tr}H\,\mathbf{1}_2})\,\overline{z}\bigr)_j =0,
\end{align*} where $\mathbf{1}_2$ is a $2 \times 2$ identity matrix. Equivalently,  \begin{align*}
B( \partial_{z}f,\;(H+\mathrm{tr}H\,\mathbf{1}_2)z) + B( \partial_{\overline{z}}f,\;(\overline{H+\mathrm{tr}H\,\mathbf{1}_2})\,\overline{z})=0,
\end{align*} where $\partial_z f = \dfrac{\partial f}{\partial z}$ and $\partial_{\overline{z}} f = \dfrac{\partial f}{\partial \overline{z}} $.
\end{lemma}

\begin{proof}
Fix $H\in\mathfrak{u}(2)$ and consider the one-parameter subgroup\begin{align*}
a(t)=\begin{pmatrix}\exp(tH)&0\\[2pt]0&\det(\exp(tH))^{-1}\end{pmatrix}\in A,\qquad\rho\bigl(\exp(tH)\bigr)=\det \bigl(\exp(tH)\bigr)\,\exp(tH).
\end{align*} Since $f$ is $A$-invariant, for all $t$ near $0$,\begin{align*}
f\Bigl(\rho\bigl(\exp(tH)\bigr)z,\;\overline{\rho\bigl(\exp(tH)\bigr)z}\Bigr)=f(z,\overline{z}).
\end{align*}Differentiate at $t=0$ and using the chain rule, we have 
\begin{align}
\nonumber
0 &= \left.\frac{d}{dt} \right\vert_{t=0} f\Bigl(\rho\bigl(\exp(tH)\bigr)z,\overline{\rho\bigl(\exp(tH)\bigr)z}\Bigr)\\
&=\sum_{j=1}^2 \frac{\partial f}{\partial z_j}\,\left.\frac{d}{dt}\right\vert_{t=0}\bigl(\rho\bigl(\exp(tH)\bigr)z\bigr)_j+\sum_{j=1}^2 \frac{\partial f}{\partial \overline{z}_j}\,\left.\frac{d}{dt}\right\vert_{t=0}\,\overline{\bigl(\rho\bigl(\exp(tH)\bigr)z\bigr)_j}. \label{eq:differential}
\end{align} We then compute the derivative of $\rho\bigl(\exp(tH)\bigr)z$ at $t=0$ in \eqref{eq:differential}. By implement the definition of $\rho$, this gives that\begin{align*}
\left.\frac{d}{dt}\right\vert_{t=0}\rho\bigl(\exp(tH)\bigr)= \left.\frac{d}{dt}\right\vert_{t=0}\Bigl(\det(\exp(tH))\,\exp(tH)\Bigr)=\Bigl(\mathrm{tr}H\,\mathbf{1}_2+H\Bigr).
\end{align*}Hence, \begin{align*}
\left.\frac{d}{dt}\right\vert_{t=0}\bigl(\rho\bigl(\exp(tH)\bigr)z\bigr)=(H+\mathrm{tr}H\,\mathbf{1}_2)z,\qquad \left.\frac{d}{dt}\right\vert_{t=0}\overline{\bigl(\rho\bigl(\exp(tH)\bigr)z\bigr)}=(\overline{H+\mathrm{tr}H\,\mathbf{1}_2})\,\overline{z}.
\end{align*}Substituting these into the chain rule identity and using Lemma \ref{lem:liealge} yields precisely the displayed PDE.
\end{proof}

 Now, with the following proposition, we have a geometric derivation on finding the generator of $S(\mathfrak{m})^A$.

\begin{proposition} 
\label{prop:F2bracket}
 Let $R(z):=||z||^2=z^* z=|z_1|^2+|z_2|^2 = x_4^2 + x_5^2 + x_6^2 + x_7^2 $. Then  $S(\mathfrak{m}-\varepsilon W)^A=\mathbb{R}\bigl[R\bigr]$, and consequently
\begin{align*}
\mathfrak{F}_2^{\mathrm{poly}}=\pi_{\mathfrak{m}}^*\bigl(S(\mathfrak{m}-\varepsilon W)^A\bigr)\cong \mathbb{R}[R].
\end{align*}
\end{proposition}

\begin{proof}
By Lemma \ref{lem:property}, $S(\mathfrak{g})^A$ restricting along $\tau_{-\varepsilon W}$ produces an isomorphism $S(\mathfrak{m})^A\cong S(\mathfrak{m}-\varepsilon W)^A$. Thus, it suffices to compute $S(\mathfrak{m})^A$. Under the given identification, $A$ acts unitarily on $\mathbb{C}^2$ through $\rho$ and the real polynomial ring on $S(\mathfrak{m}) \cong \mathbb{R}[z_1,z_2,\overline{z}_1,\overline{z}_2]$. Moreover, recall that from \eqref{eq:adaaction}, conjugation by $a=\mathrm{diag}(U,\det(U)^{-1})\in A$ gives  $\mathrm{Ad}(a)(\xi)=  \begin{pmatrix}
\varepsilon\,\mathbf{1}_2 & \det(U)\,U z\\[2pt]
\det(U)^{-1} z^* U^{-1} & -2\varepsilon
\end{pmatrix}$. The first Fundamental Theorem for $\mathrm{U}(2)$ shows that the generators in $\mathbb{R}[z,\overline{z}]^{\mathrm{U}(2)}$ are generated by the Hermitian pairings $z^* z$. Equivalently, we deduce directly as follows. Since $\mathbb{R}[z_1,\overline{z}_2]$ is a graded algebra, we decompose it into bihomogeneous components: \begin{align*}
\mathbb{R}[z,\overline{z}]=\bigoplus_{p,q\ge 0}\, \mathcal{P}^{p,q},\qquad \text{ where } \mathcal{P}^{p,q}=\mathrm{span}\{\,z^\alpha \overline{z}^\beta:\ |\alpha|=p,\ |\beta|=q\,\}.
\end{align*} As a $\mathrm{U}(2)$-module, $\mathcal{P}^{p,q}\cong \mathrm{Sym}^p(\mathbb{C}^2)\otimes \mathrm{Sym}^q(\overline{\mathbb{C}^2})$.
By Schur’s lemma, the invariant subspace $\bigl(\mathcal{P}^{p,q}\bigr)^{\mathrm{U}(2)}$ is $0$ unless $p=q$, and when $p=q$, it is one-dimensional, spanned by $(z^* z)^p$. Therefore, \begin{align*}
\mathbb{R}[z,\overline{z}]^{\mathrm{U}(2)}=\mathbb{R}[z^* z]=\mathbb{R}[R],
\end{align*} as claimed. Translating back to $\mathfrak{m}-\varepsilon W$ gives $S(\mathfrak{m}-\varepsilon W)^A=\mathbb{R}[R]$.

In particular, as $R = C_2 - 3\varepsilon^2$ on $\mathfrak{m} - \varepsilon W$, hence $R$ is $A$-invariant and generates $S(\mathfrak{m} - \varepsilon W)^A$. Applying $\pi_\mathfrak{m}^*$ yields \begin{align*}
\mathfrak{F}_2^{\mathrm{poly}}=\pi_{\mathfrak{m}}^*\bigl(S(\mathfrak{m}-\varepsilon W)^A\bigr)=\pi_{\mathfrak{m}}^*\left(\mathbb{R}\bigl[R\bigr]\right)
\end{align*} as claimed.
\end{proof}

Using the polynomial ansatz mentioned in \cite{campoamor2024superintegrable,campoamor2026construction}, the generators in $S(\mathfrak{m} - \varepsilon W)^A$ can also be obtained via direct computation.

\begin{remark}
 Choose the basis of $\mathfrak{u}(2)$ given by $H_1=\mathrm{i}\sigma_1$, $H_2=\mathrm{i}\sigma_2$, $H_3=\mathrm{i}\sigma_3$, $H_0=\mathrm{i}\textbf{1}_2$, where $\sigma_k$, $k = 0,1,2,3$, are the Pauli matrices and $\textbf{1}_2$ is a $2 \times 2$ identity matrix. Writing $z_1=x_4+\mathrm{i} x_5$, $z_2=x_6+\mathrm{i} x_7$, the corresponding fundamental vector fields on $\mathbb{R}^4$ with coordinates $(x_4,x_5,x_6,x_7)$ are \begin{align*}
V_1&=-x_7\,\partial_{x_4}+x_6\,\partial_{x_5}-x_5\,\partial_{x_6}+x_4\,\partial_{x_7},
\\
V_2&=x_6\,\partial_{x_4}+x_7\,\partial_{x_5}-x_4\,\partial_{x_6}-x_5\,\partial_{x_7},
\\
V_3&=-x_5\,\partial_{x_4}+x_4\,\partial_{x_5}+x_7\,\partial_{x_6}-x_6\,\partial_{x_7},
\\
V_0&=-3x_5\,\partial_{x_4}+3x_4\,\partial_{x_5}-3x_7\,\partial_{x_6}+3x_6\,\partial_{x_7}.
\end{align*} Recall that finding the $A$-invariant polynomial $f$ is equivalent to solving the linear system of first-order PDEs \begin{align*}
V_k(f)=0\quad \text{ with } k=0,1,2,3.
\end{align*} Hence, $S(\mathfrak{m})^A = \ker \bigcap_{k = 0}^3 V_k$. The quadratic monomial $R=x_4^2+x_5^2+x_6^2+x_7^2=|| z||^2$ satisfies $V_k(R)=0$ for all $k$, and conversely, solving the system in the polynomial ansatz forces $f$ to be a polynomial in $R$. Hence, $S(\mathfrak{m}-\varepsilon W)^A=\mathbb{R}[R]$.
\end{remark}

 As illustrated in Proposition \ref{prop:F2bracket}, in the irregular case, $\mathfrak{F}_2^{\mathrm{poly}}$ consists of only one $A$-invariant generator, in contrast to the $4$ generators in the regular case. Indeed, this aligns with the fact that an adjoint orbit of type $A_2$ has an irregular stabiliser with dimension $4$, which yields a smaller family of first integrals.  


Since $C_2=R+3\varepsilon^2$ and $R$ are $G$-invariant, $\mathbb{R}[R]$ lies inside $\mathfrak{F}_2^{\mathrm{poly}}$.  Hence, $$ R_0:=\mathfrak{F}_1^{\mathrm{poly}}\cap \mathfrak{F}_2^{\mathrm{poly}}    = P^*\mathbb{R}[R] = \pi_\mathfrak{m}^* \mathbb{R}[R].$$ With all the construction above, we see that $ \mathfrak{F}_1^{\mathrm{poly}}\otimes_{R_0}\mathfrak{F}_2^{\mathrm{poly}}  = \mathbb{R}[P_1,\ldots,P_8]/\langle \Phi \rangle \otimes_{\pi_\mathfrak{m}^* \mathbb{R}[R]} \pi_\mathfrak{m}^* \mathbb{R}[R]$ is verified by Theorem \ref{thm:superin}.  All Poisson brackets follow from $$\{P_i,P_j\}_\mathcal{A} =  2c_{ij}^k P_k, \  \{P_i,R\}_\mathcal{A} = 0, \text{ and } \{R,R\}_\mathcal{A} = 0.$$ This further implies that the Poisson commutative base algebra $\mathcal{B} = R_0$. Next, by a direct rank computation, we show that $T^*(\mathrm{SU}(3)/A), \mathcal{A}$ and $R_0$ forms a superintegrable system with Poisson projection chains. 

\begin{proposition}
    Let $G = \mathrm{SU}(3)$ and $A = \mathrm{S}(\mathrm{U}(2) \times \mathrm{U}(1))$ be defined as $\mathfrak{F}_1^{\mathrm{poly}}$, $ \mathfrak{F}_2^{\mathrm{poly}}$, $R_0$, and $\mathcal{A}$ as stated in Lemma \ref{lem:F1poly2} and Proposition \ref{prop:F2bracket}. Then there exist two Poisson projections $\pi_1$ and $\pi_2$ such that the following chain  \begin{align}
  T^*M\xrightarrow{\pi_1}Z_P=\mathrm{Spec} \, \mathcal{A} \xrightarrow{\pi_2}\mathrm{Spec}\, R_0\cong\mathbb{A}^1   \label{eq:superchain2}
\end{align} is superintegrable.   
\end{proposition}
\begin{proof}
 By Lemma \ref{lem:F1poly2} and Proposition \ref{prop:F2bracket}, we deduce $ \mathfrak{F}_1^{\mathrm{poly}} =\mathbb{R}[P_1,\ldots,P_8]\big/ \langle \Phi \rangle, \mathfrak{F}_2^{\mathrm{poly}} = R_0 =  \pi_\mathfrak{m}^*\mathbb{R}[R]$ and $\mathcal{A} = \mathbb{R}[P_1,\ldots,P_8]\big/ \langle \Phi \rangle$ $ \otimes_{\pi_\mathfrak{m}^*\mathbb{R}[R]} \pi_\mathfrak{m}^*\mathbb{R}[R]$. Consequently, $\mathcal{A}=\mathfrak{F}_1^{\mathrm{poly}}\otimes_{R_0}\mathfrak{F}_2^{\mathrm{poly}}\subset \mathcal{O}(T^*M)$ is a Poisson algebra, with the inclusion being a Poisson algebra morphism. It follows that the map $\pi_1$ induced by the evaluation of generators is a Poisson morphism. Moreover, since $\mathcal{B}=R_0=\pi_\mathfrak{m}^*\mathbb{R}[R]$ carries the trivial Poisson bracket, the structure map $\pi_2:Z_P\to\mathrm{Spec} \, R_0\cong\mathbb{A}^1$ is Poisson as well. We now check the dimension of the spectrum of these algebras. It follows that \begin{align*}
     \dim T^*(\mathrm{SU}(3)/A) = 8 = \dim \mathcal{A} +  \dim R_0 .
 \end{align*}   

To conclude that the chain \eqref{eq:superchain2} is superintegrable, it remains to show that $d\pi_1$ and $d\pi_2$ are surjective. By the left-equivariance of $P$, we may compute $d\pi_1$ at $g=e$. Take $\xi=X-\varepsilon W$. Using the differential of the moment map \eqref{eq:diffmoment}, for $(u,v)\in \mathfrak{m} \oplus \mathfrak{m}$, we have
\begin{align*}
dP_{(e,X)}(u,v)=[u,\xi]+v.
\end{align*} Hence, for $1\leq i\leq 7$, \begin{align*}
dP_i(u,v)=B( [u,\xi]+v,\lambda_i) =\underbrace{B( v,\lambda_i)}_{\text{fiber}}+ \underbrace{B( [u,X],\lambda_i)}_{\text{base via }X}- \varepsilon\,\underbrace{B( [u,W],\lambda_i)}_{\text{base via }W}.
\end{align*} Choose algebraic coordinates on $Z_P$ given by $(P_1,\ldots,P_7)$. The corresponding coordinates $(x_4,x_5,x_6,x_7; y_4,y_5,y_6,y_7)$ on $T^*_[e]M$ are given by writing $X=-\sum_{j= 4}^7x_j\lambda_j$ and $u=\sum_{j= 4}^7y_j\lambda_j$, such that $P_j\vert_{g=e}=-x_j$ for all $j$. The reductive commutator relations $[\mathfrak{m},\mathfrak{m}]\subset\mathfrak{a}$ and $[\mathfrak{m},\mathfrak{a}]\subset\mathfrak{m}$ imply that the Jacobian of $\pi_1$ at $(e,X)$ has the block form \begin{align*}
J_{\pi_1}(e,X)= \begin{pmatrix}
 \textbf{0}_{3\times 4} & A(X)_{3\times 4}\\[2pt]
-\mathbf{I}_{4\times 4} & -\varepsilon\,B_{4\times 4}
\end{pmatrix}, 
\end{align*} where $A(X)$ is the $3\times 4$ matrix of the linear map $u\mapsto\big(B( [u,X],\lambda_1),B( [u,X],\lambda_2),B( [u,X],\lambda_3)\big)$. Under the GM basis, we are able to check explicitly that
\begin{align*}
A(X)=
\begin{pmatrix}
 -x_7 &  x_6 & -x_5 & x_4\\[2pt]
 -x_6 & -x_7 & x_4 &  x_5\\[2pt]
 -x_5 & x_4 &  x_7 & -x_6
\end{pmatrix},
\end{align*}
whose $3 \times 3$ column minors are
\begin{align*}
\det A_{(1,2,3)}=x_7\sum_{j = 4}^7x_j^2,\quad \det A_{(1,2,4)}=-x_6\sum_{j = 4}^7x_j^2,\quad \det A_{(1,3,4)}=x_5\sum_{j = 4}^7x_j^2,\quad \det A_{(2,3,4)}=-x_4\sum_{j = 4}^7x_j^2.
\end{align*} Here $A_{(i,j,k)}$ means the matrix formed by taking $i$th, $j$th and $k$th column of it. In particular, $\mathrm{rank} \, A(X)=3$ for every $X\neq 0$. Since the lower-left block of $J_{\pi_1}(e,X)$ is $-\mathbf{I}_{4\times 4}$, elementary row/column operations give
\begin{align*}
\mathrm{rank} \,J_{\pi_1}(e,X)= 4+\mathrm{rank} \,A(X)=7\qquad (X\neq 0).
\end{align*} By left-equivariance, this rank is the same at $(g,X)$ for all $g\in G$. Hence, $d\pi_1$ is surjective on the Zariski open dense subset $\{(g,X)\in T^*M:X\neq 0\}$.

We now focus on the Jacobian of $\pi_2$. Identify $\mathrm{Spec}\, \mathcal{B}\cong\mathbb{A}^1$ with coordinate $R$. By Proposition \ref{prop:F2bracket}, $R = \sum_{j = 4}^7P_j^2$ in $Z_P$, from which $dR=\sum_{j = 4}^7 2P_j\, dP_j$. Therefore, in relation to $(P_1,\ldots,P_7)$, the Jacobian matrix for $\pi_2$ is given by \begin{align*}
J_{\pi_2}=\bigl( 0,0,0,\,2P_4,2P_5,2P_6,2P_7 \bigr),
\end{align*} and $\mathrm{rank} \,J_{\pi_2}=1$, where $(P_4,\ldots,P_7)\neq 0$. At $g=e$, this locus is $\{X\neq 0\}$ since $P_j=-x_j$ for $j= 4,\ldots,7$. By $G$-equivariance, we can translate to all $g\in G$. Thus, $d\pi_2$ is surjective on the same Zariski open dense subset, and as observed above, $\pi_2$ is Poisson because $\mathcal{B}=\mathbb{R}[R]$ carries the trivial Poisson bracket. At the exceptional locus $X=0$, we have $A(0)=0$. Therefore, $\mathrm{rank}\, J_{\pi_1}= 4$ and $\mathrm{rank}\, J_{\pi_2}=0$.
\end{proof}
 
\section{Conclusion}
\label{sec:conclusion}
In this work, we have presented an approach for constructing (super)integrable geodesic flows on a class of symmetric symplectic manifolds using algebraic and geometric methods. In particular, using invariant polynomial functions in Lie algebras and their Poisson-commutative properties, we obtained explicit first integrals for magnetic geodesic flows in $T^*(G/A)$, where $G/A$ is an adjoint orbit. An important observation from our work is the algebro-geometric construction of the integrals of motion, which provides a unifying perspective on various classical approaches (such as the argument shift method and the Gelfand-Cetlin type constructions) within a single Poisson projection chain. The example of the full flag manifold $\mathrm{SU}(3)/T$ demonstrates that our method can produce the set of integrals required for superintegrability in nontrivial cases, including situations with irregular elements in the underlying Lie algebra. All integrals of motion in this example are given by explicit polynomial functions in the phase space, illustrating the effectiveness of our constructive approach.

Note that the approaches developed in Section \ref{sec:superconstruction} can be extended to other symmetric or homogeneous spaces beyond the specific case of a semisimple, simply connected, and compact Lie group. Concretely, let $G$ be a real reductive Lie group and let $A\subset G$ be a closed reductive subgroup such that there is an $A$-invariant splitting $\mathfrak{g}=\mathfrak{a}\oplus\mathfrak{m}$. In $M=G/A$, equip $T^*M$ with the magnetic form $\omega_\varepsilon=\omega_{\mathrm{can}}+\varepsilon\,\pi^*\omega_{\mathrm{KKS}}$, with $W \in \mathfrak{g}$ fixed by $\mathrm{Ad}(A)$. Using the same maps as in Section \ref{sec:superconstruction}, \begin{align*}
   P(g,X)=\mathrm{Ad}(g)(X-\varepsilon W),\qquad \pi_\mathfrak{m}(g,X)=X-\varepsilon W, 
\end{align*} we again obtain the two polynomial blocks as follows
\begin{align*}
\mathfrak{F}_1^{\mathrm{poly}} := P^*S(\mathfrak{g}),\qquad  \mathfrak{F}_2^{\mathrm{poly}}:= \widetilde{\pi}_\mathfrak{m}^*S(\mathfrak{m}-\varepsilon W)^A,\qquad
R_0:=\mathfrak{F}_1^{\mathrm{poly}}\cap \mathfrak{F}_2^{\mathrm{poly}}.
\end{align*} Their mixed Poisson brackets vanish, and in the regular stratum $U\subset T^*M$ (where $\xi=X-\varepsilon W$ is regular), the same field-intersection argument shows  $\mathrm{Frac}(\mathfrak{F}_1^{\mathrm{poly}})\cap \mathrm{Frac}(\mathfrak{F}_2^{\mathrm{poly}})=\mathrm{Frac}(R_0)$. Hence, the fiber tensor product is injected into the composition of the two fraction fields and is reduced,
\begin{align*}
\widehat{\mathcal{A}}:=\mathfrak{F}_1^{\mathrm{poly}}\otimes_{R_0}\mathfrak{F}_2^{\mathrm{poly}}\;\hookrightarrow\; \mathrm{Frac}(\mathfrak{F}_1^{\mathrm{poly}})\,\mathrm{Frac}(\mathfrak{F}_2^{\mathrm{poly}}),\qquad  \mathrm{trdeg} \, \widehat{\mathcal{A}}=\mathrm{trdeg} \, \mathfrak{F}_1^{\mathrm{poly}}+\mathrm{trdeg}\,\mathfrak{F}_2^{\mathrm{poly}}-\mathrm{trdeg} \, R_0.
\end{align*} In particular, the Poisson projection chain persists under such an extension:
\begin{align*}
T^*M \xrightarrow{\;\pi_1\;} \mathrm{Spec}\, \mathcal{A} \xrightarrow{\;\pi_2\;} \mathrm{Spec}R_0,\qquad  \mathcal{A}:=\mathfrak{F}_1^{\mathrm{poly}}\otimes_{R_0}\mathfrak{F}_2^{\mathrm{poly}}.
\end{align*}
 Another direction for future work is to examine quantum analogues of these superintegrable systems, investigating how the classical polynomial integrals correspond to commuting quantum operators (sometimes referred to as missing label operators in representation theory). These connections could enhance our comprehension of how symmetries, integrability, and algebraic geometry interplay in theoretical physics. 
 
 On the other hand, in a more geometric context, we will systematically analyse, for each superintegrable system built from our Poisson projection chains, the symplectic foliation and the pullback behaviour of generic fibers along the canonical maps. Concretely, we will study the chain of Poisson morphisms $T^*M \xrightarrow{\ \pi_1\ } \mathrm{Spec}\,\mathcal{A}\xrightarrow{\ \pi_2\ } \mathrm{Spec}\,R_0$, describe the symplectic leaves in $\mathrm{Spec}\,\mathcal{A}$, and track how generic $\pi_2$-fibers pull back under $\pi_1$ to produce isotropic fibers in $T^*M$. This viewpoint clarifies 
 how leaf dimensions and ranks vary between different choices of $\mathcal{A}$ and along chains of subgroups also provides further geometric structures in the phase space via stratifications and reductions \cite{sjamaar1991stratified}.


\section*{Acknowledgement}

Kai Jiang was supported by the Fundamental Research Funds for the Central Universities. Guorui Ma was supported by the Research Funding at SIMIS. Ian Marquette was supported by the Australian Research Council Future Fellowship FT180100099. Yao-Zhong Zhang was supported by the Australian Research Council Discovery Project DP190101529.   The authors would like to thank Nicolai Reshetikhin, and Zhuo Chen for useful discussions on the geometric construction of the superintegrabilities, and the construction of the homogeneous spaces with He Lyu. 

\section*{Conflict of Interest}
The author declares that there is no conflict of interest.

\section*{Data Availability}
  No data associated in the manuscript.

\appendix
\section{}
\label{appA}

We now establish the Poisson structures on the reduced tensor product $\mathfrak{F}_1^{\mathrm{poly}} \otimes_{R_0} \mathfrak{F}_2^{\mathrm{poly}}$. In other words, we show that $(\widehat{\mathcal{A}},\{\cdot,\cdot\}_{\widehat{\mathcal{A}}})$ is a Poisson algebra. We first consider the tensor product $\mathfrak{F}_1^{\mathrm{poly}} \otimes \mathfrak{F}_2^{\mathrm{poly}}$. Define a bilinear operation  $ \{\cdot,\cdot\}_{\otimes}:  \mathfrak{F}_1^{\mathrm{poly}} \otimes \mathfrak{F}_2^{\mathrm{poly}} \times \mathfrak{F}_1^{\mathrm{poly}} \otimes \mathfrak{F}_2^{\mathrm{poly}} \rightarrow \mathfrak{F}_1^{\mathrm{poly}} \otimes \mathfrak{F}_2^{\mathrm{poly}}$   by \begin{align*}
    \{ h \otimes \theta,h' \otimes \theta'\}_\otimes : = \{h,h'\}_1 \otimes (\theta \theta' ) +  hh' \otimes \{\theta,\theta'\}_2. 
\end{align*} Here $\{\cdot,\cdot\}_i$ is the twisted Poisson bracket $\{\cdot,\cdot\}_\varepsilon$ restricted on $\mathfrak{F}_i^{\mathrm{poly}}$ for all $i =1,2$. 
It is straightforward to show that $\left(\mathfrak{F}_1^{\mathrm{poly}} \otimes \mathfrak{F}_2^{\mathrm{poly}}, \{\cdot,\cdot\}_\otimes\right)$ is a Poisson algebra by checking the Leibniz rule and Jacobi identity. Let $J$ be the two-sided ideal generated by the following relations \begin{align}
    (hr) \otimes \theta - h \otimes (r \theta) , \quad  h \otimes (\theta r) - (hr )\otimes \theta,
\end{align} where $h \in \mathfrak{F}_1^{\mathrm{poly}},\theta \in \mathfrak{F}_2^{\mathrm{poly}} $ and $r \in R_0$.  The fiber tensor product is\begin{align*}
\widehat{\mathcal{A}}:=\mathfrak{F}_1^{\mathrm{poly}}\otimes_{R_0}\mathfrak{F}_2^{\mathrm{poly}}:=(\mathfrak{F}_1^{\mathrm{poly}}\otimes \mathfrak{F}_2^{\mathrm{poly}})/J.
\end{align*} 
\begin{lemma} \label{lem:factor-central}
For every $r\in R_0$, we have $\{r,f\}_1=0$ for all $f\in \mathfrak{F}_1^{\mathrm{poly}}$ and $\{r,\varphi\}_2=0$ for all $\varphi\in \mathfrak{F}_2^{\mathrm{poly}}$. 
\end{lemma}

\begin{proof}
By Proposition \ref{prop:inter} and Corollary \ref{lem:domainstorsionfree}, we can write $r=P^*h=\pi_\mathfrak{m}^*\theta$ with $h\in S(\mathfrak{g})^G$, $\theta\in S(\mathfrak{m})^A$. For any $f=P^*h' \in \mathfrak{F}_1^{\mathrm{poly}}$, we then deduce \begin{align*}
    \{r,f\}_1=\{\pi_\mathfrak{m}^*\theta,\;P^*h'\}_\varepsilon=0
\end{align*} by the vanishing of the mixed-block $\{\mathfrak{F}_2^{\mathrm{poly}},\mathfrak{F}_1^{\mathrm{poly}}\}_\varepsilon=0$ stated in Lemma \ref{lem:interzero}. Similarly, we deduce $\{r,\varphi\}_2=0$.
\end{proof}

We now show that the Poisson bracket $\{\cdot,\cdot\}_{\otimes}$ descends into the bilinear operator $\{\cdot,\cdot\}_{\widehat{\mathcal{A}}}$ such that $\left(\widehat{\mathcal{A}},\{\cdot,\cdot\}_{\widehat{\mathcal{A}}}\right)$ forms a Poisson algebra.

\begin{lemma} 
The ideal $J$ is a Poisson ideal of $\left(\mathfrak{F}_1^{\mathrm{poly}}\otimes \mathfrak{F}_2^{\mathrm{poly}},\{\cdot,\cdot\}_\otimes\right)$. Hence $\{\cdot,\cdot\}_\otimes$ descends uniquely to a bilinear bracket
\begin{align*}
\{\cdot,\cdot\}_{\widehat{\mathcal{A}}}:\ \widehat{\mathcal{A}}\times \widehat{\mathcal{A}}\longrightarrow \widehat{\mathcal{A}}.
\end{align*}
\end{lemma}

\begin{proof}
Let $x=(fr)\otimes\varphi - f\otimes(r\varphi)$ and $y=f_1\otimes\varphi_1$.  Using bilinearity, Leibniz, and the factor-wise centrality of $r$ proved above, \begin{align}
\nonumber
\{(fr) \otimes \phi,\ y\}_\otimes &=\{fr,f'\}_1 \otimes \phi\phi'+(fr f') \otimes \{\phi,\phi'\}_2\\
&=r\,\{f,f'\}_1 \otimes \phi\phi'+ff'r \otimes \{\phi,\phi'\}_2, \label{eq:xa}\\
\{f \otimes (r\phi),\ y\}_\otimes
&=\{f,f'\}_1 \otimes r\phi\phi'+ff' \otimes r\,\{\phi,\phi'\}_2. \label{eq:xb}
\end{align}  Subtracting gives the sum of generators of $J$. Then using \eqref{eq:xa} and \eqref{eq:xb}, we have \begin{align*}
    \{x,y\}_\otimes = & \, \{r,f_1\}_1\otimes\varphi\varphi_1 - ff_1\otimes \varphi\{r,\varphi_1\}_2+ \big(r\{f,f_1\}_1\otimes\varphi\varphi_1 - \{f,f_1\}_1\otimes r\varphi\varphi_1\big) \\
& \, +  \big(fr f_1\otimes\{\varphi,\varphi_1\}_2 - ff_1\otimes r\{\varphi,\varphi_1\}_2\big).
\end{align*} By Lemma \ref{lem:factor-central}, the first two terms vanish, and the remaining two lie in $J$. Hence, $J$ is a Poisson ideal and the bracket descends.  
\end{proof}

We then examine the Poisson bracket on $\widehat{\mathcal{A}}$ and block rules.

\begin{proposition} 
The descended operation $\{\cdot,\cdot\}_{\widehat{\mathcal{A}}}$ makes $\widehat{\mathcal{A}}$ a Poisson algebra. It is characterised by
\begin{align}
\nonumber
&\{f\otimes_{R_0} 1,\ f'\otimes_{R_0} 1\}_{\widehat{\mathcal{A}}} =\{f,f'\}_1\otimes_{R_0} 1,\\
&\{1\otimes_{R_0} \phi,\ 1\otimes_{R_0}\phi'\}_{\widehat{\mathcal{A}}} =1\otimes_{R_0}\{\phi,\phi'\}_2, \label{eq:commutatorAre}\\
\nonumber
&\{f\otimes_{R_0} 1,\ 1\otimes_{R_0}\phi\}_{\widehat{\mathcal{A}}}=0,
\end{align}
and the Leibniz rule.  Then $\left(\widehat{\mathcal{A}},\{\cdot ,\cdot \}_{\widehat{\mathcal{A}}}\right)$ is a Poisson algebra.
\end{proposition}

\begin{proof}
Bilinearity, antisymmetry, and the Leibniz rule descend from the tensor product as $J$ is a two-sided Poisson ideal. The Jacobi identity holds by direct computational check. The relations in \eqref{eq:commutatorAre} follow directly from the definition of $\{\cdot,\cdot\}_\otimes$.
\end{proof}
 
The constructions above inspire the following definition.

\begin{definition}
\label{def:algetensorp}
  Consider the algebraic tensor product $$
\widehat{\mathcal{A}}:=\mathfrak{F}_1^{\mathrm{poly}}\otimes_{R_0}\mathfrak{F}_2^{\mathrm{poly}} = \mathfrak{F}_1^{\mathrm{poly}}\otimes \mathfrak{F}_2^{\mathrm{poly}}\big/J.
$$ The Poisson bracket $\{\cdot ,\cdot \}_{\widehat{\mathcal{A}}}:\widehat{\mathcal{A}}\times\widehat{\mathcal{A}}\to\widehat{\mathcal{A}}$ on $\widehat{\mathcal{A}}$ is given by \begin{align*}
    \{f\otimes\phi,\ f'\otimes\phi'\}_{\widehat{\mathcal{A}}} :=\{f,f'\}_\varepsilon\otimes(\phi\phi')+(ff')\otimes\{\phi,\phi'\}_\varepsilon, \quad f,f'\in \mathfrak{F}_1^{\mathrm{poly}},\ \phi,\phi'\in \mathfrak{F}_2^{\mathrm{poly}}.
\end{align*}   
\end{definition}

\section{ } 
\label{appB}
 We now review some basic terminology from the theory of commutative algebras as a generalization of the Poisson polynomial algebras studied above. Let $R$ be a finitely generated integral domain over $\mathbb{R}$. Consider two other finitely generated integral domains $A$ and $B$ that contain $R$ as a subring. We have inclusions of $k$-algebra homomorphisms $R \hookrightarrow A$ and $R \hookrightarrow B$. We are interested in the properties of the tensor product of these algebras over $R  = A \cap B$: \begin{align*}
    A \otimes_R B\,.
\end{align*}  We will focus on the algebraic structures on this $R$-algebra $A \otimes_R B$. 

 \begin{remark}
Suppose that $A=\textbf{Alg}\langle a\rangle$ and $B=\textbf{Alg}\langle b\rangle$ are finite sets $a\subset A$ and $b\subset B$, then \begin{align*}
   \textbf{Alg}\langle A\cup B\rangle=\textbf{Alg}\langle a\cup b\rangle . 
\end{align*} 
Indeed, $\textbf{Alg}\langle a\cup b\rangle$ is a subalgebra containing $A$ and $B$. Hence, it contains $\textbf{Alg}\langle A\cup B\rangle$. Conversely, $\textbf{Alg}\langle A\cup B\rangle$ contains $a\cup b$ and therefore also contains $\textbf{Alg}\langle a\cup b\rangle$.

On the other hand, the intersection $A\cap B$ is a set-theoretic intersection within the fixed algebra $F$ and is not determined by the intersections of chosen generating sets. In general, we only have $\textbf{Alg}\langle a\cap b\rangle\subseteq A\cap B$, and equality may not hold.
Moreover, the notation $A\cap B$ is meaningful only after fixing the embeddings of $A$ and $B$ into a common ambient algebra. See Lemma \ref{lem:interfiber}. For abstract $R$-algebras $A$ and $B$, the expression $A\cap B$ is not defined.
 \end{remark}  
 
\begin{definition} 
\label{def:algedis} \cite[Chapter VIII, Section 3]{LangAlgebra}
Let $A,B$ be finitely-generated integral domains containing a common integral subdomain $R$ with fraction field $K_R$. 
Let  \begin{align*}
    \{X_1,\ldots,X_n\}\subset\mathrm{Frac}(A),\quad \{Y_1,\ldots,Y_m\}\subset\mathrm{Frac}(B)
\end{align*} be the transcendence bases over $K_R$. 
We say that $A$ and $B$ are \textit{algebraically disjoint} (\textit{linearly disjoint}) over $R$ if and only if the canonical map \begin{align*}
\iota:   \mathrm{Frac}(A)\otimes_{K_R} \mathrm{Frac}(B) \hookrightarrow \mathrm{Frac}(A)\mathrm{Frac}(B) = \mathrm{Frac}(AB)
\end{align*} is injective, where $\mathrm{Frac}(A)\otimes_{K_R} \mathrm{Frac}(B)$ forms a domain and $ \mathrm{Frac}(A) \mathrm{Frac}(B)$ is the smallest subfield containing both $\mathrm{Frac}(A)$ and $\mathrm{Frac}(B)$. Therefore,  \begin{align*}
\mathrm{Frac}(\mathrm{Frac}(A)\otimes_{K_R} \mathrm{Frac}(B))=\mathrm{Frac}(A)\mathrm{Frac}(B),
\end{align*} and the transcendence bases $\{X_i\}$ for $\mathrm{Frac}(A)/K_R$ and $\{Y_j\}$ for $\mathrm{Frac}(B)/K_R$ are algebraically independent within $\mathrm{Frac}(A)\mathrm{Frac}(B)$.
\end{definition}

\begin{definition}
\label{def:redu} \cite[Chapter 1]{AtiyahMacdonald}
    For any commutative algebra $S$, the nilradical of $S$ is defined as the set of nilpotent elements \begin{align*}
   \sqrt{0}_S := \{\,f \in S : f^n = 0 \text{ for some integer } n \ge 1\,\}\,. 
\end{align*} We say $S$ is \textit{reduced} if it does not have nonzero nilpotent elements. Equivalently, $S$ is reduced if and only if $\sqrt{0}_S = \{0\}$. In a reduced algebra, the equality $f^n=0$ always implies $f=0$. 
\end{definition} 
  
We aim to show that, under appropriate hypotheses, the tensor product of two integral domains does not introduce any new algebraic polynomial relations between elements of $A$ and $B$ beyond those already coming from $R$. 
Equivalently, we will see that $A \otimes_R B$ is a reduced $R$-algebra whose transcendence degree is exactly the sum of the transcendence degrees of $A$ and $B$ minus that of $R$.

We now introduce a proposition that extends the integral domain $\mathcal{A}$ into its rational field, such that we can compute the transcendence degree at the field level.


\begin{proposition}
\label{prop:fieldtoring}
      Let $R \subseteq A$ and $R \subseteq B$ be inclusions of finitely generated integral domains over $ R$. Let $K_R = \mathrm{Frac}(R),K_A = \mathrm{Frac}(A) $ and $K_B = \mathrm{Frac}(B) $. Suppose that $A$ and $B$ are algebraically disjoint over $R$. Let $S : = R\setminus\{0\}$ be the localization of $R$. Then $S^{-1}(A \otimes_R B)$ is a domain, and \begin{align*}
          \mathrm{trdeg}_{K_R} \mathrm{Frac} \, S^{-1}(A \otimes_R B) = \mathrm{trdeg}_{K_R} K_A +  \mathrm{trdeg}_{K_R} K_B ,
      \end{align*}   the algebra of the canonical tensor product $   \mathcal{A}   :=  A\otimes_R B $ does not introduce new algebraic relations. 
      In particular, $\mathcal{A}$ has a transcendence degree  \begin{align*}
   \mathrm{trdeg}_R \, \mathcal{A} = \mathrm{trdeg}_R \, A +  \mathrm{trdeg}_R \, B - \mathrm{trdeg}_R R\,,  
 \end{align*}  and $\mathcal{A}$ is a reduced $R$-algebra.
\end{proposition}

 \begin{proof}
     Since $A$ and $B$ are integral domains containing $R$, we consider their fraction fields $K_A = \mathrm{Frac}(A)$, $K_B = \mathrm{Frac}(B)$, and $K_R = \mathrm{Frac}(R)$. In other words, we have the field extensions $K_R \subseteq K_A$ and $K_R \subseteq K_B$. Under the assumption that $A$ and $B$ are linearly disadjoint over $R$, by Definition \ref{def:algedis}, if ${X_1,\ldots,X_n}$ is a transcendence basis for $K_A$ over $K_R$ and ${Y_1,\ldots,Y_m}$ is a transcendence basis for $K_B$ over $K_R$, then \begin{align}
         K_A \otimes_{K_R} K_B \hookrightarrow K_A K_B \,,\label{eq:fraciso}
     \end{align} is injective, where $K_AK_B$ is a transcendence extension of $K_R(X_1,\ldots,X_n,\, Y_1,\ldots,Y_n)$. Here $n = \mathrm{trdeg} \, A $ and $m = \mathrm{trdeg} \, B$. Then there is a natural $R$-algebra homomorphism $\mathcal{A} \to K_A\otimes_{K_R}K_B$, which factors through the localization at $S =R\setminus\{0\}$: \begin{align}
     \nonumber
S^{-1} \left(\mathcal{A}\right) &\cong (A\otimes_R K_R)\otimes_{K_R}(B\otimes_R K_R) \\
&\cong K_A\otimes_{K_R}K_B \xhookrightarrow{ \ \varrho \ } \text{(a field algebraic over $K_R(X_1,\ldots,X_n,Y_1,\ldots,Y_m)$).} \label{eq:field}
\end{align} Here $S^{-1}(A) = A \otimes_R K_R $ and $S^{-1}(B) = B \otimes_R K_R $. In the composition chain \eqref{eq:field}, the isomorphisms are given by the localization and associativity of the tensor. The inclusions  $(A\otimes_R K)\hookrightarrow L$ and $(B\otimes_R K)\hookrightarrow M$ are $K$-linear injections, and $L$ and $M$ are linearly  disjoint over $K$ by assumption, so the third arrow $\varrho$ is injective. Hence $S^{-1}(\mathcal{A})$ embeds into a field and, in particular, is an integral domain.  Therefore, through the embedding $\varrho$, the elements of $\mathcal{A}$ can be regarded as rational functions in the independent transcendental $X_i$ and $Y_j$ over $K_R$. In particular, any potential algebraic relation among elements of $\mathcal{A}$ would produce a polynomial equation in $X_i$ and $Y_j$ of over $K_R$ that vanishes in $K_R(X_i, Y_j)$.

Suppose, by contradiction, that there exist new algebraic relations in $\mathcal{A}$. Without loss of generality, assume that there is a nonzero polynomial \begin{align*}
    f(t_{1,1}, \ldots, t_{1,n}, \, t_{2,1}, \ldots, t_{2,m}) \in R[x_{1,1},\ldots,x_{1,n}, \, x_{2,1},\ldots,x_{2,m}],
\end{align*}  where the indeterminates $x_{1,i}$ and $x_{2,j}$ are such that if we substitute $x_{1,i} = X_i \in K_A$  and $x_{2,j} = Y_j \in K_B$, the polynomial $f$ vanishes in $\mathcal{A}$. 
In other words, we assume $ f(X_1,\ldots,X_n,\, Y_1,\ldots,Y_m) = 0  $ to be an identity that holds in the ring $\mathcal{A}$. By applying the embedding of $\mathcal{A}$ into the field $K_R(X_1,\ldots,X_n,Y_1,\ldots,Y_m)$, the same polynomial relation also holds in the field $K_R(X_i, Y_j)$. However, since $X_1,\ldots,X_n,Y_1,\ldots,Y_m$ are algebraically independent over $K_R$, a nonzero polynomial cannot vanish on them identically in the field of rational functions. The only polynomial in $R[x_{1,i}, x_{2,j}]$ that vanishes for such a substitution into $K_R(X_i,Y_j)$ is the zero polynomial. Hence, no new algebraic polynomial relations can appear in $\mathcal{A}$ beyond those that are already present in $R$. 


As a consequence, we can determine the transcendence degree on the generic fiber $\mathcal{A}$ over $R$. The natural maps $A\to S^{-1}(\mathcal{A})$, $a\mapsto a\otimes 1$, $B\to S^{-1}(\mathcal{A})$, and $b\mapsto 1\otimes b$ induce inclusions
\begin{align*}
K_A,\;K_B \subseteq \mathrm{Frac} \left(S^{-1}(\mathcal{A})\right).
\end{align*}
Hence, $\mathrm{Frac} \left(S^{-1}(\mathcal{A})\right)$ is the composite of $K_A$ and $K_B$ over $K_R$, which is algebraic over $K_R(X_1,\ldots,X_n,Y_1,\ldots,Y_m)$. Since ${X_1,\ldots,X_n,Y_1,\ldots,Y_m}$ forms a transcendence basis for the fraction field of $\mathcal{A}$ over $K_R$, we have \begin{align*}
    \mathrm{trdeg}_{K_R} \mathrm{Frac}(\mathcal{A}) = n+m\,,
\end{align*} where $n = \mathrm{trdeg}_{K_R} \mathrm{Frac}(A)$ and $m = \mathrm{trdeg}_{K_R} \mathrm{Frac}(B)$. But $\mathrm{trdeg}_{K_R} \mathrm{Frac}(A) = \mathrm{trdeg}_R A$ as $K_R$ is the fraction field of $R$. The argument holds for the transcendence degree of $B$. Also note $\mathrm{trdeg}_R R = \mathrm{trdeg}_{K_R} K_R = 0$ by definition. Therefore, \begin{align*}
    \mathrm{trdeg}_R \,\mathcal{A}  = \mathrm{trdeg}_R \, A +  \mathrm{trdeg}_R \, B \,,
\end{align*} since here $\mathrm{trdeg}_R R = 0$. And more generally, if $R$ already has some transcendence degree over $\mathbb{R}$, that would be subtracted as stated. This shows that the transcendence degree of $\mathcal{A} $ is exactly the sum of those of $A$ and $B$ over $R$, indicating no reduction in freedom (i.e., no new algebraic constraints).

Finally, we deduce the nilradical of $\mathcal{A}$. By embedding \eqref{eq:field}, $S^{-1}(\mathcal{A})\hookrightarrow K_A K_B$. A field contains no nonzero nilpotent elements as the only solution to $\alpha^N=0$ in a field is $\alpha=0$. Therefore, if an element $z \in \mathcal{A}$ is nilpotent (i.e., there exists a $N$ such that $z^N=0$), its image in $K_AK_B$ would also satisfy equation $(	\text{image of } z)^N = 0$ in a field. This forces the image of $z$ to be $0$ in the field and hence $z=0$ in $\mathcal{A}$. This shows that $\mathcal{A}$ does not have nonzero nilpotents. In other words, the nilradical $\sqrt{0}$ of $\mathcal{A}$ is zero, so $\mathcal{A}$, by Definition \ref{def:redu}, is reduced.  

In conclusion, $\mathcal{A}$ introduces no new algebraic relations and is reduced, as claimed.
 \end{proof}
 \begin{remark}
     \label{re:36}
 Let $R\subset A,B$ and $S$ be the same as defined above. Using the canonical chain defined in \eqref{eq:field}, we have the composition of injectives from the integral domain $S^{-1}(\mathcal{A})  \hookrightarrow K_A\otimes_{K_R} K_B    \xrightarrow{ \ \varrho \  } K_A K_B $. Therefore, any polynomial relation that holds in $\mathcal{A}$ must also hold in $K_A\otimes_{K_R} K_B$. Thus, checking for \textit{new relations} reduces to the generic fiber over $K_R$.  If $\varrho$ is injective, then $S^{-1}(\mathcal{A})$ embeds into the domain $K_A K_B$, so the localization is reduced. 
 \end{remark}

 \begin{corollary} \label{cor:norela}
Let $R\subset A,B\subset C$ be $R$-subalgebras which are finitely generated integral domains. Let \begin{align*}
    \mu:A\otimes_R B\longrightarrow C,\quad a\otimes b\longmapsto ab
\end{align*} be the multiplication map. Then  $\mathrm{Im}(\mu)=\mathrm{\textbf{Alg}} \langle A\cup B\rangle$. Moreover, if the induced map \begin{align*}
  K_A\otimes_K K_B\longrightarrow LM 
\end{align*}  is injective, then $  S^{-1}(\ker\mu)=\{0\}$. 
In particular, if $A\otimes_R B$ is $R$-torsion free, then $\ker\mu=\{0\}$,  and hence \begin{align*}
    \textbf{Alg} \langle A\cup B\rangle \cong A\otimes_R B .
\end{align*} \end{corollary}

  As discussed above, we now apply the construction to our polynomial Poisson algebras $\mathfrak{F}_i^{\mathrm{poly}}$ for all $i = 1,2$. Based on Proposition \ref{prop:fieldtoring}, after the field extension, we need to show that 
$\mathfrak{F}_1^{\mathrm{poly}}$ and $\mathfrak{F}_2^{\mathrm{poly}}$ are algebraically disjoint over $R_0$. 
Let \begin{align*}
L:=\mathrm{Frac}\,\mathfrak{F}_1^{\mathrm{poly}},\qquad M:=\mathrm{Frac}\,\mathfrak{F}_2^{\mathrm{poly}}\qquad K:= \mathrm{Frac}\,R_0.
\end{align*} If $L$ and $M$ are algebraically disjoint over $K$, then the natural map $L\otimes_K M \longrightarrow LM$ is injective. By \cite[Ch. VIII, \S3-\S4]{LangAlgebra}, we need to show that the transcendence basis of $M$ and $L$ are free over $K$, the intersection is clear and the extension $L/K$ is regular. Consequently, $\ker \mu$ remains zero,  so $\delta_{\mathrm{alg}}=0$, and hence\begin{align*}
\mathrm{rank}_{R_0}\widehat{\mathcal{A}}=\mathrm{rank}_{R_0} \left(\mathfrak{F}_1^{\mathrm{poly}}\otimes_{R_0}\mathfrak{F}_2^{\mathrm{poly}}\right).
\end{align*} Starting with the following definition:  \begin{definition}
 The Poisson polynomial algebras $\mathfrak{F}_1^{\mathrm{poly}},\mathfrak{F}_2^{\mathrm{poly}}$ and $R_0$ are integral domains. Then we define $\widehat{\mathcal{A}} = \mathfrak{F}_1^{\mathrm{poly}}\otimes_{R_0}\mathfrak{F}_2^{\mathrm{poly}}$ and $\mathcal{A} =  \textbf{Alg} \left\langle F_1^{\mathrm{poly}} \cup F_2^{\mathrm{poly}} \right \rangle= \sqrt{\widehat{\mathcal{A}}}$. The rank of $\mathcal{A}$ is given by \begin{align*}
\mathrm{rank}\, \mathcal{A} = \mathrm{rank}\,  \widehat{\mathcal{A}}  - \delta_{\mathrm{alg}},
 \end{align*} where $\delta_{\mathrm{alg}} \in \mathbb{R}^+$.
\end{definition}

\begin{remark}
\label{re:35}
 (i)  Since $\mathcal{O}(T^*M)$ is reduced, every subalgebra of $\mathcal{O}(T^*M)$ is reduced. In particular,  $ 
\mathcal{A}=\mathrm{im}\,\mu  $ is reduced. Thus, $\mathcal{A}$ is already reduced, and no further reduction is required. Hence, we have $\mathcal{A} = \sqrt{\widehat{\mathcal{A}}}$.

(ii) For any finite $R$-algebra $B$, we have $\mathrm{rank}_RB=\mathrm{rank}_R\sqrt{B}$. Indeed, a localization at the generic point eliminates nilpotent components in $B$: \begin{align*}
K\otimes_RB  \cong K\otimes_R\bigl(B/\mathrm{nil}(B)\bigr)  \cong K\otimes_R\sqrt{B},
\end{align*} where $\mathrm{nil}(B)$ is the nilpotent component in $B$. Hence, $\mathrm{rank}_R\widehat{\mathcal{A}}=\mathrm{rank}_R\sqrt{\widehat{\mathcal{A}}}$.

(iii) By Definition \ref{def:algetensorp}, $\widehat{\mathcal{A}}=\bigl(\mathfrak{F}_1^{\mathrm{poly}}\otimes_{R_0}\mathfrak{F}_2^{\mathrm{poly}}\bigr) /\ker\mu$, so $\ker\mu=0$ gives an isomorphism $\widehat{\mathcal{A}}\ \xrightarrow{\cong }\ \mathcal{A}$. Hence $\delta_{\mathrm{alg}} = \dim \ker \mu$.

\end{remark}

We now show $\mathcal{A} \cong \widehat{\mathcal{A}}$ by checking the condition imposed by Proposition \ref{prop:fieldtoring} one by one. As we mentioned before, we show the algebraically disjoint property by checking the regular extension and clear intersection. 

\begin{lemma} \label{lem:invalgclose}
Let $G$ be a simply connected, semisimple, and compact Lie group, and let $G$ act on $ Z_P = \mathrm{Ad}(G)(\mathfrak{m} - \varepsilon W)$. Define $L  = \mathbb{R}( Z_P) = \mathrm{Frac}\left(\mathfrak{F}_1^{\mathrm{poly}}\right)$ and $K  = \mathbb{R}( Z_P)^G$. Then $K$ is algebraically closed in $L$. Equivalently, the extension $L/K$ is regular.
\end{lemma}

\begin{proof}
Let $\alpha\in L$ be algebraic over $K$, and let $p(t)\in K[t]$ be its monic minimal polynomial. Since $K$ is $G$-invariant, the coefficients of $p(t)$ is $G$-invariant. Hence, for any $g\in G$, the element $g\cdot \alpha$ satisfies the same polynomial $p(t)$ because the coefficients lie in $K=L^G$. That is, \begin{align*}
    p(g\cdot \alpha) = g \cdot p(\alpha) = g \cdot 0 = 0.
\end{align*} Therefore, the orbit $G\cdot\alpha = \{g \cdot \alpha: g \in G\}$ is contained in the finite set of roots of $p(t)  \in K[t]$.  Let $\mathcal{R} = \{\beta_1, \dots, \beta_m\} \subset L$ denote this set. To prove $\alpha \in K$, it suffices to show that $\alpha$ is $G$-invariant. In other words, the map $G \to \mathbb{R}$ given by $g \mapsto g \cdot \alpha$ is constant. 
Fix a point $z \in Z_P$. We define the evaluation map $\Psi_z: G \to \mathbb{R}$ by \begin{align*}
    \Psi_z(g) = (g \cdot \alpha)(z) = \alpha(g^{-1} \cdot z).
\end{align*}  Since $G$ acts smoothly on $Z_P$, and $\alpha$ is regular on the orbit of $z$, the map $\Psi_z$ is continuous. The image $\Psi_z(G)$ is contained in the finite set of real values $\{\beta_1(z), \dots, \beta_m(z)\} \subset \mathbb{R}$. 
The only connected subsets of a finite set are singletons. Thus, $\Psi_z$ is constant, implying $(g \cdot \alpha)(z) = \alpha(z)$ for all $g \in G$. Since $z$ was chosen arbitrarily from a dense open set $U$, we conclude that $g \cdot \alpha = \alpha$ are elements of the function field $L$. Therefore, $\alpha \in L^G = K$.
\end{proof}

\begin{remark}
In a more general setting, let $G$ be a connected affine algebraic group, and let $G$ act rationally on an irreducible $\mathbb{R}$-variety $X$. Define $L  = \mathbb{R}(X)$ and $K  = \mathbb{R}(X)^G$. Then $L/K$ is regular.   
\end{remark}

We now show that in our setting, $L \cap M = K$.

\begin{lemma} \label{lem:orbitconst}
Let $\phi: = a/b\in \mathrm{Frac} \, S(\mathfrak{g})$ and $\psi:= c/d\in \mathrm{Frac} \,S(\mathfrak{m})^A$ be defined with $a,b \in S(\mathfrak{g})$ and $c,d \in S(\mathfrak{m})^A$. Suppose  $\phi(\mathrm{Ad}(G)u)=\psi(u)$  holds on a non-empty Zariski open subset $$U   =\left\{(g,u)\in T^*M : b(\mathrm{Ad}(g)u)\neq 0,\ d(u)\neq 0\right\}$$ of  $T^*M$. Then $\phi$ is $\mathrm{Ad}(G)$-invariant as a rational function on $\overline{\mathrm{Ad}(G)\mathfrak{m}}\subset\mathfrak{g}$, hence $ \phi\in \mathbb{R}(\overline{\mathrm{Ad}(G)\mathfrak{m}})^G =\mathrm{Frac}\big(\mathbb{R}[\overline{\mathrm{Ad}(G)\mathfrak{m}}]^G\big) $.
\end{lemma}

\begin{proof}
To show that $\phi$ is $\mathrm{Ad}(G)$-invariant, it is sufficient to show that for a generic $y$, we have $d\phi_y([X,y]) = 0$ for any $X \in \mathfrak{g}$. Fix $(g,u) \in U$, define the one-parameter flow of $g$ by $g(t)=g \exp(tY)$ for all $t \in \mathbb{R}$. Clearly, $g = g(0)$. As $U$ is open and contains $(g,u)$, the continuity of $t \mapsto (g(t),u)$ implies that there exists a $\varepsilon > 0$ such that $(g(t),u) \in U$ for all $|t| < \varepsilon$. Hence, by the assumption, in this neighbourhood, we must have $\phi(\mathrm{Ad}(g(t))u)=\psi(u)$. Differentiating the identity $\phi(\mathrm{Ad}(g(t))u)=\psi(u)$ at $t=0$ gives
\begin{align}
0 = \left.\dfrac{d}{dt}\right\vert_{t= 0} \phi\left(\mathrm{Ad}(g\exp(tY)u)\right)  =  d\phi_{\mathrm{Ad}(g)u}\big(\mathrm{Ad}(g)[Y,u]\big) = B \left(\nabla \phi(\mathrm{Ad}(g)u),\mathrm{Ad}(g) [Y,u]\right). \label{eq:differentialra}
\end{align} With an $\mathrm{Ad}(G)$-invariant inner product $B$, the evaluation in \eqref{eq:differentialra} equals $ B([\mathrm{Ad}(g^{-1})\nabla\phi(\mathrm{Ad}(g)u),u],Y) $ for all $Y\in\mathfrak{g}$. Thus, \begin{align}
    [\mathrm{Ad}(g^{-1})\nabla\phi(\mathrm{Ad}(g)u),u]=0. \label{eq:commutatorrel}
\end{align}  Let $y=\mathrm{Ad}(g)u$ be generic. From the trivial commutator relations in \eqref{eq:commutatorrel}, we therefore get $[\nabla\phi(y),y]=0$, i.e., $\phi$ has zero directional derivatives along the tangent to the adjoint orbit through $y$. Thus, $\phi$ is constant on a Zariski open subset of each generic adjoint orbit in $\overline{\mathrm{Ad}(G)\mathfrak{m}}$, so $\phi$ is $\mathrm{Ad}(G)$-invariant.
\end{proof}

 By Proposition \ref{prop:inter}, we have $R_0=\mathfrak{F}_1^{\mathrm{poly}}\cap \mathfrak{F}_2^{\mathrm{poly}}$ inside $\mathcal{O}(T^*M)$. 
Then we always have the following inclusion:
\begin{equation}
K=\mathrm{Frac}(\mathfrak{F}_1^{\mathrm{poly}}\cap \mathfrak{F}_2^{\mathrm{poly}})\,\subseteq\, L\cap M   \label{eq:inclusion}
\end{equation} inside  $\mathrm{Frac}\,\mathcal{O}(T^*M)$. The statement of Lemma \ref{lem:LcapM} below is the reverse inclusion $L\cap M\subseteq K$, hence $L\cap M=K$.
  
\begin{lemma} \label{lem:LcapM}
With $L:=\mathrm{Frac} \, \left(\mathfrak{F}_1^{\mathrm{poly}}\right)=\mathbb{R}(Z_P)$, $M:=\mathrm{Frac} \, \left(\mathfrak{F}_2^{\mathrm{poly}}\right)=\mathbb{R}(\mathfrak{m})^A$, and $K:= \mathrm{Frac} \left( R_0\right)$, we have
\begin{align*}
L\cap M=K .
\end{align*}
\end{lemma}

\begin{proof}
By Proposition \ref{prop:inter}, inclusion $K\subset L\cap M$ is immediate. For reverse inclusion, let $f\in L\cap M$. Then there exist $\phi\in \mathrm{Frac}\, S(\mathfrak{g})$ and $\psi\in \mathrm{Frac} \,S(\mathfrak{m})^A$ such that, as rational functions on $\mathcal{O}(T^*M)$, \begin{align*}
f=\phi\circ \bar{P} = \psi \circ \pi_\mathfrak{m} \quad\text{and}\quad f((g,u))=\psi(u).
\end{align*} writes $\phi=a/b$ with $a,b\in S(\mathfrak{g})$ and $\psi=c/d$ with $c,d\in S(\mathfrak{m})^A$. In $U$, we have that $\left(\phi \circ \bar{P}\right)\vert_U = \left( \psi \circ \pi_\mathfrak{m}\right)\vert_U$. This leads to the identity \begin{equation}
a(\mathrm{Ad}(g)u)\,d(u)= c(u)\,b(\mathrm{Ad}(g)u) \label{eq:cross}
\end{equation} on $U$. Fix $u$ in a non-empty Zariski open $\mathfrak{m}^\circ\subset\mathfrak{m}$ such that $d(u)\neq 0$ and $ G_u^\circ:=\{\,g\in G: (g,u)\in U\,\} $ is non-empty Zariski open in $G$. For all $g\in G_u^\circ$, \eqref{eq:cross} gives \begin{align*}
\frac{a(\mathrm{Ad}(g)u)}{b(\mathrm{Ad}(g)u)}=\frac{c(u)}{d(u)} .
\end{align*}  As the equality $\varphi(\mathrm{Ad}(g)u) = \psi(u)$ holds on $U$, by Lemma \ref{lem:orbitconst}, $\varphi$ is $\mathrm{Ad}(G)$-invariant on $\mathrm{Ad}(G) (\mathfrak{m}-\varepsilon W)$. It follows that the rational function $\phi=a/b$ is constant on a Zariski open subset of the adjoint orbit $\mathrm{Ad}(G)u$. As $u$ varies over $\mathfrak{m}^\circ$, the union of these orbits is Zariski dense in $\overline{\mathrm{Ad}(G)\mathfrak{m}}$. Thus, $\phi$ is $\mathrm{Ad}(G)$-invariant in $\overline{\mathrm{Ad}(G)\mathfrak{m}}$ as a rational function: \begin{align*}
\phi\in \mathbb{R}(\overline{\mathrm{Ad}(G)\mathfrak{m}})^G.
\end{align*} Here, by geometric invariant theory for semisimple groups, we have \begin{align*}
\mathbb{R}(\overline{\mathrm{Ad}(G)\mathfrak{m}})^G = \mathrm{Frac}\big(\mathbb{R}[\overline{\mathrm{Ad}(G)\mathfrak{m}}]^G\big),  \text{ where }  \mathbb{R}[\overline{\mathrm{Ad}(G)\mathfrak{m}}]^G\cong S(\mathfrak{g})^G\big/\big(I(\overline{\mathrm{Ad}(G)\mathfrak{m}}) \cap S(\mathfrak{g})^G\big).
\end{align*} Therefore, there exist invariant polynomials $r,t \in S(\mathfrak{g})^G$ mapping to $\overline r,\overline{t}\in \mathbb{R}[\overline{\mathrm{Ad}(G)\mathfrak{m}}]^G$ such that $\phi=\overline r/\overline{t}$ in $\mathbb{R}(\overline{\mathrm{Ad}(G)\mathfrak{m}})$. Pulling back along $\bar{P}$ yields \begin{align*}
f=\phi\circ\bar{P}=\frac{\overline r\circ\bar{P}}{\overline{t}\circ\bar{P}} =\frac{r\circ\bar{P}}{t\circ\bar{P}}.
\end{align*} Since $r,t$ are $\mathrm{Ad}(G)$-invariant on $\mathfrak{g}$, for all $(g,u)\in X$, we have \begin{align*}
(r\circ\bar{P})((g,u))=r(\mathrm{Ad}(g)u)=r(u)=(r\vert_{\mathfrak{m}})(u),
\end{align*} and similarly for $t$. Hence, $r\circ\bar{P},\,t\circ\bar{P}$ lies both in $\mathfrak{F}_1^{\mathrm{poly}}$ and in $\mathfrak{F}_2^{\mathrm{poly}}$, i.e., in $R_0$. Therefore, $f\in \mathrm{Frac} R_0=K$, proving $L\cap M\subset K$.

In conclusion, $K = L \cap M$ as required.
\end{proof}

Together with Lemma \ref{lem:invalgclose} and Lemma \ref{lem:LcapM}, we now upgrade the field intersection argument to the required linear disjointness statement by using the fixed field structure of the left $G$-action.

\begin{lemma}\cite[Ch. VIII,  Section 3]{LangAlgebra}
\label{lem:artinindependence}
Let a group $\Gamma$ act by automorphisms on a field $E$, and let $E_0:=E^\Gamma$. If $x_1,\ldots,x_N\in E$ are $E_0$-linearly independent, then there exist $\gamma_1,\ldots,\gamma_N\in\Gamma$ such that
\begin{align*} \det\big(\gamma_i(x_j)\big)_{1\leq i,j\leq N}\neq 0 . \end{align*}
\end{lemma}

\begin{proof}
We show the statement by contradiction. Assume that for every choice of $\gamma_1,\ldots,\gamma_N \in \Gamma$, the determinant in the statement vanishes. Then the vectors $$(\gamma(x_1),\ldots,\gamma(x_N))\in E^N = \underbrace{E \oplus \ldots \oplus E}_{N\text{ copies}}$$ span a proper $E$-subspace. Hence, there would be a non-trivial relation \begin{align} 
\sum_{j=1}^N a_j\gamma(x_j)=0\quad\text{for all }\gamma\in\Gamma,\qquad a_j\in E .  \label{eq:artinrelation}
\end{align} Choose such a relation with the smallest possible number of non-zero coefficients. Thus, there is a subset $J\subset\{1,\ldots,N\}$, with $J\neq\varnothing$, such that $a_j\neq 0$ exactly for $j\in J$, and $|J|$ is minimal among all non-trivial relations of the form \eqref{eq:artinrelation} holding for every $\gamma\in\Gamma$. Pick some $j_0\in J$ and divide the relation by $a_{j_0}$, so that after renumbering we may assume $j_0=1$ and \begin{align*}
 a_1=1,\qquad \sum_{j\in J} a_j\gamma(x_j)=0\quad\text{for all }\gamma\in\Gamma.
\end{align*} Now fix $\tau\in\Gamma$ and apply $\tau$ to the displayed relation. Since $\tau(\gamma(x_j))=(\tau\gamma)(x_j)$, we obtain \begin{align}
 \sum_{j\in J} \tau(a_j)(\tau\gamma)(x_j)=0\quad\text{for all }\gamma\in\Gamma. \label{eq:lineardecompos}
\end{align} Since left multiplication $\Gamma \to \Gamma$, $\gamma \mapsto \tau\gamma$, is a bijection, we may take $\gamma':= \tau \gamma$. Then \eqref{eq:lineardecompos} is equivalently \begin{align}
 \sum_{j\in J} \tau(a_j)\gamma(x_j)=0\quad\text{for all }\gamma \in \Gamma. \label{eq:obersvation}
\end{align} Subtracting \eqref{eq:obersvation} from \eqref{eq:lineardecompos} yields \begin{align*}
 \sum_{j\in J} (\tau(a_j)-a_j)\gamma(x_j)=0\quad\text{for all }\gamma\in\Gamma.
\end{align*} The coefficient of $\gamma(x_1)$ is $\tau(a_1)-a_1=0$, so this is a relation supported on (a subset of) $J\setminus\{1\}$. By the minimality of $|J|$, all remaining coefficients must vanish: $\tau(a_j)=a_j$ for every $j\in J$. Since $\tau\in\Gamma$ was arbitrary, and each $a_j$ is fixed by $\Gamma$, we have $a_j \in E^\Gamma=E_0$. Finally, taking $\gamma=\mathrm{id}$ gives a non-trivial $E_0$-linear relation among $x_1,\ldots,x_N$, contradicting their $E_0$-linear independence.
\end{proof}

 We now apply Lemma \ref{lem:artinindependence} to the left $G$-action on the field $L$, whose fixed field will be identified below with $K=L^G$. The lemma produces, from any $K$-linearly independent family $\ell_1,\ldots,\ell_N\in L$, group elements $g_1,\ldots,g_N\in G$ such that the matrix $(g_i(\ell_j))$ has a non-zero determinant. Since $G$ fixes $M$ pointwise, this non-degeneracy allows us to test $M$-linear relations among the $\ell_j$ by applying the $g_i$, forcing all coefficients to vanish. This yields the desired linear disjointness of $L$ and $M$ over $K$.

\begin{proposition} 
\label{prop:algdisjoint} \cite[Chapter VIII, Section 3 and 4]{LangAlgebra}
Let $\Omega = \mathrm{Frac} (\mathcal{O}(T^*M))$. With $L,M,K$ as in Lemma \ref{lem:LcapM}, the fields $M$ and $L$ are linearly disjoint over $K$ and
\begin{align*}
\mathrm{trdeg}_K(LM)=\mathrm{trdeg}_K L+\mathrm{trdeg}_K M.
\end{align*} In particular, $L \otimes_K M$ is an integral domain.
\end{proposition}



\begin{proof}
The left $G$-action on $T^*M\simeq G\times_A(\mathfrak{m}-\varepsilon W)$ is given by
\begin{align*} h\cdot[g,\xi]=[hg,\xi] . \end{align*}
Since $P(hg,\xi) = \mathrm{Ad}(h)P(g,\xi)$, the field $L = \mathrm{Frac}\,\mathfrak{F}_1^{\mathrm{poly}}$ is $G$-stable. Since $\pi_{\mathfrak{m}}(hg,\xi)=\xi=\pi_{\mathfrak{m}}(g,\xi)$, the field $M = \mathrm{Frac}\,\mathfrak{F}_2^{\mathrm{poly}}$ is fixed pointwise by $G$.

Moreover, the fixed field of the $G$-action on $L$ is exactly $K$. Indeed, $L=\mathbb{R}(Z_P)$ with $Z_P = \mathrm{Ad}(G)(\mathfrak{m}-\varepsilon W)$, and Proposition \ref{prop:inter} together with Remark \ref{re:poisscenter} identify $R_0$ with $\mathbb{R}[Z_P]^G$. Therefore, \begin{align*}
L^G = \mathbb{R}(Z_P)^G = \mathrm{Frac}\,\mathbb{R}[Z_P]^G = \mathrm{Frac}\,R_0=K . 
\end{align*} Now let $\ell_1,\ldots,\ell_N\in L$ be $K$-linearly independent. By Lemma \ref{lem:artinindependence}, there exist $g_1,\ldots,g_N\in G$ such that
\begin{align*} \det\big(g_i(\ell_j)\big)_{1\leq i,j\leq N}\neq 0 . \end{align*}
Suppose that $\sum_{j=1}^N m_j\ell_j=0$ with $m_j\in M$. Applying $g_i$ to this relation and using that $G$ fixes $M$ pointwise (so $g_i(m_j) = m_j$) gives
\begin{align*}
 \sum_{j = 1}^N m_j g_i(\ell_j) = 0\quad\text{for }i = 1,\ldots,N .
\end{align*}
In other words, if we set $A:=(g_i(\ell_j))_{1\leq i,j\leq N}\in M_N(L)\subset M_N(\Omega)$ and $m:= (m_1,\ldots,m_N)^t\in M^N\subset \Omega^N$, then the above equations read
\begin{align*}
 A\,m = 0\quad\text{in }\Omega^N.
\end{align*}
By construction, $\det(A)\neq 0$ in $L$ (hence also in $\Omega$), so $A$ is invertible over $\Omega$ (e.g., by the adjugate formula $A^{-1}=\mathrm{adj}(A)/\det(A)$). Multiplying the equation $A m = 0$ by $A^{-1}$ yields $m=0$, i.e. $m_1=\cdots=m_N=0$.

Therefore, any $K$-linearly independent family in $L$ is also $M$-linearly independent. Equivalently, the natural map $L\otimes_K M\to \Omega$ is injective, which is the definition of linear disjointness of $L$ and $M$ over $K$. The transcendence-degree identity then follows from standard properties of linearly disjoint field extensions.
\end{proof}

\begin{corollary}
\label{cor:deltaalgzero}
For the multiplication map $\mu:\widehat{\mathcal{A}} = \mathfrak{F}_1^{\mathrm{poly}}\otimes_{R_0}\mathfrak{F}_2^{\mathrm{poly}}\to\mathcal{O}(T^*M)$, we have  $\ker\mu = 0$ (and in particular $(R_0\setminus\{0\})^{-1}\ker\mu = 0$). Consequently, $\delta_{\mathrm{alg}} = 0$.
\end{corollary}

\begin{proof}
By Corollary \ref{cor:norela} applied to Proposition \ref{prop:algdisjoint}, we already know that $(R_0\setminus\{0\})^{-1}\ker\mu = 0$. That is, $\ker\mu$ is an $R_0$ torsion ideal in $\widehat{\mathcal{A}}$. Since $\widehat{\mathcal{A}}$ is an integral domain, multiplication by any $0\neq r\in R_0\subset \widehat{\mathcal{A}}$ is injective. Hence, an $R_0$ torsion element of $\widehat{\mathcal{A}}$ must be zero. Applying this to elements of $\ker\mu$, we conclude that $\ker\mu=0$.

Finally, the transcendence degree identity in Proposition \ref{prop:algdisjoint} gives $\delta_{\mathrm{alg}}=0$.
\end{proof}

   \bibliographystyle{unsrt}
\bibliography{bibliography.bib}
\end{document}